\pdfoutput=1
\documentclass[12pt,english]{article}
\usepackage[T1]{fontenc}
\usepackage[latin9]{inputenc}
\usepackage{verbatim}
\usepackage{float}
\usepackage{mathrsfs}
\usepackage{amsmath}
\usepackage{amsthm}
\usepackage{amssymb}
\usepackage{geometry}
\usepackage{setspace}
\usepackage{esint}
\makeatletter
\theoremstyle{definition}
\newtheorem{defn}{\protect\definitionname}
\theoremstyle{plain}
\newtheorem{prop}{\protect\propositionname}
\theoremstyle{plain}
\newtheorem{thm}{\protect\theoremname}
\theoremstyle{plain}
\newtheorem{lem}{\protect\lemmaname}
\ifx\proof\undefined\
  \newenvironment{proof}[1][\proofname]{\par
    \normalfont\topsep6\p@\@plus6\p@\relax
    \trivlist
    \itemindent\parindent
    \item[\hskip\labelsep
          \scshape
      #1]\ignorespaces
  }{%
    \endtrivlist\@endpefalse
  }
  \providecommand{\proofname}{Proof}
\fi
\theoremstyle{plain}
\newtheorem{cor}{\protect\corollaryname}

\usepackage[all]{nowidow}
\usepackage{natbib}

\usepackage{tikz}
\usetikzlibrary{calc}
\usetikzlibrary{external}
\usepackage{pgf}
\usepackage{amsmath,amsfonts,amssymb,amsthm}
\usetikzlibrary{automata, positioning, arrows}
\usepackage{float}
\usepackage{tikz-cd}
\usetikzlibrary{patterns}

\makeatother

\usepackage{babel}
\providecommand{\corollaryname}{Corollary}
\providecommand{\definitionname}{Definition}
\providecommand{\lemmaname}{Lemma}
\providecommand{\propositionname}{Proposition}
\providecommand{\theoremname}{Theorem}

\begin{document}
\author{Dirk Bergemann\thanks{Department of Economics, Yale University, dirk.bergemann@yale.edu}\quad{}Stephen
Morris\thanks{Department of Economics, Massachusetts Institute of Technology, semorris@mit.edu}\quad{}Rafael
Veiel\thanks{Department of Economics, Massachusetts Institute of Technology, veiel@mit.edu}}
\title{A Strategic Topology on Information Structures\thanks{We acknowledge financial support through NSF Grant SES 2049744. We
are grateful for comments of seminar participants at the Cowles Foundation
Economic Theory conference, the Warwick economic theory conference,
the Panorama of Mathematics II at the Hausdorff Center.}}
\maketitle
\begin{center}
\textbf{Abstract}
\par\end{center}

Two information structures are said to be close if, with high probability,
there is approximate common knowledge that interim beliefs are close
under the two information structures. We define an ``almost common
knowledge topology'' reflecting this notion of closeness. We show
that it is the coarsest topology generating continuity of equilibrium
outcomes. An information structure is said to be simple if each player
has a finite set of types and each type has a distinct first-order
belief about payoff states. We show that simple information structures
are dense in the almost common knowledge topology and thus it is without
loss to restrict attention to simple information structures in information
design problems. 
\begin{flushleft}
\par\end{flushleft}

\begin{flushleft}
\par\end{flushleft}

\section{Introduction\label{sec:Introduction}}

\subsection{Motivation and Results}

Players make choices in a game where payoffs depend on some unknown
state of the world. Optimal strategic behavior will then depend on
players' beliefs about the states, their beliefs about others' beliefs,
and so on. The (common prior) information structure of the players
is then a probability distribution over the players' beliefs and higher-order
beliefs about the state. We are interested in how the information
structure impacts equilibrium outcomes in such a game. It is known
that equilibrium outcomes are sensitive to the infinite tails of higher-order
beliefs (\cite{rubinstein89email} and \cite{cava93}). Our main result
is a characterization of the coarsest topology on information structures
generating continuity of equilibrium outcomes. 

Following \cite{monderer1989approximating}, an event is said to be
\emph{common $p\text{-belief}$} if everyone believes it with probability
at least $p\text{, everybody believes with probability at least \ensuremath{p}}$
that everyone believes it with probability at least $p\text{,}$ and
so on. An event is said to be \emph{approximate common knowledge}
if it is common $p\text{-}$belief with $p\text{ close to \ensuremath{1\text{. }}}$Now
two information structures are close in the almost common knowledge
topology if, with high ex ante probability, there is approximate common
knowledge that their higher-order beliefs are close (in the product
topology). Our main result (Theorem \ref{thm: main} in Section \ref{thm: main})
establishes that the almost common knowledge topology is the coarsest
topology generating continuity of equilibrium outcomes.

Our definition of information structures excludes any (payoff-irrelevant)
correlating devices that players might have access to. In the language
of \cite{mertens1985formulation}, we restrict attention to \emph{non-redundant}
information structures. However, our main result allows the correlating
device to appear within the equilibrium solution concept: we study
\emph{belief-invariant Bayes correlated equilibria }where players'
action choices can be correlated but only when the correlation does
not alter players' beliefs and higher-order beliefs about the state.
While we think our notions of information structures and equilibrium
are the most natural for our main exercise, we also spell out (in
Section \ref{sec:Bayes-Nash-Equilibrium}) how our results change
if we allow more general information structures with redundancies
and Bayes Nash Equilibrium (which rules out correlation in the solution
concept). 

We say that an information structure is simple if its support is finite
and each type of a player has a distinct first-order belief. We show
that simple information structures are dense. One implication is that
it is without loss to focus on simple information structures in solving
such information design problems. 

\subsection{Alternative Topologies and Related Literature}

\cite{monderer1996proximity} and \cite{kajii1998payoff}, building
on \cite{monderer1989approximating}, identified topologies on information
structures that generated continuity of equilibrium outcomes. And
the topologies identified both have the same flavor as the almost
common knowledge topology. The key difference is that both papers
fix (different) sets of information structures with the restriction
that each profile of types gives rise to a distinct state, and is
not associated with higher-order beliefs. This makes the topology
hard to interpret. It also makes both directions of the proof easier
than in our problem. We describe, following our main result (Theorem
\ref{thm: main} in Section \ref{thm: main}), which steps of our
proof relate to this early work and which are novel to this paper.
We discuss the join measurability condition implicit in \cite{monderer1996proximity}
and \cite{kajii1998payoff} in Section \ref{subsec:Countable-Type-Spaces}.
Other differences are that the earlier work focuses on Bayes Nash
equilibrium and restricts attention to countable information structures
to ensure existence; we allow uncountable information structures (sets
of the universal type space) and we ensure equilibrium existence by
incorporating redundancies in our solution concept. But we also show
how our results apply to Bayes Nash equilibrium and general information
structures in Section \ref{sec:Bayes-Nash-Equilibrium}. The earlier
work measures closeness of equilibrium outcomes by players' expected
payoffs in equilibrium; we note in Section \ref{subsec:Value-Based-Topology}
that our results would be unchanged if we used this approach. 

\cite{dekel2006topologies} defined an interim strategic topology
on hierarchies of beliefs under the solution concept of interim correlated
rationalizability (ICR). Two belief hierarchies were said to be close
in the interim strategic topology if, in any game, an action that
was ICR at one hierarchy was approximately ICR at the other hierarchy.
It is well-known (e.g., from the work of Rubinstein (1989), Carlsson
and van Damme (1993) and Weinstein and Yildiz (2007)) that closeness
in the product topology is not sufficient for close strategic behavior.
\cite{chen2017characterizing} characterized the interim strategic
topology of \cite{dekel2006topologies} and showed that it requires
closeness of beliefs about some tail events (``frames''). However,
the connection to approximate common knowledge has been unclear. We
show that two information structures are close in our almost common
knowledge topology if and only if there is a high ex ante probability
that hierarchies are close in the interim strategic topology. 

\section{Model}

This section will introduce the model used in our main analysis, with
four main ingredients. 
\begin{enumerate}
\item We will hold fixed a finite set of players, a finite set of (payoff-relevant)
\emph{states} and a probability distribution over the states. 
\item We define a \emph{base game} to consist of players' possible actions
and their payoff functions; i.e., how each player's payoff depends
on the action profile chosen and the state. 
\item An \emph{information structure} will consist of a probability distribution
over the state and players' beliefs and higher-order beliefs about
the state, with the appropriate marginal over states.
\item The equilibrium solution concept will be the \emph{belief-invariant
Bayes correlated equilibrium} (BIBCE), see Definition 8 in \cite{bemo16}.
This is a joint distribution over the information structure and the
players' actions such that players' actions are best responses and
measurable with respect to players' beliefs and higher-order beliefs
about the state, but allowing payoff-irrelevant correlation of actions. 
\end{enumerate}
Our model of an information structure is restrictive in that it rules
out players observing multiple signals giving rise to the same beliefs
and higher-order beliefs. Equivalently, it rules out what \cite{mertens1985formulation}
labelled \emph{redundant} types, i.e., players observing payoff-irrelevant
signals through which they can correlate their behavior. On the other
hand, our equilibrium solution concept allows players to observe payoff-irrelevant
correlating devices in choosing actions. Thus we have made a modelling
choice to put correlation devices in the solution concept rather than
the information structure. 

In order to relate our work to the literature and to applications,
we will later discuss (in Section \ref{sec:Bayes-Nash-Equilibrium})
how our results easily translate to a setting with general information
structures (allowing \emph{redundancies}) and the more relaxed solution
concept of Bayes Nash equilibrium, i.e., if we move correlation possibilities
from the solution concept to the information structure. 

\subsection{Setting and Base Game}

There is a finite set of players $I$ and a finite set of (payoff)
states $\Theta$. Throughout the paper, we will fix a prior $\mu\in\Delta(\Theta)$.
Without loss, we will maintain the assumption that $\mu$ has full
support. 

A \emph{base game} then describes players' actions and payoffs: thus
a base game is a tuple $\text{\ensuremath{\mathcal{G}}}=\left((A_{i})_{i\in I},(u_{i})_{i\in I}\right)$,
where $A_{i}$ is a finite set of actions for player $i$ and $u_{i}:A_{i}\times A_{-i}\times\Theta\to[-M,M]$
is a payoff function for player $i$, where $A_{-i}:=\prod_{j\neq i}A_{j}$
and $M>0$ is an exogenous payoff bound. 

\subsection{Information Structures}

We follow \cite{hars67} and \cite{mertens1985formulation} in identifying
signals or types with players' beliefs and higher-order beliefs, or
hierarchies of beliefs. We first formally define the set of hierarchies
of beliefs.
\begin{defn}
(Hierarchies of Beliefs) The profile of players' hierarchies of beliefs,
$(\mathcal{T}_{i})_{i\in I}$ is defined recursively as follows: For
every $i$, let $\mathcal{\mathcal{T}}_{i}^{0}:=\left\{ *\right\} $
be a singleton and let $\mathcal{\mathcal{T}}_{i}^{1}:=\Delta(\Theta)$.
Given profiles $(\mathcal{T}_{i}^{m-1})_{i\in I}$ for $m>1$, define
for every $i$,
\[
\mathcal{T}_{i}^{m}:=\left\{ \left(\left(\tau_{i}^{1},\dots,\tau_{i}^{m-1}\right),\tau_{i}^{m}\right)\in\mathcal{T}_{i}^{m-1}\times\Delta(\mathcal{T}_{-i}^{m-1}\times\Theta):\text{marg}_{\Theta\times\mathcal{T}_{-i}^{m-2}}(\tau_{i}^{m})=\tau_{i}^{m-1}\right\} .
\]
Let $\mathcal{T}_{i}$ denote the set of sequences $\tau_{i}=(\tau_{i}^{m})_{m}$
so that for every $\overline{m}\in\mathbb{\mathbb{N}}$, the truncated
sequence $(\tau_{i}^{m})_{m\leq\overline{m}}$ belongs to $\mathcal{T}_{i}^{\overline{{m}}}$.

For simplicity, we will follow \cite{mertens1985formulation} in imposing
the product topology on hierarchies $\mathcal{T}_{i}$. Let $d_{\Pi}$
be a metric on $\Omega:=\Theta\times\mathcal{T}$ inducing the product
topology on $\mathcal{T}:=\prod_{i\in I}\mathcal{T}_{i}$ and the
discrete topology on $\Theta$.\footnote{The product topology on $\mathcal{T}$ is the coarsest topology so
that projections $\text{proj}_{\mathcal{T}^{m}}\colon\mathcal{T}\to\mathcal{T}^{m}$
are continuous, where for every $m$, $\mathcal{T}^{m}$ is endowed
with the weak topology. \cite{dekel2006topologies} provide a metric
that induces the product topology on $\mathcal{T}$, which for any
discount factor $\eta\in(0,1)$ can be described by $d_{\Pi}((\theta,\tau),(\hat{\theta},\hat{\tau}))=\boldsymbol{1}_{\theta\neq\hat{\theta}}+\sum_{n=1}^{\infty}\eta^{n}\max_{i}d_{w}^{n}(\tau_{i}^{n},\hat{\tau}_{i}^{n})$,
where $d_{w}^{n}$ is a product metric inducing the weak topology
on $\mathcal{T}_{i}^{n}$.} However, we later discuss (Section \ref{subsec:Interim-based-Topology})
why the use of the product topology on hierarchies is just for convenience
and not important for our arguments. \cite{mertens1985formulation}
show that for every $(\tau=(\tau_{i})_{i\in I},\theta)\in\Omega$
and every $i\in I$ there is a unique belief $\tau_{i}^{*}\in\Delta(\mathcal{T}_{-i}\times\Theta)$
so that for all $m\in\mathbb{N}$, $\tau_{i}^{m}=\text{marg}_{\mathcal{T}_{-i}^{m-1}\times\Theta}\left(\tau_{i}^{*}\right)$,
where $\tau\mapsto\tau^{*}=(\tau_{i}^{*})_{i\in I}$ is a homeomorphism.
Let $\mathscr{B}$ denote the Borel sigma-algebra on $\Omega$. 
\end{defn}
We will refer to $\Omega$ as the ``universal state space'', with
typical element $\omega=\left(\tau,\theta\right)$, where $\tau\in\mathcal{T}$.
Now an information structure is just a probability distribution on
the universal state space that respects the prior on states and the
labelling of hierarchies. 

\begin{defn}
(Information Structure). An information structure $P$ is a Borel
probability measure on $\Omega$ that satisfies two conditions: (i)
{[}consistency{]} the marginal of $P$ on $\varTheta$ is $\mu$;
(ii) {[}labelling{]} for every player $i$, there is a version of
the conditional probability $P_{i}\colon\mathcal{T}_{i}\to\Delta(\Theta\times\mathcal{T}_{-i})$
of $P$ so that 
\[
\tau_{i}^{*}=P_{i}(\tau_{i}),\ P\text{-a.s.}
\]

We will write $\mathcal{P}\subseteq\Delta(\Omega)$ for the set of
(common prior) information structures. We will sometimes describe
these as \emph{non-redundant} information structures because we do
not allow multiple types or signals of a player giving rise to the
same hierarchies of beliefs. This is to contrast them with the more
general \emph{redundant} information structures discussed in Section
\ref{sec:Bayes-Nash-Equilibrium}.

An information structure is \emph{finite} if it has finite support.
A hierarchy $\tau_{i}\in\mathcal{T}_{i}$ is \emph{finite} if it is
in the support of a finite information structure. The set of finite
types in $\mathcal{T}_{i}$ is denoted by $\mathcal{T}_{i}^{0}$ and
the set of finite states $\Omega^{0}\subseteq\Omega$ is given by
\[
\Omega^{0}:=\left\{ (\tau,\theta)\in\Omega:\forall\ i,\ \tau_{i}^{*}\in\mathcal{T}_{i}^{0}\right\} .
\]
\end{defn}

\subsection{Solution Concept: \protect \\
Belief-Invariant Bayes Correlated Equilibrium}

Together, a base game $\text{\ensuremath{\mathcal{G}}}$ and an information
structure $P$ define a game of incomplete information $(\text{\ensuremath{\mathcal{G}}},P)$.
Now we define our main equilibrium solution concept. We will be allowing
players' action choices to be correlated. So players' action choices
will be described by a \emph{decision rule,} mapping states and hierarchies
to action profiles; thus for any incomplete information game $\left(\text{\ensuremath{\mathcal{G}}},P\right)$,
a decision rule is a measurable map $\sigma:\Theta\times\mathcal{T}\to\Delta(A)$,
where $A:=\prod_{i\in I}A_{i}$ and $\Delta(A)$ is endowed with the
Euclidean topology.

An information structure $P$ and decision rule $\sigma$ jointly
induce a measure $\sigma\circ P\in\Delta(A\times\Theta\times\mathcal{T})$
in the natural way. We will be interested in \emph{outcomes} specifying
a joint distribution over actions and states $\nu\in\Delta(A\times\Theta)$.
Decision rule $\sigma$ \emph{induces} outcome $\nu_{\sigma}$ if
$\nu_{\sigma}$ is the marginal of $\sigma\circ P$ on $A\times\Theta$.
For every player $i$, a decision rule $\sigma$ and hierarchy $\tau_{i}\in\mathcal{T}_{i}$
induce a belief $\sigma\circ\tau_{i}\in\Delta(A\times\Theta\times\mathcal{T}_{-i})$,
which for every measurable set $E\subseteq\Theta\times\mathcal{T}_{-i}$
and action profile $a\in A$ satisfies $\sigma\circ\tau_{i}(\left\{ a\right\} \times E)=\int_{E}\sigma(a|\theta,\tau_{-i},\tau_{i})\text{ d}\sigma\circ\tau_{i}^{*}$.%

Our notion of equilibrium will be (one version of) incomplete information
correlated equilibrium. Two properties will be key. 
\begin{defn}
A decision rule $\sigma$ is \emph{$\varepsilon$-obedient} if for
every player $i$, action $a_{i}$ and deviation $a_{i}'$, 
\[
\intop\sum_{a_{-i}\in A_{-i}}(u_{i}(a_{i},a_{-i},\theta)-u_{i}(a_{i}',a_{-i},\theta))\text{ d}\sigma\circ\tau_{i}(a_{i},a_{-i},\tau_{-i},\theta)>-\varepsilon,\ a.s.
\]
\end{defn}
Thus if we interpret the decision rule $\sigma:\Theta\times\mathcal{T}\to\Delta(A)$
as a recommendation rule for a mediator in the game, $\varepsilon$-obedience
requires that players will not lose more than $\varepsilon$ by following
the recommendation. 

Importantly, this is an interim version of $\varepsilon$-obedience:
we are requiring that the decision rule is almost always approximately
optimal.\footnote{An ex-ante notion of $\varepsilon$-obedience would require interim
$\varepsilon$-obedience only with probability $1-\varepsilon$, and
thus allows players to choose actions that are not $\varepsilon$
best responses with probability $1-\varepsilon$. The ex-ante notion
of $\varepsilon$-obedience is much more permissive: \cite{engl95lower}
has established the lower hemi-continuity of ex ante $\varepsilon$-Bayes
Nash equilibrium in games of incomplete information with respect to
weak convergence of priors on a fixed type space.}

\begin{defn}
A decision rule $\sigma$ is \emph{belief-invariant} if, for every
$a_{i}\in A_{i}$, the marginal probability $\sigma({a_{i}}\times A_{-i}|(\tau_{i},\tau_{-i}),\theta)=\sigma_{i}(a_{i}|\tau_{i})$
does not depend on $(\tau_{-i},\theta)$.\footnote{\cite{lehrer10signaling} and \cite{forges06revisited} introduced
the notion of belief-invariance to the study of incomplete information
correlated equilibrium: \cite{forges93five} (implicitly) and \cite{lehrer10signaling}
and \cite{forges06revisited} (explicitly) imposed the restriction
in their definitions of the belief-invariant Bayesian solution. This
latter solution imposes (implicitly or explicitly) join feasibility,
so that the decision rule depends on the prolile of players' signals.
\cite{liu15} showed that a subjective version of BIBCE is equivalent
to interim correlated rationalizability. A belief-invariant decision
rule can be interpreted as a payoff-irrelevant correlating device
(where player $i$ observes private signal $a_{i}$). Now BIBCE is
equivalent to Bayes Nash equilibria if players are allowed to observe
a correlated device.}
\end{defn}
\begin{defn}
A decision rule $\sigma$ is an $\varepsilon$-belief-invariant Bayes
correlated equilibrium ($\varepsilon${- BIBCE}) of $\left(\text{\ensuremath{\mathcal{G}}},P\right)$
if it satisfies $\varepsilon$-obedience and belief invariance.
\end{defn}
We will say that a decision rule is a BIBCE if it is a $0$-BIBCE. 

\section{Main Result\label{sec:Main-Result}}

In this section, we first define an ``approximate common knowledge''
topology on information structures and then show it is the coarsest
topology to generate continuity of BIBCE outcomes (our main result).

\subsection{The Approximate Common Knowledge Topology}

In words, we will say that two information structures are close if
each assigns high ex ante probability to the event that there is approximate
common knowledge that their interim beliefs are close. An event is
``approximate common knowledge'', if for some $p$ close to $1$,
everyone believes the event with probability at least $p$, everyone believes with probability at least $p$
that everyone believes the event with probability at least $p\text{,}$
and so on.... (\cite{monderer1989approximating}). The subtlety in
formalizing this definition is how we define the event that interim
beliefs are close. 

We first define $\varepsilon$-neighborhoods, the set of universal
states $\varepsilon$-close to a given universal state $\omega\in\Omega$. 
\begin{defn}
For every $\varepsilon>0$ and $\omega\in\Omega$, the $\varepsilon$-neighborhood
of $\omega$ is given by 
\[
\mathcal{N}_{\varepsilon}(\omega):=\left\{ \omega'\in\Omega:d_{\Pi}(\omega,\omega')<\varepsilon\right\} .
\]
For any $\varepsilon>0$, we define the $\varepsilon$-support of
$P\in\mathcal{P}$, to be the $\varepsilon$-neighborhoods of universal
states whose $\varepsilon$-neighborhoods are assigned positive probability
by $P$: 
\[
\text{supp}_{\varepsilon}(P):=\bigcup_{\omega\in\Omega:P(\mathcal{N}_{\varepsilon}(\omega))>0}\mathcal{N}_{\varepsilon}(\omega).
\]
Thus $\text{supp}_{\varepsilon}(P)$ is an $\varepsilon$-expansion
of the support of $P$. For any two priors $P,P'$, we define the
intersection of their $\varepsilon$-supports: 
\end{defn}
\begin{verse}
\[
\hat{T}_{\varepsilon}(P,P'):=\text{supp}_{\varepsilon}(P)\cap\text{supp}_{\varepsilon}(P').
\]
\end{verse}
\begin{figure}
\centering{}\begin{center}
\tikzset{every picture/.style={line width=0.75pt}} 
\begin{tikzpicture}[x=0.75pt,y=0.75pt,yscale=-1.2,xscale=1.2]
\draw  [fill={rgb, 255:red, 155; green, 155; blue, 155 }  ,fill opacity=0.78 ] (133,169.67) .. controls (133,148.13) and (150.46,130.67) .. (172,130.67) .. controls (193.54,130.67) and (211,148.13) .. (211,169.67) .. controls (211,191.21) and (193.54,208.67) .. (172,208.67) .. controls (150.46,208.67) and (133,191.21) .. (133,169.67) -- cycle ;
\draw  [fill={rgb, 255:red, 155; green, 155; blue, 155 }  ,fill opacity=0.78 ] (219.67,169.67) .. controls (219.67,148.13) and (237.13,130.67) .. (258.67,130.67) .. controls (280.21,130.67) and (297.67,148.13) .. (297.67,169.67) .. controls (297.67,191.21) and (280.21,208.67) .. (258.67,208.67) .. controls (237.13,208.67) and (219.67,191.21) .. (219.67,169.67) -- cycle ;
\draw   (99,169.67) .. controls (99,129.44) and (131.61,96.83) .. (171.83,96.83) .. controls (212.06,96.83) and (244.67,129.44) .. (244.67,169.67) .. controls (244.67,209.89) and (212.06,242.5) .. (171.83,242.5) .. controls (131.61,242.5) and (99,209.89) .. (99,169.67) -- cycle ;
\draw   (186.33,169.67) .. controls (186.33,129.44) and (218.94,96.83) .. (259.17,96.83) .. controls (299.39,96.83) and (332,129.44) .. (332,169.67) .. controls (332,209.89) and (299.39,242.5) .. (259.17,242.5) .. controls (218.94,242.5) and (186.33,209.89) .. (186.33,169.67) -- cycle ;
\draw  [fill={rgb, 255:red, 245; green, 166; blue, 35 }  ,fill opacity=0.69 ] (216,112) .. controls (216.67,114) and (244,130) .. (244.67,169.67) .. controls (245.33,209.33) and (216,227.33) .. (215.33,227.33) .. controls (214.67,227.33) and (186.67,208) .. (186.33,169.67) .. controls (186,131.33) and (215.33,110) .. (216,112) -- cycle ;
\draw    (214.67,74.67) -- (215.3,131.67) ; \draw [shift={(215.33,134.67)}, rotate = 269.36] [fill={rgb, 255:red, 0; green, 0; blue, 0 }  ][line width=0.08]  [draw opacity=0] (10.72,-5.15) -- (0,0) -- (10.72,5.15) -- (7.12,0) -- cycle    ;
\draw (138.33,163.57) node [anchor=north west][inner sep=0.75pt]  [font=\footnotesize]  {$\text{supp}( P')$};
\draw (148.33,250.07) node [anchor=north west][inner sep=0.75pt]  [font=\footnotesize]  {$\text{supp}_{\varepsilon }( P')$};
\draw (240.17,250.07) node [anchor=north west][inner sep=0.75pt]  [font=\footnotesize]  {$\text{supp}_{\varepsilon }( P)$};
\draw (249,163.57) node [anchor=north west][inner sep=0.75pt]  [font=\footnotesize]  {$\text{supp}( P)$};
\draw (189,50.73) node [anchor=north west][inner sep=0.75pt]  [font=\footnotesize]  {$\hat{T}_{\varepsilon }( P,P')$};
\end{tikzpicture} 
\par\end{center}\caption{Sets $\hat{T}_{\varepsilon}(P,P')$ and $\varepsilon$-supports of
$P$ and $P'$.}
\end{figure}
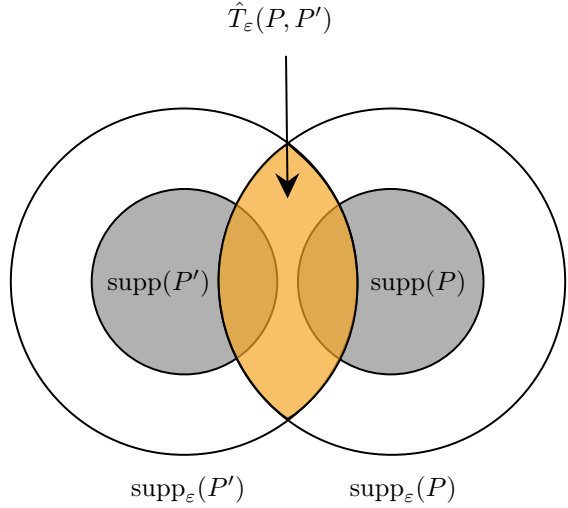

This is the set of universal states where interim beliefs are close
to interim beliefs in the support of both priors. Following \cite{monderer1989approximating},
for every probability $p\in[0,1]$ and event $E\in\mathscr{B}$, we
define $B^{p}(E)$ to be the set of universal states where all players
assign probability at least $p$ to the event $E$. Thus\textbf{
\[
B^{p}(E):=\left\{ (\tau,\theta)\in E:\forall i,\tau_{i}^{*}(E_{-i})\geq p\right\} 
\]
}where $E_{-i}:=\text{proj}_{\Theta\times\mathcal{T}_{-i}}(E)$. For
every $m\in\mathbb{N},$ let $\left[B^{p}\right]^{m}(E):=B^{p}\circ\cdots\circ B^{p}(E)$
denote the $m$-fold application of $B^{p}$; $\left[B^{p}\right]^{m}(E):=B^{p}\circ\cdots\circ B^{p}(E)$
is the set of universal states where all players assign probability
at least $p$ to all players assigning at least probability $p$....
($m$ times) to event $E$ being true. Now an event is said to be
\emph{common $p$-belief} if this is true for all $m$.
\begin{defn}
(Common $p$-Belief) For any event $E\in\mathscr{B}$, the event that
$E$ is \emph{common $p$-belief}, $C^{p}(E)$, is defined as
\[
C^{p}(E):=\bigcap_{m\in\mathbb{N}}\left[B^{p}\right]^{m}(E).
\]
\end{defn}
An event is said to be approximate common knowledge if it is common
$p$-belief for $p$ close to $1\text{. }$We can now define the approximate
common knowledge (ACK) distance\footnote{A distance is a map $d:\mathcal{P}\times\mathcal{P}\to[0,\infty)$
which is zero on the diagonal.} between any two priors, $P,P'\in\mathcal{P}$. 
\begin{defn}
\label{def:ACK}(Approximate Common Knowledge Distance) For every
$P,P'\in\mathcal{P}$, let 

\[
d^{ACK}\left(P,P'\right):=\inf\left\{ \varepsilon\geq0:\begin{array}{c}
P\left(C^{1-\varepsilon}\left(\hat{T}_{\varepsilon}(P,P')\right)\right)\geq1-\varepsilon\\
P'\left(C^{1-\varepsilon}\left(\hat{T}_{\varepsilon}(P,P')\right)\right)\geq1-\varepsilon
\end{array}\right\} .
\]
\end{defn}
Thus two priors are close if, under both priors, there is high probability
that there is approximate common knowledge that their interim beliefs
are close. 
\begin{defn}
(Approximate Common Knowledge Topology) The approximate common knowledge
(ACK) topology is the topology generated by open sets $\{ P'\in\mathcal{P}:d^{ACK}\left(P,P'\right)<\varepsilon\}_{P\in\mathcal{P}}$.
\end{defn}
While $d^{ACK}$$ $ fails the triangle inequality, the topology is
metrizable. 
\begin{prop}
\label{prop:The-ACK-topology}The ACK topology is metrizable. 
\end{prop}
The proof (in the Appendix) constructs a metric that generates the
same topology. Because the topology is metrizable, it is characterized
by the convergent sequences.
\begin{defn}
A sequence $(P^{k})_{k}$ in $\mathcal{P}$ ACK-converges to $P$
if and only if $d^{ACK}\left(P^{k},P\right)\rightarrow0$.
\end{defn}
It is useful to consider what the topology looks like in some special
cases. If there is only one payoff relevant state (i.e., $|\Theta|=1$),
then there is a unique universal state, where there is common knowledge
of the state. If there is only one player, then an information structure
is given by a distribution over common first order beliefs about $\Theta$,
and the topology reduces to the weak topology on $\Delta\left(\Delta\left(\Theta\right)\right)$.
Similarly, if there are many players but types are perfectly correlated.
If types are independent, then a canonical information structure is
given by a distribution in $\Delta(\Delta(\Theta)^{I})$, and the
topology again reduces to the weak topology. If we restrict to information
structures with a fixed, finite supports, then an information structure
is a probability distribution on a finite set and the topology reduces
to the weak topology. So in order for the ACK topology to be interesting,
there must be at least two states, at least two players, an unbounded
number of types and neither independence or perfect correlation. 

\subsection{Continuity of Equilibrium Outcomes}

We now define continuity of equilibrium outcomes (for BIBCE). For
every $\varepsilon>0$ and $(\mathcal{G},P)$, let $\mathcal{B}^{\varepsilon}(\mathcal{G},P)$
denote the collection of all $\varepsilon$-BIBCE under $(\mathcal{G},P)$.
Define the set of outcomes (i.e., elements of $\Delta\left(A\times\Theta\right)$)
that are induced by a $\varepsilon$-BIBCE of $\left(\text{\ensuremath{\mathcal{G}}},P\right)$
\[
\mathcal{O}^{\varepsilon}(\mathcal{G},P):=\left\{ \nu_{\sigma}:\sigma\in\mathcal{B}^{\varepsilon}(\mathcal{G},P)\right\} .
\]
We write $\mathcal{O}_{\varepsilon}\left(\text{\ensuremath{\mathcal{G}}},P\right)$
for the set of outcomes that are within $\varepsilon$ of outcomes
induced by an $\varepsilon$-BIBCE of $\left(\text{\ensuremath{\mathcal{G}}},P\right)$,
so 
\[
\mathcal{O_{\varepsilon}}\left(\text{\ensuremath{\mathcal{G}}},P\right):=\left\{ v\in\Delta\left(A\times\Theta\right):\exists\ \sigma\in\mathcal{B^{\varepsilon}}(\mathcal{G},P)\text{ s.t. }||v-v_{\sigma}||_{2}<\varepsilon\right\} ,
\]
where $||\cdot||_{2}$ is the Euclidean norm on outcomes. Now we will
say that two information structures are strategically close for base
game $\mathcal{G}$ if the sets of BIBCE outcomes are close. 
\begin{defn}
(Strategic Distance) For every $P,P'\in\mathcal{P}$ and base game
$\mathcal{G}$, let

\[
d^{*}\left(P,P'|\mathcal{G}\right):=\inf\left\{ \varepsilon\geq0:\begin{array}{c}
\mathcal{O}^{0}(\mathcal{G},P)\subseteq\mathcal{O}_{\varepsilon}(\mathcal{G},P')\\
\mathcal{O}^{0}(\mathcal{G},P')\subseteq\mathcal{O}_{\varepsilon}(\mathcal{G},P)
\end{array}\right\} .
\]
\end{defn}
A topology on $\mathcal{P}$ generates continuity of equilibrium outcomes
if for every $\varepsilon>0$, every base game $\mathcal{G}$ and
every $P\in\mathcal{P}$, the set $\left\{ P'\in\mathcal{P}:d^{*}\left(P,P'|\mathcal{G}\right)<\varepsilon\right\} $
is open. 

\subsection{Main Result}

The following is our main result. 

\begin{thm}
\label{thm: main}The approximate common knowledge topology is the
coarsest topology on $\mathcal{P}$ that gives rise to continuity
of equilibrium outcomes.
\end{thm}
We will first sketch the main ideas in the proof, distinguishing between
steps that are inherited from earlier work, and how, and which steps
are and must be novel. To establish that approximate common knowledge
convergence implies strategic convergence (i.e. that the ACK topology
generates continuity of strategic outcomes), one must show that if
two information structures are close enough in the approximate common
knowledge topology, then in any game, the set of $\varepsilon$-BIBCE
outcomes will be close. This builds on the arguments in \cite{monderer1996proximity}
and \cite{kajii1998payoff} (which in turn build on an argument in
\cite{monderer1989approximating}). These two papers fix an equilibrium
strategy profile under one information structure and show that there
is an approximate equilibrium in any nearby information structure
where the strategy profile is unchanged on the event where there is
approximate common knowledge that interim beliefs are close, but can
vary arbitrarily elsewhere. This proof strategy relies on a well-defined
notion of holding the strategy profile fixed on the approximate common
knowledge event. In our context, there is no well-defined notion of
holding the strategy profile fixed on the approximate common knowledge
event. In particular, if the two information structures are minimal
(they do not have common knowledge subsets), the supports of distinct
information structures will be disjoint, as illustrated in Figure
\ref{fig:Sets-,-and}. So instead we will continuously extend the
equilibrium decision rule $\sigma$ under one information structure
$P$ to its $\varepsilon$-support $\text{supp}_{\varepsilon}(P)$
and thus to the event:
\[
\hat{T}_{\varepsilon}(P,P'):=\text{supp}_{\varepsilon}(P)\cap\text{supp}_{\varepsilon}(P'),
\]
and thus the approximate common knowledge event:
\[
C^{1-\varepsilon}\left(\hat{T}_{\varepsilon}(P,P')\right)\subseteq\hat{T}_{\varepsilon}(P,P').
\]

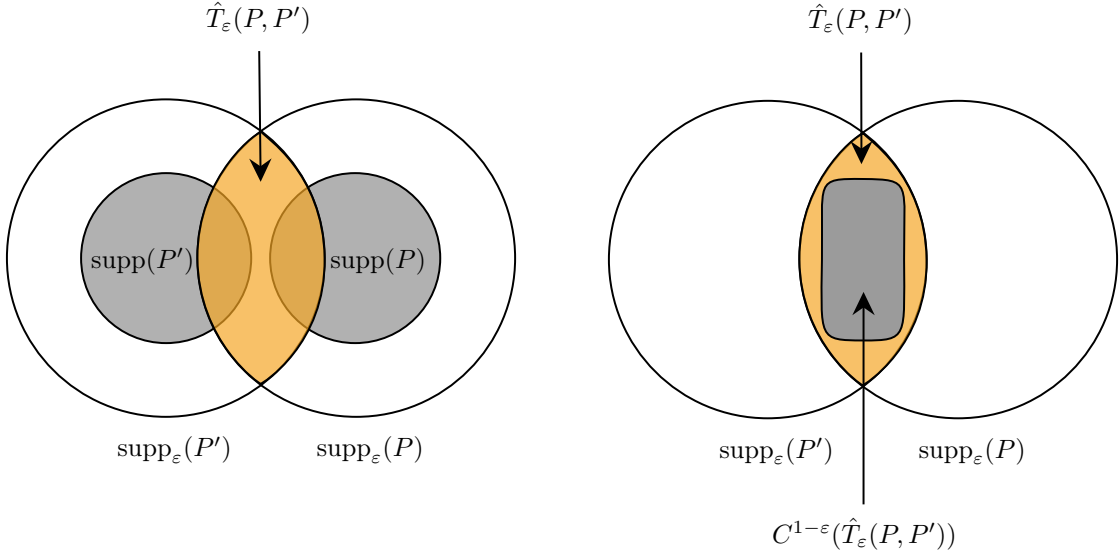
\begin{figure}
\centering{}\begin{center}
\begin{center}
\tikzset{every picture/.style={line width=0.75pt}} 
\begin{tikzpicture}[x=0.75pt,y=0.75pt,yscale=-1.1,xscale=1.1]
\draw  [fill={rgb, 255:red, 155; green, 155; blue, 155 }  ,fill opacity=0.78 ] (113,149.67) .. controls (113,128.13) and (130.46,110.67) .. (152,110.67) .. controls (173.54,110.67) and (191,128.13) .. (191,149.67) .. controls (191,171.21) and (173.54,188.67) .. (152,188.67) .. controls (130.46,188.67) and (113,171.21) .. (113,149.67) -- cycle ;
\draw  [fill={rgb, 255:red, 155; green, 155; blue, 155 }  ,fill opacity=0.78 ] (199.67,149.67) .. controls (199.67,128.13) and (217.13,110.67) .. (238.67,110.67) .. controls (260.21,110.67) and (277.67,128.13) .. (277.67,149.67) .. controls (277.67,171.21) and (260.21,188.67) .. (238.67,188.67) .. controls (217.13,188.67) and (199.67,171.21) .. (199.67,149.67) -- cycle ;
\draw   (79,149.67) .. controls (79,109.44) and (111.61,76.83) .. (151.83,76.83) .. controls (192.06,76.83) and (224.67,109.44) .. (224.67,149.67) .. controls (224.67,189.89) and (192.06,222.5) .. (151.83,222.5) .. controls (111.61,222.5) and (79,189.89) .. (79,149.67) -- cycle ;
\draw   (166.33,149.67) .. controls (166.33,109.44) and (198.94,76.83) .. (239.17,76.83) .. controls (279.39,76.83) and (312,109.44) .. (312,149.67) .. controls (312,189.89) and (279.39,222.5) .. (239.17,222.5) .. controls (198.94,222.5) and (166.33,189.89) .. (166.33,149.67) -- cycle ;
\draw  [fill={rgb, 255:red, 245; green, 166; blue, 35 }  ,fill opacity=0.69 ] (196,92) .. controls (196.67,94) and (224,110) .. (224.67,149.67) .. controls (225.33,189.33) and (196,207.33) .. (195.33,207.33) .. controls (194.67,207.33) and (166.67,188) .. (166.33,149.67) .. controls (166,111.33) and (195.33,90) .. (196,92) -- cycle ;
\draw    (194.67,54.67) -- (195.3,111.67) ; \draw [shift={(195.33,114.67)}, rotate = 269.36] [fill={rgb, 255:red, 0; green, 0; blue, 0 }  ][line width=0.08]  [draw opacity=0] (10.72,-5.15) -- (0,0) -- (10.72,5.15) -- (7.12,0) -- cycle    ;
\draw   (355,150.33) .. controls (355,110.11) and (387.61,77.5) .. (427.83,77.5) .. controls (468.06,77.5) and (500.67,110.11) .. (500.67,150.33) .. controls (500.67,190.56) and (468.06,223.17) .. (427.83,223.17) .. controls (387.61,223.17) and (355,190.56) .. (355,150.33) -- cycle ;
\draw   (442.33,150.33) .. controls (442.33,110.11) and (474.94,77.5) .. (515.17,77.5) .. controls (555.39,77.5) and (588,110.11) .. (588,150.33) .. controls (588,190.56) and (555.39,223.17) .. (515.17,223.17) .. controls (474.94,223.17) and (442.33,190.56) .. (442.33,150.33) -- cycle ;
\draw  [fill={rgb, 255:red, 245; green, 166; blue, 35 }  ,fill opacity=0.69 ] (472,92.67) .. controls (472.67,94.67) and (500,110.67) .. (500.67,150.33) .. controls (501.33,190) and (472,208) .. (471.33,208) .. controls (470.67,208) and (442.67,188.67) .. (442.33,150.33) .. controls (442,112) and (471.33,90.67) .. (472,92.67) -- cycle ;
\draw    (470.67,55.33) -- (470.79,103.57) ; \draw [shift={(470.8,106.57)}, rotate = 269.85] [fill={rgb, 255:red, 0; green, 0; blue, 0 }  ][line width=0.08]  [draw opacity=0] (10.72,-5.15) -- (0,0) -- (10.72,5.15) -- (7.12,0) -- cycle    ;
\draw  [fill={rgb, 255:red, 155; green, 155; blue, 155 }  ,fill opacity=1 ] (471.6,113.37) .. controls (494.23,113.6) and (490.24,114.6) .. (490.34,150.36) .. controls (490.44,186.12) and (492.31,187.54) .. (471.56,187.46) .. controls (450.8,187.37) and (452.22,187.01) .. (452.68,150.36) .. controls (453.14,113.71) and (448.97,113.13) .. (471.6,113.37) -- cycle ;
\draw    (471.75,262.67) -- (471.75,168.17) ; \draw [shift={(471.75,165.17)}, rotate = 90] [fill={rgb, 255:red, 0; green, 0; blue, 0 }  ][line width=0.08]  [draw opacity=0] (10.72,-5.15) -- (0,0) -- (10.72,5.15) -- (7.12,0) -- cycle    ;
\draw (116.33,143.57) node [anchor=north west][inner sep=0.75pt]  [font=\footnotesize]  {$\text{supp}( P')$};
\draw (128.33,230.07) node [anchor=north west][inner sep=0.75pt]  [font=\footnotesize]  {$\text{supp}_{\varepsilon }( P')$};
\draw (219.67,230.07) node [anchor=north west][inner sep=0.75pt]  [font=\footnotesize]  {$\text{supp}_{\varepsilon }( P)$};
\draw (226,143.57) node [anchor=north west][inner sep=0.75pt]  [font=\footnotesize]  {$\text{supp}( P)$};
\draw (169,30.73) node [anchor=north west][inner sep=0.75pt]  [font=\footnotesize]  {$\hat{T}_{\varepsilon }( P,P')$};
\draw (404.33,230.73) node [anchor=north west][inner sep=0.75pt]  [font=\footnotesize]  {$\text{supp}_{\varepsilon }( P')$};  
\draw (495.67,230.73) node [anchor=north west][inner sep=0.75pt]  [font=\footnotesize]  {$\text{supp}_{\varepsilon }( P)$};
\draw (445,31.4) node [anchor=north west][inner sep=0.75pt]  [font=\footnotesize]  {$\hat{T}_{\varepsilon }( P,P')$};
\draw (428.5,267.4) node [anchor=north west][inner sep=0.75pt]  [font=\footnotesize]  {$C^{1-\varepsilon }(\hat{T}_{\varepsilon }( P,P'))$};
\end{tikzpicture}
\par\end{center}
\par\end{center}\caption{\label{fig:Sets-,-and}Sets $\hat{T}_{\varepsilon}(P,P')$,$C^{1-\varepsilon}(\hat{T}_{\varepsilon}(P,P'))$
and $\varepsilon$-supports of $P$ and $P'$.}
\end{figure}

This is one place where we are exploiting properties of the product
topology on types to show that the extended decision rule allows us
to find an approximate equilibrium under $P'$ where the continuous
extension of $\sigma$ on $C^{1-\varepsilon}\left(\hat{T}_{\varepsilon}(P,P')\right)$
is held fixed. This argument goes through with any refinement of the
product topology as we will discuss later.

To establish that strategic convergence implies approximate common
knowledge convergence, it is enough to show that, for any two information
structures that are not close in the approximate common knowledge
topology, one can construct a game where an equilibrium outcome under
one information structure is not close to any approximate equilibrium
under the other information structure. \cite{monderer1996proximity}
and \cite{kajii1998payoff} do this by showing that if two information
structures $P$ and $P'$ are not close, there is an ``infecting''
event $D^{0}$ under one of the information structures, say $P\text{, }$such
that there is no approximate common knowledge event on the complement
of $C^{1-\varepsilon}\left(\hat{T}_{\varepsilon}(P,P')\right)$. One
can then construct a binary coordination game in the spirit of the
email game of \cite{rubinstein89email} where there is a unique equilibrium
under $P$ because there is a dominant strategy on the ``infecting''
event, but there are multiple equilibria under $P'$. This strategy
is not available to us because we cannot assume that there is a dominant
strategy on the ``infecting'' event because payoffs must be measurable
with respect to the payoff states $\Theta$. 

Instead, we show that if two information structures are not close,
there exists an ``infecting'' event $D^{m}$ that is measurable
with respect to $m$-order beliefs. For $m$ large enough, this event
has the property that some types in the support of $P$ that are not
in $\hat{T}_{\varepsilon}(P,P')$ are closer to any element in $D^{m}$
than to any element in $\hat{T}_{\varepsilon}(P,P')$. Hence, for
$m$ large enough, some types can be excluded from $C^{1-\varepsilon}\left(\hat{T}_{\varepsilon}(P,P')\right)$
based on their beliefs on a $m$-order-measurable event. From event
$D_{\varepsilon}^{0}=D^{m}$ we obtain a cover $(D_{\varepsilon}^{n})_{n\in\mathbb{N}}$
of the complement of $C^{1-\varepsilon}\left(\hat{T}_{\varepsilon}(P,P')\right)$
recursively, where 
\[
D_{\varepsilon}^{n}=\text{supp}(P)\setminus B^{1-\varepsilon}(D_{\varepsilon}^{n-1}),
\]
for all $n\geq1$. Now we can construct a first game where players
are incentivized to announce their approximate $m$-th order beliefs
on a finite grid, (for $m$ large enough chosen as a function of $\varepsilon$)
and build an email game style binary action coordination game on top
of that. For the first component of the game, we can construct a game
with an iterated scoring rule with the property that it is $\varepsilon$-rationalizable
for players to truthfully report finite-order beliefs, from a finite
grid which are closest in $d_{\Pi}$. This scoring rule is also used
in \cite{dekel2006topologies} and \cite{gossner20value}. This game
alone cannot be used to induce different outcomes in $P$ and $P'$
since on the set $\hat{T}_{\varepsilon}(P,P')$, reporting the same
approximate finite-order beliefs is an $\varepsilon$-BIBCE for the
types in both priors. As we cannot a priori rule out the possibility
of $\hat{T}_{\varepsilon}(P,P')$ containing the support of $P'$
the scoring rule is not suitable for separating outcomes. Players
therefore also choose an additional action: Either action zero or
action one. No matter what additional action is chosen by the opponent,
action one is the unique, strict best-reply for players who reported
themselves to be in $D_{\varepsilon}^{0}$. All other types will match
action one if they believe with probability at least $\varepsilon$
that their opponent also chose action one. Hence iterative deletion
of dominated strategies implies that action one is played on $D_{\varepsilon}^{n}$.
Figure \ref{fig:Infection-argument.} illustrates the region $D_{\varepsilon}^{0}$
and the role of the scoring rule. Type profiles in the red region
will play action one and an action represented by a circle in the
right panel of the Figure. No type in the support of $P'$ will play
an action corresponding to a type in $D_{\varepsilon}^{0}$. Action
one will infect all type profiles in the orange shaded region who
are also in the support of $P$ but will not infect types in the support
of $P'$. 
\begin{center}
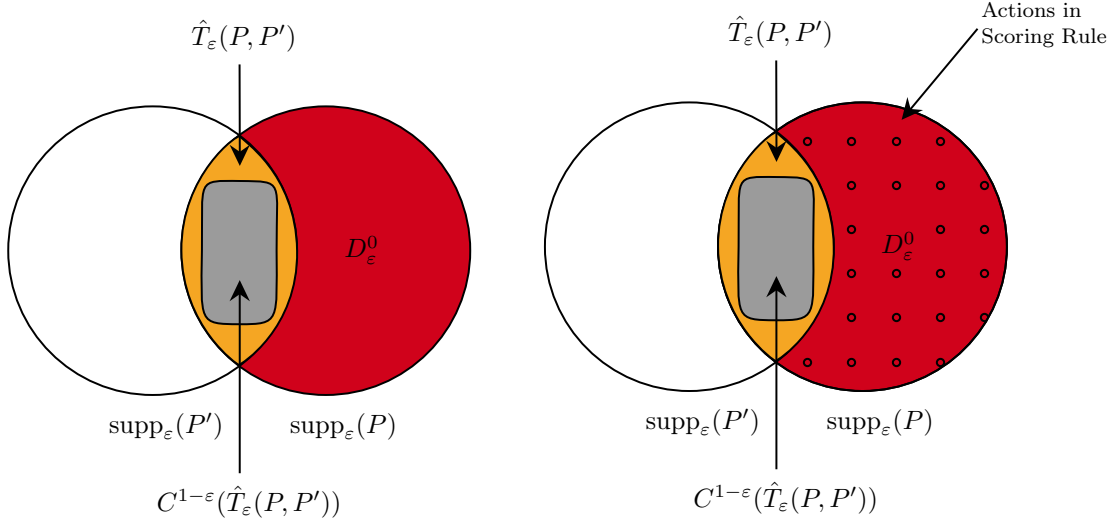
\begin{figure}
\begin{centering}
\begin{center}
  \tikzset{ pattern size/.store in=\mcSize,  pattern size = 5pt, pattern thickness/.store in=\mcThickness,  pattern thickness = 0.3pt, pattern radius/.store in=\mcRadius,  pattern radius = 1pt} \makeatletter \pgfutil@ifundefined{pgf@pattern@given by name@_5r8oj7tka}{ \makeatletter \pgfdeclarepatternformonly[\mcRadius,\mcThickness,\mcSize]{_5r8oj7tka} {\pgfpoint{-0.5*\mcSize}{-0.5*\mcSize}} {\pgfpoint{0.5*\mcSize}{0.5*\mcSize}} {\pgfpoint{\mcSize}{\mcSize}} { \pgfsetcolor{\tikz@pattern@color} \pgfsetlinewidth{\mcThickness} \pgfpathcircle\pgfpointorigin{\mcRadius} \pgfusepath{stroke} }} \makeatother \tikzset{every picture/.style={line width=0.75pt}} 
\begin{tikzpicture}[x=0.75pt,y=0.75pt,yscale=-1,xscale=1]
\draw  [fill={rgb, 255:red, 208; green, 2; blue, 27 }  ,fill opacity=1 ] (200.33,146.33) .. controls (200.33,106.11) and (232.94,73.5) .. (273.17,73.5) .. controls (313.39,73.5) and (346,106.11) .. (346,146.33) .. controls (346,186.56) and (313.39,219.17) .. (273.17,219.17) .. controls (232.94,219.17) and (200.33,186.56) .. (200.33,146.33) -- cycle ; \draw   (113,146.33) .. controls (113,106.11) and (145.61,73.5) .. (185.83,73.5) .. controls (226.06,73.5) and (258.67,106.11) .. (258.67,146.33) .. controls (258.67,186.56) and (226.06,219.17) .. (185.83,219.17) .. controls (145.61,219.17) and (113,186.56) .. (113,146.33) -- cycle ; \draw  [fill={rgb, 255:red, 245; green, 166; blue, 35 }  ,fill opacity=1 ] (230,88.67) .. controls (230.67,90.67) and (258,106.67) .. (258.67,146.33) .. controls (259.33,186) and (230,204) .. (229.33,204) .. controls (228.67,204) and (200.67,184.67) .. (200.33,146.33) .. controls (200,108) and (229.33,86.67) .. (230,88.67) -- cycle ; \draw    (229.67,52.33) -- (229.79,100.57) ; \draw [shift={(229.8,103.57)}, rotate = 269.85] [fill={rgb, 255:red, 0; green, 0; blue, 0 }  ][line width=0.08]  [draw opacity=0] (10.72,-5.15) -- (0,0) -- (10.72,5.15) -- (7.12,0) -- cycle    ; \draw  [fill={rgb, 255:red, 155; green, 155; blue, 155 }  ,fill opacity=1 ] (230,111.17) .. controls (252.63,111.41) and (248.24,110.6) .. (248.34,146.36) .. controls (248.44,182.12) and (250.31,183.54) .. (229.56,183.46) .. controls (208.8,183.37) and (210.22,183.01) .. (210.68,146.36) .. controls (211.14,109.71) and (207.37,110.94) .. (230,111.17) -- cycle ; \draw    (229.75,258.67) -- (229.75,164.17) ; \draw [shift={(229.75,161.17)}, rotate = 90] [fill={rgb, 255:red, 0; green, 0; blue, 0 }  ][line width=0.08]  [draw opacity=0] (10.72,-5.15) -- (0,0) -- (10.72,5.15) -- (7.12,0) -- cycle    ; \draw  [fill={rgb, 255:red, 208; green, 2; blue, 27 }  ,fill opacity=1 ] (471,144.33) .. controls (471,104.11) and (503.61,71.5) .. (543.83,71.5) .. controls (584.06,71.5) and (616.67,104.11) .. (616.67,144.33) .. controls (616.67,184.56) and (584.06,217.17) .. (543.83,217.17) .. controls (503.61,217.17) and (471,184.56) .. (471,144.33) -- cycle ; \draw  [pattern=_5r8oj7tka,pattern size=16.725pt,pattern thickness=0.75pt,pattern radius=1.35pt, pattern color={rgb, 255:red, 0; green, 0; blue, 0}] (471,144.33) .. controls (471,104.11) and (503.61,71.5) .. (543.83,71.5) .. controls (584.06,71.5) and (616.67,104.11) .. (616.67,144.33) .. controls (616.67,184.56) and (584.06,217.17) .. (543.83,217.17) .. controls (503.61,217.17) and (471,184.56) .. (471,144.33) -- cycle ; \draw   (383.67,144.33) .. controls (383.67,104.11) and (416.28,71.5) .. (456.5,71.5) .. controls (496.72,71.5) and (529.33,104.11) .. (529.33,144.33) .. controls (529.33,184.56) and (496.72,217.17) .. (456.5,217.17) .. controls (416.28,217.17) and (383.67,184.56) .. (383.67,144.33) -- cycle ; \draw  [fill={rgb, 255:red, 245; green, 166; blue, 35 }  ,fill opacity=1 ] (500.67,86.67) .. controls (501.33,88.67) and (528.67,104.67) .. (529.33,144.33) .. controls (530,184) and (500.67,202) .. (500,202) .. controls (499.33,202) and (471.33,182.67) .. (471,144.33) .. controls (470.67,106) and (500,84.67) .. (500.67,86.67) -- cycle ; \draw    (500.33,50.33) -- (500.46,98.57) ; \draw [shift={(500.47,101.57)}, rotate = 269.85] [fill={rgb, 255:red, 0; green, 0; blue, 0 }  ][line width=0.08]  [draw opacity=0] (10.72,-5.15) -- (0,0) -- (10.72,5.15) -- (7.12,0) -- cycle    ; \draw  [fill={rgb, 255:red, 155; green, 155; blue, 155 }  ,fill opacity=1 ] (500.67,109.17) .. controls (523.3,109.41) and (518.9,108.6) .. (519.01,144.36) .. controls (519.11,180.12) and (520.98,181.54) .. (500.22,181.46) .. controls (479.47,181.37) and (480.89,181.01) .. (481.35,144.36) .. controls (481.81,107.71) and (478.03,108.94) .. (500.67,109.17) -- cycle ; \draw    (500.42,256.67) -- (500.42,162.17) ; \draw [shift={(500.42,159.17)}, rotate = 90] [fill={rgb, 255:red, 0; green, 0; blue, 0 }  ][line width=0.08]  [draw opacity=0] (10.72,-5.15) -- (0,0) -- (10.72,5.15) -- (7.12,0) -- cycle    ; \draw    (601.86,34.33) -- (565.49,77.98) ; \draw [shift={(563.57,80.29)}, rotate = 309.8] [fill={rgb, 255:red, 0; green, 0; blue, 0 }  ][line width=0.08]  [draw opacity=0] (10.72,-5.15) -- (0,0) -- (10.72,5.15) -- (7.12,0) -- cycle    ;
\draw (162.33,226.73) node [anchor=north west][inner sep=0.75pt]  [font=\footnotesize]  {$\text{supp}_{\varepsilon }( P')$}; \draw (253.67,226.73) node [anchor=north west][inner sep=0.75pt]  [font=\footnotesize]  {$\text{supp}_{\varepsilon }( P)$}; \draw (204,28.4) node [anchor=north west][inner sep=0.75pt]  [font=\footnotesize]  {$\hat{T}_{\varepsilon }( P,P')$}; \draw (186.5,263.4) node [anchor=north west][inner sep=0.75pt]  [font=\footnotesize]  {$C^{1-\varepsilon }(\hat{T}_{\varepsilon }( P,P'))$}; \draw (281.17,137.4) node [anchor=north west][inner sep=0.75pt]  [font=\footnotesize,color={rgb, 255:red, 0; green, 0; blue, 0 }  ,opacity=1 ]  {$D_{\varepsilon }^0$}; \draw (433,224.73) node [anchor=north west][inner sep=0.75pt]  [font=\footnotesize]  {$\text{supp}_{\varepsilon }( P')$}; \draw (524.33,224.73) node [anchor=north west][inner sep=0.75pt]  [font=\footnotesize]  {$\text{supp}_{\varepsilon }( P)$}; \draw (474.67,26.4) node [anchor=north west][inner sep=0.75pt]  [font=\footnotesize]  {$\hat{T}_{\varepsilon }( P,P')$}; \draw (457.17,261.4) node [anchor=north west][inner sep=0.75pt]  [font=\footnotesize]  {$C^{1-\varepsilon }(\hat{T}_{\varepsilon }( P,P'))$}; \draw  [draw opacity=0][fill={rgb, 255:red, 208; green, 2; blue, 27 }  ,fill opacity=1 ]  (548.83,131) -- (572.83,131) -- (572.83,152) -- (548.83,152) -- cycle  ; \draw (551.83,135.4) node [anchor=north west][inner sep=0.75pt]  [font=\footnotesize,color={rgb, 255:red, 0; green, 0; blue, 0 }  ,opacity=1 ]  {$D_{\varepsilon }^0$}; \draw (602.43,19.76) node [anchor=north west][inner sep=0.75pt]  [font=\scriptsize] [align=left] {Actions in\\ Scoring Rule};
\end{tikzpicture} 
\par\end{center}
\par\end{centering}
\caption{\label{fig:Infection-argument.}Infection argument.}
\end{figure}
\par\end{center}

Hence, in every $\varepsilon$-BIBCE the types in orange and red shaded
regions who are in the support of $P$ will play action $1$. There
is a $\varepsilon$-BIBCE where all types play action $0$ under $P'$.
This establishes that the sets of outcomes of the priors are at least
$\varepsilon$ apart. In this second part of the proof, we exploit
properties of the product topology to ensure that we can cover the
red region $D_{\varepsilon}^{0}$ with a finite grid and there is
a finite game where players find it $\varepsilon$-optimal to report
the closest element in the grid. 

We now report two Lemmas before proceeding with the formal proof.
Let $\mathcal{N}_{\delta}(E):=\bigcup_{\hat{\omega}\in E}\left\{ \omega\in\Theta\times\mathcal{T}:d_{\Pi}(\hat{\omega},\omega)<\delta\right\} $
denote the $\delta$-ball around an event $E\in\mathscr{B}$. 
\begin{lem}
\label{claim:GridPart}For every $\varepsilon>0$ and any event $E\in\mathscr{B}$,
there is a finite event $G_{\varepsilon}\subseteq\Omega\setminus\mathcal{N}_{\varepsilon}(E)$
so that for every $\omega\in\Omega,$ 
\[
\min_{\omega_{g}\in G_{\varepsilon}}d_{\Pi}(\omega,\omega_{g})<\varepsilon\implies\omega\in\Omega\setminus E
\]
 and 
\[
\omega\in\Omega\setminus\mathcal{N}_{\varepsilon}(E)\implies\min_{\omega_{g}\in G_{\varepsilon}}d_{\Pi}(\omega,\omega_{g})<\varepsilon.
\]
\end{lem}
Lemma \ref{claim:GridPart} implies that no type in the support of
$P'$ will be $\varepsilon$-close (in the product topology) to a
grid element in the red shaded region $D_{\varepsilon}^{0}$ and every
type in the support of $P$ that is also in the red shaded region
is $\varepsilon$-close (in the product topology) to a grid element
in the red shaded region $D_{\varepsilon}^{0}$. Lemma \ref{claim_GridGame}
below then establishes that there exists a finite game, where it is
uniquely $\varepsilon$-rationalizable for every type to report the
closest grid element (which is also less than $\varepsilon$-close).
\begin{lem}
\label{claim_GridGame}For every $\varepsilon>0$, there is a finite
set of action profiles $A\subseteq\mathcal{T}$ and payoffs $u_{i}^{\varepsilon}\colon A\times\Theta\to\mathbb{R}$
for every player $i$, so that for every information structure $P$,
every $\varepsilon$-obedient decision rule $\sigma$ satisfies 
\[
\sigma(a|\theta,\tau)>0\implies a\in\arg\min_{a\in A}d_{\Pi}(\left(\theta,\tau\right),\left(\theta,a\right))\text{ and }d_{\Pi}(\left(\theta,\tau\right),\left(\theta,a\right))<\varepsilon.
\]
\end{lem}
We now state two known results which we also use in our proof. The
existence of BIBCE was established by \cite{stinchcombe2011correlated},
who also established it for general information structures, with redundancies
as we later report in Proposition \ref{prop: existence redundant}. 
\begin{prop}
\label{prop:(Existence-of-Equilibria)}(Existence of BIBCE) \textup{There
exists a BIBCE for every} \textup{$\left(\text{\ensuremath{\mathcal{G}}},P\right)$.}
\end{prop}
Theorem A in \cite{stinchcombe2011correlated} has established the
existence of a more demanding notion of incomplete information correlated
equilibrium, when one looks at correlated equilibrium of the agent
normal form of the game of incomplete information. \cite{forges93five}
called these ``agent normal form correlated equilibria.'' Since
every agent normal form correlated equilibrium induces an outcome
equivalent BIBCE, existence of equilibria is guaranteed. We first
establish that ACK-convergence is sufficient for continuity of equilibrium
outcomes. 
\begin{prop}
The set $\Omega^{0}\subseteq\Omega$ is dense in the product topology
on $\Omega$. 
\end{prop}
\begin{proof}
\cite{lipman2003finite} shows that finite common prior types are
dense in the product topology on $\Omega$. 
\end{proof}
\begin{prop}
\label{prop: sufficiency}For every base game $\mathcal{G}$ and every
$\varepsilon>0$, there is $\delta>0$ so that if \linebreak $d^{ACK}(P,P')<\delta$
then $d^{*}(P,P'|\mathcal{G})<\varepsilon$. 
\end{prop}
\begin{proof}
Let $\Omega_{P}:=\cap_{\epsilon>0}\text{supp}_{\varepsilon}(P)$,
$\Omega_{P'}:=\cap_{\epsilon>0}\text{supp}_{\varepsilon}(P')$ and
$\hat{\Omega}_{\delta}:=C^{1-\delta}(\hat{T}_{\delta}(P,P'))$. It
is without loss of generality to assume that $\Omega_{P}\cap\Omega_{P'}=\varnothing$.
Fix a base game with payoffs given by $u_{i}\colon A\times\Theta\to[-M,M]$
for each player $i$. Suppose $d^{ACK}(P,P')<\delta=\frac{\varepsilon}{6M(u)}$,
where $M(u):=\max\left\{ M,|A\times\Theta|\right\} <\infty$ and recall
that game $$2M\geq\max_{i}\sup_{a_{i},a_{i}',a_{-i},\theta}|u_{i}(a_{i},a_{-i},\theta)-u_{i}(a_{i}',a_{-i},\theta)|.$$
By Lemma \ref{claim:GridPart}, there is a finite set $G_{\delta}$
so that $\hat{\Omega}_{\delta}\subseteq\mathcal{N}_{\varepsilon}(G_{\delta})$
and so for every $\omega\in\hat{\Omega}_{\delta}$, $\text{\ensuremath{\min_{g\in G_{\delta}}d_{\Pi}(\omega,g)\leq\delta}}.$
Let $\zeta_{\delta}\colon\hat{\Omega}_{\delta}\to G_{\delta}$ be
any map satisfying $d_{\Pi}(\omega,\zeta_{\delta}(\omega))\leq\delta,\ \forall\ \omega\in\hat{\Omega}_{\delta}$.
We show that every decision rule which is obedient under $P$, admits
a decision rule arbitrarily close to it which is $6M\delta$-obedient
under $P'$. Consider a decision rule $\sigma\colon\Omega_{P}\to\Delta(A)$
satisfying for every $\tau_{i}$ and for every $a_{i}'\in A_{i}$,
\[
\int_{\Omega_{P}}\sum_{a\in A}\Delta u_{i}(a,a_{i}',\omega_{\theta})\sigma(a|\omega)P(\text{d}\omega|\tau_{i})\geq0,
\]
where $\Delta u_{i}(a,a_{i}',\theta)=u_{i}(a,\theta)-u_{i}(a_{i}',a_{-i},\theta)$
and $\omega_{\theta}=\text{proj}_{\Theta}(\omega)$. We have existence
of such a decision rule from Proposition \ref{prop:(Existence-of-Equilibria)}.
We now extend $\sigma$ to $\hat{\Omega}_{P}$ as follows: For every
$\omega\in\hat{\Omega}_{\delta}\cup\Omega_{P}$, define the $\delta$-extension
\[
\sigma_{\delta}(\omega):=\begin{cases}
\int\sigma(\omega')\ P(\text{d}\omega'|\zeta_{\delta}(\omega)) & \text{ if }\omega\in\hat{\Omega}_{\delta},\\
\sigma(\omega) & \text{ if }\omega\notin\hat{\Omega}_{\delta}.
\end{cases}
\]
For all $i$ and $\tau\in\text{proj}_{\mathcal{T}}(\Omega_{P})$,
\begin{multline*}
\bigg|\int_{\hat{\Omega}_{\delta}\cup\Omega_{P}}\sum_{a\in A}\Delta u_{i}(a,a_{i}',\omega_{\theta})\left(\sigma_{\delta}(a|\omega)-\sigma(a|\omega)\right)\ \tau_{i}(\text{d}\omega)\bigg|\\
\leq\sum_{g\in G_{\delta},a\in A}\Delta u_{i}(a,a_{i}',\omega_{\theta}))\\
\bigg|\left(\int\left(\int\sigma(a|\omega')\ P(\text{d}\omega'|\zeta_{\delta}^{-1}(g),\tau_{i}')\right)\ P(\text{d}\tau_{i}'|\zeta_{\delta}^{-1}(g))\right)-\int\sigma(a|\omega')\ \tau_{i}(\text{d}\omega|\zeta_{\delta}^{-1}(g))\bigg|\tau_{i}(\zeta_{\delta}^{-1}(g)\\
=\sum_{g\in G_{\delta},a\in A}\Delta u_{i}(a,a_{i}',\omega_{\theta})\bigg|\int\left(\int\sigma(a|\omega')\left(\tau_{i}'(\text{d}\omega'|\zeta_{\delta}^{-1}(g))-\tau_{i}(\text{d}\omega'|\zeta_{\delta}^{-1}(g))\right)\right)P(\text{d}\tau_{i}'|\zeta_{\delta}^{-1}(g))\bigg|P(\zeta_{\delta}^{-1}(g)|\tau_{i})\\
\leq2M\delta.
\end{multline*}
Since $\sigma$ is obedient, we conclude that $\sigma_{\delta}$ is
$2M\delta$-obedient. By construction, for every $(\theta',\tau')\in\hat{\Omega}_{\delta}\cap\Omega_{P'}$
there exists $(\theta,\tau)\in\Omega_{P}$ so that $\zeta_{\delta}(\theta',\tau')=\zeta_{\delta}(\theta,\tau)$
and $d_{\Pi}((\theta',\tau'),(\theta,\tau))<\delta$. Since $\sigma_{\delta}$
is $\zeta_{\delta}$-measurable (hence finite valued) we have that
\[
\int_{\hat{\Omega}_{\delta}\cup\Omega_{P}}\sum_{a\in A}\Delta u_{i}(a,a_{i}',\omega_{\theta})\sigma_{\delta}(a|\omega)\left(\text{ d}P(\omega|\tau_{i})-\text{ d}P'(\omega|\tau_{i}')\right)<2M\delta.
\]
$2M\delta$-obedience of $\sigma_{\delta}$ thus implies that 
\[
\begin{array}{cc}
\int_{\hat{\Omega}_{\delta}}\sum_{a\in A}\Delta u_{i}(a,a_{i}',\omega_{\theta})\sigma_{\delta}(a|\omega)\text{ d}P'(\omega|\tau_{i}') & \geq-4M\delta\end{array}.
\]
Moreover, $(\theta',\tau')\in\hat{\Omega}_{\delta}$ implies that
for any measurable $\hat{\sigma}\colon\Omega_{P'}\setminus\hat{\Omega}_{\delta}\to\Delta(A)$,
we have $\int_{\Omega_{P'}\setminus\hat{\Omega}_{\delta}}\sum_{a\in A}\Delta u_{i}(a,a_{i}',\omega_{\theta})\hat{\sigma}(a|\omega)\ P'(\text{d}\omega|\tau_{i}')\geq-2M\delta$
and so 
\begin{equation}
\int_{\Omega_{P'}}\sum_{a\in A}\Delta u_{i}(a,a_{i}',\omega_{\theta})\sigma_{\delta}(a|\omega)\ P'(\text{d}\omega|\tau_{i}')\geq-6M\delta.\label{eq:mainproof2}
\end{equation}
So $\sigma_{\delta}$ satisfies $6M\delta$-obedience under $P'$
if restricted to type profiles in $\hat{\Omega}_{\delta}$. We now
argue that there exists a $6M\delta$-obedient decision rule under
$P'$ that agrees with the extension $\sigma_{\delta}$ on $\hat{\Omega}_{\delta}$.
For every player $i$, let $D_{i}:=\left\{ \tau\in\mathcal{T}:\tau_{i}(\hat{\Omega}_{\delta})\leq1-\delta\right\} $
and note that if $P'(\hat{\Omega}_{\delta})<1$ then $P'(\bigcup_{i\in I}D_{i})>0$. 

Consider the auxiliary payoffs $\tilde{u}\colon\Omega\times A\to\mathbb{R}^{I}$,
defined for every $(\theta,\tau)\in\Omega$ and $a\in A$ as follows
\[
\tilde{u}_{i}((\theta,\tau),a):=\begin{cases}
u_{i}(\theta,a_{-i},a_{i}), & \text{ if }\tau_{i}\left(\hat{\Omega}_{\delta}\right)\leq1-\delta\text{ }\\
\boldsymbol{1}_{\sigma_{\delta}(a|\theta,\tau)>0}, & \text{otherwise.}
\end{cases}
\]
By Proposition \ref{prop:(Existence-of-Equilibria)} we deduce that
the incomplete information game $(\tilde{u},P')$ admits an obedient
decision rule $\bar{\sigma}$ that coincides with $\sigma_{\delta}$
on $\hat{\Omega}_{\delta}$. This induces a $6M\delta$-obedient decision
of $(u,P')$ defined on all of $\Omega_{P'}$: Indeed, by obedience
of $\bar{\sigma}$ for every player $i$ and $(\theta,\tau')\in\Omega_{P'}$
so that $\tau_{i}'(\hat{\Omega}_{\delta})\leq1-\delta$, 
\begin{multline}
\int_{\Omega_{P'}}\sum_{a\in A}\Delta\tilde{u}_{i}(a,a_{i}',\omega)\bar{\sigma}(a|\omega)\ P'(\text{d}\omega|\tau_{i}')=\int_{\Omega_{P'}}\sum_{a\in A}\Delta u_{i}(a,a_{i}',\omega_{\theta})\bar{\sigma}(a|\omega)\ P'(\text{d}\omega|\tau_{i}')\geq0.\label{eq:mainproof3}
\end{multline}
Moreover, for every $(\theta,\tau')\in\Omega_{P'}$ so that $\tau_{i}'(\hat{\Omega}_{\delta})>1-\delta$,
there is $\tilde{\tau}'\in\hat{\Omega}_{\delta}$ so that $\tilde{\tau}'_{i}=\tau{}_{i}'$.
So consider the combined decision rule,
\[
\sigma'(\omega):=\begin{cases}
\sigma_{\delta}(\omega) & ,\text{ if }\omega\in\hat{\Omega}_{\delta}\\
\bar{\sigma}(\omega) & ,\text{ if }\omega\in\Omega_{P'}\setminus\hat{\Omega}_{\delta}.
\end{cases}
\]
Combining (\ref{eq:mainproof2}) and (\ref{eq:mainproof3}) we deduce
that for every $\tau$
\[
\int_{\Omega_{P'}}\sum_{a\in A}\Delta u_{i}(a,a_{i}',\omega_{\theta})\sigma'(a|\omega)\ P'(\text{d}\omega|\tau_{i}')\geq-6M\delta
\]

We will now show that $\nu_{P',\sigma'}$ is close to $\nu_{P,\sigma}$.
For any $(a,\theta)\in A\times\Theta$,

\begin{multline*}
|\nu_{P,\sigma}(a,\theta)-\nu_{P',\sigma'}(a,\theta)|\\
\leq\bigg|\sum_{g\in G_{\delta}}\left(\int\sigma(a|\tau',\theta)\ P(\theta,\text{d}\tau'|g)\right)(P(\theta,\text{d}\tau)-P'(\theta,\text{d}\tau))\bigg|\\
+|P'(\hat{\Omega}_{\delta}|\theta)-P(\hat{\Omega}_{\delta}|\theta)|\leq2\delta.
\end{multline*}

Hence $\sum_{a,\theta}|\nu_{P,\sigma}(a,\theta)-\nu_{P',\hat{\sigma}}(a,\theta)|^{2}<|\Theta\times A|4\delta^{2}$
and so $d^{u}(P,P')<6M(u)\delta$.
\end{proof}
We now establish that failure of ACK-convergence implies a failure
of convergence of equilibrium outcomes.
\begin{prop}
\label{prop: necessity-1}For every $\varepsilon>0$, if $d^{ACK}(P,P')\geq\varepsilon$
then there is a game $\mathcal{G}$ so that $d^{*}(P,P'|\mathcal{G})\geq\varepsilon$. 
\end{prop}
\begin{proof}
We now establish that for all $\varepsilon\in(0,1/2)$ so that for
all $P,P'$ satisfy $d^{ACK}(P,P')\geq\varepsilon$, we also have
$d^{*}(P,P'|\hat{\mathcal{G}})\geq\varepsilon$ for some game $\hat{\mathcal{G}}$:
If convergence fails in our metric, then there must be some game on
which ex-ante strategic convergence fails. Then we must find such
a game $\mathcal{\hat{G}}$. The condition $d^{ACK}(P,P')\geq\varepsilon$
means that $P\left(\hat{\Omega}_{\varepsilon}\right)\leq1-\varepsilon$,
$P'\left(\hat{\Omega}_{\varepsilon}\right)\leq1-\varepsilon$ or both.
Suppose that $P(\hat{\Omega}_{\varepsilon})\leq1-\varepsilon$ and
$P'(\hat{\Omega}_{\varepsilon})>1-\varepsilon$. First, note that
$P(\hat{T}_{\varepsilon}(P,P'))<1$. Indeed, if $P(\hat{T}_{\varepsilon}(P,P'))=1$,
we also have that $P(\hat{\Omega}_{\varepsilon})=1$, which is a contradiction.
Let $D_{\varepsilon,P}:=\text{supp}_{\varepsilon}(P)\setminus\hat{T}_{\varepsilon}(P,P')$
and $D_{\varepsilon,P}^{\complement}:=\Omega\setminus\mathcal{N}_{\varepsilon}(D_{\varepsilon,P})$.
From Lemmas \ref{claim:GridPart} and \ref{claim_GridGame} we conclude
that there is $m$ and $z$, an associated game $\hat{\mathcal{G}}=\left(A^{m,z},\left(u_{i}^{m,z}\right)_{i}\right)$
where the finite collection of action profiles takes the form $A^{m,z}=\hat{D}_{\varepsilon,P}\cup\hat{D}_{\varepsilon,P}^{\complement}$,
with $\hat{D}_{\varepsilon,P}\subseteq D_{\varepsilon,P}$ a finite
$\varepsilon$-grid, $\hat{D}_{\varepsilon,P}^{\complement}\subseteq D_{\varepsilon,P}^{\complement}$
a finite $\varepsilon$-grid, and so that for every $\varepsilon$-BIBCE,
$\hat{\sigma}'\in\mathcal{B^{\varepsilon}}(\hat{\mathcal{G}},P')$
and every $\omega\in\Omega_{P'}$, $\hat{\sigma}'(\hat{D}_{\varepsilon,P}|\omega)=0$.
Moreover, for every $\omega\in D_{\varepsilon,P}$ and every $\varepsilon$-BIBCE,
$\hat{\sigma}\in\mathcal{B^{\varepsilon}}(\hat{\mathcal{G}},P)$,
$\hat{\sigma}(\hat{D}_{\varepsilon,P}|\omega)=1$. 

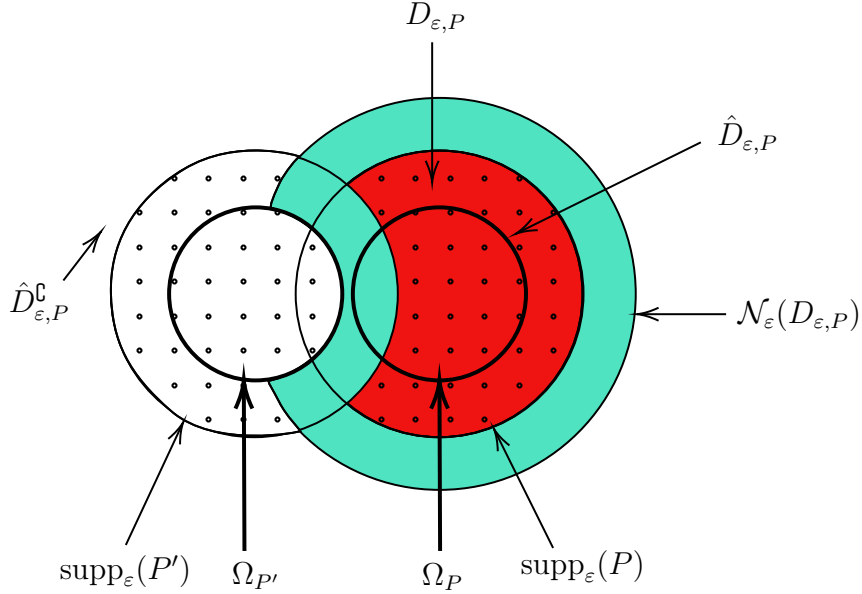
\begin{figure}[H]
\begin{centering}
\begin{center}
    \tikzset{ pattern size/.store in=\mcSize,  pattern size = 5pt, pattern thickness/.store in=\mcThickness,  pattern thickness = 0.3pt, pattern radius/.store in=\mcRadius,  pattern radius = 1pt} \makeatletter \pgfutil@ifundefined{pgf@pattern@name@_0aq1wrn9f}{ \makeatletter \pgfdeclarepatternformonly[\mcRadius,\mcThickness,\mcSize]{_0aq1wrn9f} {\pgfpoint{-0.5*\mcSize}{-0.5*\mcSize}} {\pgfpoint{0.5*\mcSize}{0.5*\mcSize}} {\pgfpoint{\mcSize}{\mcSize}} { \pgfsetcolor{\tikz@pattern@color} \pgfsetlinewidth{\mcThickness} \pgfpathcircle\pgfpointorigin{\mcRadius} \pgfusepath{stroke} }} \makeatother
  \tikzset{ pattern size/.store in=\mcSize,  pattern size = 5pt, pattern thickness/.store in=\mcThickness,  pattern thickness = 0.3pt, pattern radius/.store in=\mcRadius,  pattern radius = 1pt} \makeatletter \pgfutil@ifundefined{pgf@pattern@name@_b9nf8fzgk}{ \makeatletter \pgfdeclarepatternformonly[\mcRadius,\mcThickness,\mcSize]{_b9nf8fzgk} {\pgfpoint{-0.5*\mcSize}{-0.5*\mcSize}} {\pgfpoint{0.5*\mcSize}{0.5*\mcSize}} {\pgfpoint{\mcSize}{\mcSize}} { \pgfsetcolor{\tikz@pattern@color} \pgfsetlinewidth{\mcThickness} \pgfpathcircle\pgfpointorigin{\mcRadius} \pgfusepath{stroke} }} \makeatother
  \tikzset{ pattern size/.store in=\mcSize,  pattern size = 5pt, pattern thickness/.store in=\mcThickness,  pattern thickness = 0.3pt, pattern radius/.store in=\mcRadius,  pattern radius = 1pt} \makeatletter \pgfutil@ifundefined{pgf@pattern@name@_l8zlcy7k8}{ \makeatletter \pgfdeclarepatternformonly[\mcRadius,\mcThickness,\mcSize]{_l8zlcy7k8} {\pgfpoint{-0.5*\mcSize}{-0.5*\mcSize}} {\pgfpoint{0.5*\mcSize}{0.5*\mcSize}} {\pgfpoint{\mcSize}{\mcSize}} { \pgfsetcolor{\tikz@pattern@color} \pgfsetlinewidth{\mcThickness} \pgfpathcircle\pgfpointorigin{\mcRadius} \pgfusepath{stroke} }} \makeatother \tikzset{every picture/.style={line width=0.75pt}}  \begin{tikzpicture}[x=0.75pt,y=0.75pt,yscale=-1,xscale=1]
\draw  [draw opacity=0][pattern=_0aq1wrn9f,pattern size=13.049999999999999pt,pattern thickness=0.75pt,pattern radius=0.75pt, pattern color={rgb, 255:red, 0; green, 0; blue, 0}] (125.75,67.5) -- (423.25,67.5) -- (423.25,315.75) -- (125.75,315.75) -- cycle ; \draw  [fill={rgb, 255:red, 80; green, 227; blue, 194 }  ,fill opacity=1 ] (208.83,193) .. controls (208.83,138.37) and (253.12,94.08) .. (307.75,94.08) .. controls (362.38,94.08) and (406.67,138.37) .. (406.67,193) .. controls (406.67,247.63) and (362.38,291.92) .. (307.75,291.92) .. controls (253.12,291.92) and (208.83,247.63) .. (208.83,193) -- cycle ; \draw  [fill={rgb, 255:red, 255; green, 0; blue, 0 }  ,fill opacity=0.91 ] (235.17,193) .. controls (235.17,153.05) and (267.55,120.67) .. (307.5,120.67) .. controls (347.45,120.67) and (379.83,153.05) .. (379.83,193) .. controls (379.83,232.95) and (347.45,265.33) .. (307.5,265.33) .. controls (267.55,265.33) and (235.17,232.95) .. (235.17,193) -- cycle ; \draw  [pattern=_b9nf8fzgk,pattern size=13.049999999999999pt,pattern thickness=0.75pt,pattern radius=0.75pt, pattern color={rgb, 255:red, 0; green, 0; blue, 0}] (235.67,193) .. controls (235.67,153.05) and (268.05,120.67) .. (308,120.67) .. controls (347.95,120.67) and (380.33,153.05) .. (380.33,193) .. controls (380.33,232.95) and (347.95,265.33) .. (308,265.33) .. controls (268.05,265.33) and (235.67,232.95) .. (235.67,193) -- cycle ; \draw  [fill={rgb, 255:red, 80; green, 227; blue, 194 }  ,fill opacity=1 ] (142,193) .. controls (142,153.05) and (174.38,120.67) .. (214.33,120.67) .. controls (254.28,120.67) and (286.67,153.05) .. (286.67,193) .. controls (286.67,232.95) and (254.28,265.33) .. (214.33,265.33) .. controls (174.38,265.33) and (142,232.95) .. (142,193) -- cycle ; \draw  [line width=1.5]  (264,193) .. controls (264,168.88) and (283.55,149.33) .. (307.67,149.33) .. controls (331.78,149.33) and (351.33,168.88) .. (351.33,193) .. controls (351.33,217.12) and (331.78,236.67) .. (307.67,236.67) .. controls (283.55,236.67) and (264,217.12) .. (264,193) -- cycle ; \draw  [fill={rgb, 255:red, 255; green, 255; blue, 255 }  ,fill opacity=1 ] (237,123.75) .. controls (237.98,123.51) and (205.33,113.67) .. (178.5,130.25) .. controls (151.67,146.83) and (141.33,169.17) .. (142,193) .. controls (142.67,216.83) and (150.5,233.25) .. (172.5,252.25) .. controls (194.5,271.25) and (237.5,262.75) .. (237.5,262.25) .. controls (237.5,261.75) and (230.33,254.5) .. (227,247.83) .. controls (223.67,241.17) and (220.26,236.75) .. (222,236.75) .. controls (223.74,236.75) and (259,226.75) .. (258.67,193) .. controls (258.33,159.25) and (225.17,149.83) .. (223.17,149.83) .. controls (221.17,149.83) and (224.83,141.75) .. (227.17,136.83) .. controls (229.5,131.92) and (236.02,123.99) .. (237,123.75) -- cycle ; \draw  [line width=1.5]  (171.33,193) .. controls (171.33,168.88) and (190.88,149.33) .. (215,149.33) .. controls (239.12,149.33) and (258.67,168.88) .. (258.67,193) .. controls (258.67,217.12) and (239.12,236.67) .. (215,236.67) .. controls (190.88,236.67) and (171.33,217.12) .. (171.33,193) -- cycle ; \draw  [pattern=_l8zlcy7k8,pattern size=13.049999999999999pt,pattern thickness=0.75pt,pattern radius=0.75pt, pattern color={rgb, 255:red, 0; green, 0; blue, 0}] (237,123.75) .. controls (237.98,123.51) and (205.33,113.67) .. (178.5,130.25) .. controls (151.67,146.83) and (141.33,169.17) .. (142,193) .. controls (142.67,216.83) and (150.5,233.25) .. (172.5,252.25) .. controls (194.5,271.25) and (237.5,262.75) .. (237.5,262.25) .. controls (237.5,261.75) and (230.33,254.5) .. (227,247.83) .. controls (223.67,241.17) and (220.26,236.75) .. (222,236.75) .. controls (223.74,236.75) and (259,226.75) .. (258.67,193) .. controls (258.33,159.25) and (225.17,149.83) .. (223.17,149.83) .. controls (221.17,149.83) and (224.83,141.75) .. (227.17,136.83) .. controls (229.5,131.92) and (236.02,123.99) .. (237,123.75) -- cycle ; \draw   (235.17,193) .. controls (235.17,153.05) and (267.55,120.67) .. (307.5,120.67) .. controls (347.45,120.67) and (379.83,153.05) .. (379.83,193) .. controls (379.83,232.95) and (347.45,265.33) .. (307.5,265.33) .. controls (267.55,265.33) and (235.17,232.95) .. (235.17,193) -- cycle ; \draw    (439.5,116.5) -- (347.79,161.86) ; \draw [shift={(346,162.75)}, rotate = 333.68] [color={rgb, 255:red, 0; green, 0; blue, 0 }  ][line width=0.75]    (10.93,-3.29) .. controls (6.95,-1.4) and (3.31,-0.3) .. (0,0) .. controls (3.31,0.3) and (6.95,1.4) .. (10.93,3.29)   ; \draw    (304,66) -- (304,134.75) ; \draw [shift={(304,136.75)}, rotate = 270] [color={rgb, 255:red, 0; green, 0; blue, 0 }  ][line width=0.75]    (10.93,-3.29) .. controls (6.95,-1.4) and (3.31,-0.3) .. (0,0) .. controls (3.31,0.3) and (6.95,1.4) .. (10.93,3.29)   ; \draw [line width=1.5]    (308,322.75) -- (307.68,239.67) ; \draw [shift={(307.67,236.67)}, rotate = 89.78] [color={rgb, 255:red, 0; green, 0; blue, 0 }  ][line width=1.5]    (14.21,-4.28) .. controls (9.04,-1.82) and (4.3,-0.39) .. (0,0) .. controls (4.3,0.39) and (9.04,1.82) .. (14.21,4.28)   ; \draw    (365,321.25) -- (336.06,262.96) ; \draw [shift={(335.17,261.17)}, rotate = 63.59] [color={rgb, 255:red, 0; green, 0; blue, 0 }  ][line width=0.75]    (10.93,-3.29) .. controls (6.95,-1.4) and (3.31,-0.3) .. (0,0) .. controls (3.31,0.3) and (6.95,1.4) .. (10.93,3.29)   ; \draw [line width=1.5]    (209.33,322.75) -- (209.01,239.67) ; \draw [shift={(209,236.67)}, rotate = 89.78] [color={rgb, 255:red, 0; green, 0; blue, 0 }  ][line width=1.5]    (14.21,-4.28) .. controls (9.04,-1.82) and (4.3,-0.39) .. (0,0) .. controls (4.3,0.39) and (9.04,1.82) .. (14.21,4.28)   ; \draw    (149.5,319.75) -- (177.65,259.48) ; \draw [shift={(178.5,257.67)}, rotate = 115.04] [color={rgb, 255:red, 0; green, 0; blue, 0 }  ][line width=0.75]    (10.93,-3.29) .. controls (6.95,-1.4) and (3.31,-0.3) .. (0,0) .. controls (3.31,0.3) and (6.95,1.4) .. (10.93,3.29)   ; \draw    (118,186.25) -- (135.78,163.09) ; \draw [shift={(137,161.5)}, rotate = 127.51] [color={rgb, 255:red, 0; green, 0; blue, 0 }  ][line width=0.75]    (10.93,-3.29) .. controls (6.95,-1.4) and (3.31,-0.3) .. (0,0) .. controls (3.31,0.3) and (6.95,1.4) .. (10.93,3.29)   ; \draw    (452.25,203.25) -- (409.25,203.25) ; \draw [shift={(407.25,203.25)}, rotate = 360] [color={rgb, 255:red, 0; green, 0; blue, 0 }  ][line width=0.75]    (10.93,-3.29) .. controls (6.95,-1.4) and (3.31,-0.3) .. (0,0) .. controls (3.31,0.3) and (6.95,1.4) .. (10.93,3.29)   ;
\draw (297.33,327.4) node [anchor=north west][inner sep=0.75pt]    {$\Omega _{P}$}; \draw (446.17,101.23) node [anchor=north west][inner sep=0.75pt]    {$\hat{D}_{\varepsilon ,P}$}; \draw (89,185.73) node [anchor=north west][inner sep=0.75pt]    {$\hat{D}_{\varepsilon ,P}^{\complement }$}; \draw (288.83,44.4) node [anchor=north west][inner sep=0.75pt]    {$D_{\varepsilon ,P}$}; \draw (201.5,327.07) node [anchor=north west][inner sep=0.75pt]    {$\Omega _{P'}$}; \draw (114.67,323.4) node [anchor=north west][inner sep=0.75pt]    {$\text{supp}_{\varepsilon }( P')$}; \draw (347.17,320.9) node [anchor=north west][inner sep=0.75pt]    {$\text{supp}_{\varepsilon }( P)$}; \draw (456.83,194.4) node [anchor=north west][inner sep=0.75pt]    {$\mathcal{N}_{\varepsilon }( D_{\varepsilon ,P})$};
\end{tikzpicture}  
\par\end{center}
\par\end{centering}
\caption{\label{fig:Infection-argument.-2}Domain of Grid Game.}
\end{figure}
Figure \ref{fig:Infection-argument.-2} above illustrates the sets
we just defined. The area enclosed by the two bold circles represent
the supports of $P$ and $P'$ respectively. The red shaded area represents
the set $D_{\varepsilon,P}$ where the dots on top of the red shaded
area represents $\hat{D}_{\varepsilon,P}$, the actions chosen by
type profiles in $D_{\varepsilon,P}$. The green shaded area represents
$\mathcal{N}_{\varepsilon}(D_{\varepsilon,P})\setminus D_{\varepsilon,P}$,
which contains no action in $A^{m,z}$. Finally, the dots with white
background represent the remaining actions, i.e. the set $\hat{D}_{\varepsilon,P}^{\complement}$.
Every type in $\Omega_{P'}$ will pick an action from that set.Let 

\[
D_{P}^{m}:=\left\{ \omega\in\Omega_{P}:\exists\ \hat{\sigma}\in\mathcal{B^{\varepsilon}}(\hat{\mathcal{G}},P)\text{ s.t. }\hat{\sigma}(\hat{D}_{\varepsilon,P}|\omega)=1\right\} .
\]
For every player $i$, let $D_{P,i}^{m}:=\left\{ \tau_{i}:\exists\ (\theta,\hat{\tau})\in\text{supp}(P),\ \left(\theta,\left(\tau_{i},\hat{\tau}_{-i}\right)\right)\in D_{P}^{m}\right\} $.
Define the sets $D_{P'}^{m}$ and $D_{P',i}^{m}$ for prior $P'$
analogously. Define the sequence $(D_{P}^{m+n})_{n\in\mathbb{N}}$,
recursively for every $n\geq1$,
\[
D_{P,i}^{m+n}:=\Omega_{P}\setminus B_{i}^{1-\varepsilon}(D_{P}^{m+n-1}),\ D_{P}^{m+n}:=\prod_{i\in I}D_{P,i}^{m+n}.
\]
Let $\hat{D}_{P,i}:=\left\{ \tau_{i}:\exists\ (\theta,\tau_{-i})\in\Theta\times\mathcal{T}_{-i}\text{ s.t. }(\theta,\tau)\in\hat{D}_{\varepsilon,P}\right\} $.
Consider the action set $A_{i}^{*}:=\left\{ a^{*},a^{**}\right\} ^{I}\times A^{m,z}$
and let payoffs of $\hat{\mathcal{G}}^{*}$ be given by 
\[
\hat{u}_{i}(a^{0},d,\theta)=\begin{cases}
\varepsilon\boldsymbol{1}_{a_{i}^{0}=a^{*}}+u_{i}^{m,z}\left(d,\theta\right) & \text{ if }d_{i}\in\hat{D}_{P,i}\\
\boldsymbol{1}_{a_{i}^{0}=a^{*}}+u_{i}^{m,z}\left(d,\theta\right) & \text{ if }d_{i}\notin\hat{D}_{P,i}\text{ and }\exists\ j\neq i\text{ s.t. }a_{j}^{0}=a^{*}
\end{cases}
\]
In this game, playing $d_{i}\in\hat{D}_{P,i}$ and $a^{*}$ is a dominant
action for every type in $D_{P,i}^{m}$. Under payoffs $u^{m,z}$
and prior $P'$, reporting a grid element in $D_{P}^{m}$ has cost
at least $\varepsilon>0$. Hence every type under prior $P'$ rationalizes
both $a^{**}$ and $a^{*}$ while no type under prior $P$ in the
set $D_{P,i}^{m}$ rationalizes $a^{**}$. Now consider $\tau_{i}\in D_{P,i}^{m+1}$.
Again, playing $a_{i}^{0}=a^{*}$ is uniquely $\varepsilon$-rationalizable
since all types in $D_{P,i}^{m}$ play $a^{*}$: Indeed, for every
$\tau_{i}\in D_{P,i}^{m+1}$, there is a player $i$ so that $\tau_{i}(D_{P,-i}^{m})\geq\varepsilon$.
For every $n\in\mathbb{N}$, there is a player $i_{n}$ so that $D_{P,i_{n}}^{m+n}\neq\varnothing$.
Proceeding inductively we obtain that $a_{i}^{0}=a^{*}$ remains the
unique $\varepsilon$-best-reply for some player at type profiles
in $\Omega_{P}\setminus\hat{\Omega}_{\varepsilon}$. Deduce that there
is a $\varepsilon$-BIBCE where all types of all players play $a^{**}$
under $P'$ while for all $\varepsilon$-BIBCE under $P$, $P$ assigns
at least probability $\varepsilon$ to type profiles where some player's
type plays $a^{*}$. Hence $d^{*}(P,P'|\hat{\mathcal{G}}^{*})\geq\varepsilon$.
A symmetric argument shows that whenever 1) $P(\hat{\Omega}_{\varepsilon})>1-\varepsilon$
and $P'(\hat{\Omega}_{\varepsilon})\leq1-\varepsilon$ or 2) $P(\hat{\Omega}_{\varepsilon})>1-\varepsilon$
and $P'(\hat{\Omega}_{\varepsilon})>1-\varepsilon$ implies that outcomes
are $\varepsilon$ apart under a similarly constructed game. Note
that under $P(\hat{\Omega}_{\varepsilon})>1-\varepsilon$ and $P'(\hat{\Omega}_{\varepsilon})>1-\varepsilon$,
the game with payoffs $u^{m,z}$ is enough to separate the outcomes
of $P$ and $P'$. 

The grid game alone would not work for our proof: Indeed, we could
construct a sequence of priors $P^{k}$ for any prior $P$ so that
$d^{ACK}(P^{k},P)>\varepsilon$ for all $k$ but where $P^{k}(T_{\varepsilon}(P^{k},P))\uparrow1$,
i.e. the infecting event has diminishing ex-ante probability. If we
used the grid game only, we would have that $d^{*}(P,P'|\hat{\mathcal{G}})\downarrow0$.
Based on this game alone, the ACK topology would be too strong. 
\end{proof}
Now Theorem \ref{thm: main} follows immediately from Propositions
\ref{prop: sufficiency} and \ref{prop: necessity-1}. 

We will conclude this section by discussing a key property of ACK
topology (denseness) and an extension of our main result (so we have
continuity of exact BIBCE). Both will be important for our applications. 

\subsection{\label{subsec:Denseness-of-Simple}Denseness of Simple Types}

An information structure $P$ is \emph{finite} if the support of $P$
is finite. A finite information structure \emph{$P$} is a \emph{first-order
belief} information structure if each type has a distinct first-order
belief. An information structure $P$ is \emph{simple} if it is both
a finite and a first-order belief information structure. We denote
the collection of simple information structures 
\[
\mathcal{P}^{*}:=\left\{ P\in\mathcal{P}:|\text{supp}(P)|<\infty,\ \forall\ (\theta,\tau),(\hat{\theta},\hat{\tau})\in\text{supp}(P),\forall\ i,\ \tau_{i}\neq\hat{\tau}_{i}\implies\tau_{i}^{1}\neq\hat{\tau}_{i}^{1}\right\} 
\]

It is often convenient to work with simple information structures.
We have:
\begin{prop}
\label{prop: denseness }Simple information structures are dense in
$\mathcal{P}$ under the ACK topology.
\end{prop}
\begin{proof}
Fix any information structure $P$. For any $\delta>0$ we construct
a finite information structure $P_{\delta}$ so that $d^{*}(P,P_{\delta})<\delta$.
To do so, we construct a finite grid $G_{\delta}$ of state and type
pairs whose $\delta$-neighborhood covers $\Omega_{P}$, i.e. $\Omega_{P}\subseteq\mathcal{N}_{\delta}(G_{\delta})$
and for all $\omega\in\Omega_{P},$ $\min_{g\in G_{\varepsilon}}d_{\Pi}(\omega,g)\leq\delta$.
Consider any partition on $\Omega_{P}$ given by the pre-image of
any map $\zeta\colon\Omega_{P}\to G_{\delta}$, where for all $\omega\in\Omega_{P}$,
we have that $d_{\Pi}(\omega,\zeta(\omega))\leq\delta$. Consider
the information structure $P_{\delta}$ given by $P\circ\zeta^{-1}$.
This information structure is finite and has the property that for
any $\tau\in\text{supp}(P{}_{\delta})$ there exists $\tau'\in\text{supp}(P)$
so that $d_{\Pi}(\tau,\tau')<\delta$ and vice versa; for all $\tau'\in\text{supp}(P)$
there is $\tau\in\text{supp}(P_{\delta})$ so that $d_{\Pi}(\tau,\tau')<\delta$.
Indeed, for every player $i$ and $\omega\in\Omega_{P}$ and $(\theta,\tau)=\zeta(\omega)$,
beliefs of any measurable event $E\in\mathscr{B}$ at $\tau_{i}$
are given by, 
\[
P\circ\zeta^{-1}(E|\tau_{i})=\int_{\zeta_{i}^{-1}(\tau_{i})}P(E|\hat{\tau}_{i})P(\text{d}\hat{\tau}_{i}|\zeta_{i}^{-1}(\tau_{i})),
\]
where $\zeta_{i}^{-1}(\tau_{i}):=\bigcup_{\left(\hat{\theta},\hat{\tau}\right)\in G_{\delta}:\hat{\tau}_{i}=\tau_{i}}\left\{ \tilde{\tau}_{i}:\exists\ (\tilde{\theta},\tilde{\tau}_{-i})\text{ s.t. }(\tilde{\theta},\tilde{\tau})\in\zeta^{-1}(\hat{\theta},\hat{\tau})\right\} $.
First order beliefs of $P\circ\zeta^{-1}$ take the form $P\circ\zeta^{-1}(\theta|\tau_{i})=\int_{\zeta_{i}^{-1}(\tau_{i})}P(\theta|\hat{\tau}_{i})P(\text{d}\hat{\tau}_{i}|\zeta_{i}^{-1}(\tau_{i}))$
and so for every $\tilde{\tau}_{i}\in\zeta_{i}^{-1}(\tau_{i})$, first-order
beliefs of $P\circ\zeta^{-1}(\cdot|\tau_{i})$ and $P(\cdot|\tilde{\tau}_{i})$
are no more than $\delta$ apart in the Euclidean topology on $\Delta(\Theta)$.
Deduce that $P\circ\zeta^{-1}(\cdot|\tau_{i})$ and $P(\cdot|\tilde{\tau}_{i})$
have $\delta$-close hierarchies of beliefs in the product topology
on $\mathcal{T}_{i}$. Hence $\hat{T}_{\delta}(P,P_{\delta})=\text{supp}_{\delta}(P)$
and so $d^{ACK}\left(P,P_{\delta}\right)<\delta$. For $\delta$ small
enough, there is $\hat{\zeta}\colon\mathcal{T}\to\mathcal{T},$so
that our choice of $\zeta$ above satisfies $\zeta(\theta,\tau)=(\theta,\hat{\zeta}(\tau))$
and so the marginal of $P\circ\zeta^{-1}$ on $\Theta$ coincides
with that of $P$. Hence finite, canonical information structures
are dense in $\mathcal{P}$ under the strategic topology. It remains
to show that $P_{\delta}$ is close to a simple canonical information
structure. Let $\Omega_{i,P_{\delta}}(\tau_{i}):=\left\{ \hat{\tau}_{i}:P_{\delta}(\hat{\tau}_{i})>0,\hat{\tau}_{i}^{1}=\tau_{i}^{1}\right\} $.
For any $\epsilon>0$ let $\rho_{i,\varepsilon}\colon\Theta\times\mathcal{T}_{i}\to\Delta(\left\{ 0,1\right\} )$
have the property that $\rho_{i,\varepsilon}(\theta,\hat{\tau}_{i})=\rho_{i,\varepsilon}(\theta',\hat{\tau}_{i})$
for all $\theta,\theta'\in\Theta$ and all $\hat{\tau}_{i}$ in the
support of $P_{\delta}$ satisfying $|\Omega_{i,P_{\delta}}(\hat{\tau}_{i})|=1$.
For any $\hat{\tau}_{i}$ in the support of $P_{\delta}$ satisfying
$|\Omega_{i,P_{\delta}}(\hat{\tau}_{i})|>1$ let $||\rho_{i,\varepsilon}(\theta,\hat{\tau}_{i})-\rho_{i,\varepsilon}(\theta',\hat{\tau}_{i})||_{2}<\epsilon$
so that for all distinct $\tilde{\tau}_{i},\bar{\tau}_{i}\in\Omega_{i,P_{\delta}}(\tau_{i})$
and $s_{i},s_{i}'\in\left\{ 0,1\right\} $,
\[
\frac{\rho_{i,\varepsilon}(s_{i}|\theta,\tilde{\tau_{i}})P_{\delta}(\theta|\tilde{\tau}_{i})}{\sum_{\hat{\theta}}\rho_{i,\varepsilon}(s_{i}|\hat{\theta},\tilde{\tau_{i}})P_{\delta}(\hat{\theta}|\tilde{\tau}_{i})}\neq\frac{\rho_{i,\varepsilon}(s_{i}'|\theta,\bar{\tau}_{i})P_{\delta}(\theta|\bar{\tau}_{i})}{\sum_{\hat{\theta}}\rho_{i,\varepsilon}(s_{i}'|\hat{\theta},\bar{\tau}_{i})P_{\delta}(\hat{\theta}|\bar{\tau}_{i})}.
\]
The prior $\hat{P}_{\delta}(\theta,\tau,s)=P_{\delta}(\theta,\tau)\prod_{i}\rho_{i,\varepsilon}(s_{i}|\theta,\tau)$
induces information structure $\tilde{P}_{\delta}$ so that \linebreak $d^{ACK}(P,\tilde{P}_{\delta})<\delta+\varepsilon$.
Deduce that simple information structures are dense in $\mathcal{P}$. 
\end{proof}

\section{General Information Structures\protect \\
and Bayes Nash Equilibrium\label{sec:Bayes-Nash-Equilibrium} }

Our main approach in this paper is to remove correlating devices from
the information structure (and thus work with non-redundant information
structures) and put correlating devices in the solution concept (BIBCE).
However, to relate our work to the literature and discuss applications
it is useful to discuss how our results apply when we allow correlating
devices in the information structure (and thus allow for general \emph{redundant
information structures}) but remove correlating devices from the solution
concept (and work with Bayes Nash equilibrium).

\subsection{General Information Structures}

A (common prior) \emph{general information structure} describes a
set of signals for each player and a joint distribution over states
and profiles of signals: thus a general information structure is a
tuple $\text{\ensuremath{\mathcal{S}}}=\left((S_{i})_{i\in I},Q\right)$,
where each $S_{i}$ is a measurable space of signals\footnote{We adopt the convention of referring to ``signals'' rather than
``types'' when describing general information structures in this
section. We reserve the terminology ``type'' for hierarchies of
beliefs, introduced in the next section.} for player $i$ and $Q$ is a probability measure\footnote{The product of measurable spaces is always endowed with the product
sigma algebra.} in $\Delta(\Theta\times S)$ whose marginal on $\Theta$ is given
by $\mu$.

We dub this object a general information structure. \cite{liu15}
describes how one can always decompose a general information structure
into a (non-redundant) information structure and a correlation device.
While a general information structure is described by a signal space
and a probability measure, we will adopt the convention of describing
non-redundant information structures as just a measure, leaving it
understood that the player's signal space is the universal space of
hierarchies. 

Every redundant information structure and can be naturally mapped
to its non-redundant information structure by essentially integrating
out redundant types. In particular, we map information structure $\mathcal{S}=\left((S_{i})_{i\in I},Q\right)$
to a (non-redundant) information structure $P$ as follows. For every
$i$, and version of the conditional probability $Q_{i}$, first-order
beliefs can be obtained for any player $i$, $\overline{\tau}_{i}^{1}(s_{i})=\text{marg}_{\Theta}(Q_{i}(s_{i}))$
for every $s_{i}\in S_{i}$. For every $m>1$ and $m-1$-order beliefs
representation $\overline{\tau}_{j}^{m-1}\colon S_{j}\to\mathcal{T}_{j}^{m-1}$,
for every $j$, obtain the $m$-order belief representation of player
$i$, $\overline{\tau}_{i}^{m}(s_{i})=Q_{i}(s_{i})\circ(\text{id}\times\overline{\tau}_{-i}^{m-1})^{-1}$,
where $\text{id}$ is the identity on $\Theta$ and $\bar{\tau}_{-i}^{m-1}\colon s_{-i}\mapsto(\bar{\tau}_{j}^{m-1}(s_{j}))_{j\neq i}$.
The representation of $s_{i}$ in $\mathcal{T}_{i}$ is then given
by $\overline{\tau}_{i}(s_{i})=(\overline{\tau}_{i}^{m}(s_{i}))_{m}$.
For every $s\in S$ let $\overline{\tau}(s):=(\overline{\tau}_{i}(s_{i}))_{i\in I}$
and we write $P_{\mathcal{S}}$ for the information structure thus
induced by redundant information structure $\mathcal{S}$. 

\subsection{Solution Concepts}

Now we will say that a base game and a redundant information structure
$\left(\text{\ensuremath{\mathcal{G}}},\text{\ensuremath{\mathcal{S}}}\right)$
together define a \emph{Bayesian game}.

Belief-invariant Bayes correlated equilibrium will be defined as before
for general information structures and Bayesian games. For completeness,
we will spell out the definition allowing for general information
structures, and also define Bayes Nash Equilibrium. 

For any Bayesian game $\left(\text{\ensuremath{\mathcal{G}}},\text{\ensuremath{\mathcal{S}}}\right)$,
a decision rule is a measurable map $\sigma:\Theta\times S\to\Delta(A)$,
where $A:=\prod_{i\in I}A_{i}$ and $\Delta(A)$ is endowed with the
Euclidean topology. A general information structure $\text{\ensuremath{\mathcal{S}}}$
and decision rule $\sigma$ jointly induce a measure $\sigma\circ Q\in\Delta(A\times\Theta\times S)$
in the natural way. We will be interested in \emph{outcomes} specifying
a joint distribution over actions and states $\nu\in\Delta(A\times\Theta)$.
Decision rule $\sigma$ \emph{induces} outcome $\nu_{\sigma}$ if
$\nu_{\sigma}$ is the marginal of $\sigma\circ Q$ on $A\times\Theta$.
For every player $i$, a decision rule $\sigma$ and a version of
the conditional probability $Q_{i}\colon S_{i}\to\Delta(\Theta\times S_{-i})$
of $Q$ induce a belief for every signal $s_{i}\in S_{i}$, $\sigma\circ Q_{i}\colon S_{i}\to\Delta(A\times\Theta\times S_{-i})$,
which for every measurable set $E\subseteq A\times\Theta\times S_{-i}$
and every signal $s_{i}\in S_{i}$ satisfies $\sigma\circ Q_{i}(E|s_{i})=\int_{E}\sigma(a|\theta,s_{-i},s_{i})\text{ d}Q_{i}(\theta,s_{-i}|s_{i})$.

Now we have: 
\begin{defn}
A decision rule $\sigma$ is $\varepsilon$-obedient if, for every
player $i$ , there is a version of the conditional probability $Q_{i}\colon S_{i}\to\Delta(\Theta\times S_{-i})$
so that every action $a_{i}\in A_{i}$ and deviation $a_{i}'$, 
\[
\int_{S\times\Theta}\sum_{a_{-i}\in A_{-i}}(u_{i}(a_{i},a_{-i},\theta)-u_{i}(a_{i}',a_{-i},\theta))\text{ d}\sigma\circ Q_{i}(a_{i},a_{-i},s_{-i},\theta|s_{i})>-\varepsilon,\ a.s.
\]
\end{defn}
\begin{defn}
A decision rule $\sigma$ is belief-invariant if, for every $a_{i}\in A_{i}$,
the marginal probability $\sigma({a_{i}}\times A_{-i}|(s_{i},s_{-i}),\theta)=\sigma_{i}(a_{i}|s_{i})$
does not depend on $(s_{-i},\theta)$.
\end{defn}
\begin{defn}
A decision rule $\sigma$ is an $\varepsilon$-belief-invariant Bayes
correlated equilibrium ($\varepsilon${- BIBCE}) of $\left(\text{\ensuremath{\mathcal{G}}},\mathcal{S}\right)$
if it satisfies $\varepsilon$-obedience and belief invariance.
\end{defn}
If we have a (non-redundant) information structure, these definitions
reduce to those introduced earlier. 

Now a decision rule $\sigma$ is \emph{conditionally independent}
if $\sigma(a|(s_{i},s_{-i}),\theta)=\prod_{i\in I}\sigma_{i}(a_{i}|s_{i})$,
for every $\left(a,s,\theta\right)\in A\times S\times\Theta$. Conditional
independence requires that any randomization in a player's actions
depends on their type only and is conditionally independent of others'
types and the state. If a decision rule is conditionally independent,
it gives a behavioral strategy for each player in the incomplete information
game. 
\begin{defn}
A decision rule $\sigma$ is an $\varepsilon$-Bayes Nash equilibrium
($\varepsilon$-BNE) of $(\mathcal{G},\mathcal{S})$ if it satisfies
obedience ($\varepsilon$-obedience), belief-invariance and conditional
independence.
\end{defn}
We will say that a decision rule is a BIBCE or BNE if it is a $0$-BIBCE
or $0$-BNE respectively. We will write $\mathcal{B}(\mathcal{G},\mathcal{S})$
and $\mathcal{B}^{BNE}(\mathcal{G},\mathcal{S})$ for the set of BIBCE
and BNE decision rules, and $\nu_{\sigma}$ for the outcome in $\Delta\left(A\times\Theta\right)$induced
by decision rule $\sigma$; we will write $\mathcal{O}\left(\text{\ensuremath{\mathcal{G}}},\mathcal{S}\right)$
for the set of BIBCE outcomes
\[
\mathcal{O}\left(\text{\ensuremath{\mathcal{G}}},\mathcal{S}\right):=\left\{ \nu_{\sigma}\in\Delta\left(A\times\Theta\right):\sigma\in\mathcal{B}(\mathcal{G},\mathcal{S})\right\} 
\]
 and $\mathcal{O}^{BNE}\left(\text{\ensuremath{\mathcal{G}}},\mathcal{S}\right)$
for the set BNE outcomes:
\[
\mathcal{O}^{BNE}\left(\text{\ensuremath{\mathcal{G}}},\mathcal{S}\right):=\left\{ \nu_{\sigma}\in\Delta\left(A\times\Theta\right):\sigma\in\mathcal{B}^{BNE}(\mathcal{G},\mathcal{S})\right\} 
\]

\subsection{Existence of Equilibria}

We have existence of BIBCE. 
\begin{prop}
\label{prop: existence redundant}(Existence of BIBCE) \textup{There
exists a BIBCE for every} \textup{$\left(\text{\ensuremath{\mathcal{G}}},\mathcal{S}\right)$.}
\end{prop}
This was already established for (non-redundant) information structures
in Section \ref{prop:(Existence-of-Equilibria)}. For a redundant
information structure, it is enough to find a BIBCE for its non-redundant
version and extend the BIBCE to have players ignore redundancies. 

However, strong conditions are required to ensure the existence of
BNE, and this has been one obstacle to constructing topologies on
information for BNE. \cite{milgrom2015distributional} and \cite{balder88generalized}
are apparently the best available (even if we restrict attention to
finite action games). Two sufficient conditions from \cite{milgrom2015distributional}
are important. First, existence is guaranteed if information structures
have countable support. Second, existence is guaranteed if the measure
on signals is absolutely continuous with respect to the product of
the marginal measures on individual player's signals. \cite{monderer1996proximity}
and \cite{kajii1998payoff} therefore restricted attention countable
information structures. However, existence of BNE or even $\varepsilon\text{-BNE }$fails
when these properties fail: see \cite{simon03ergodic}, \cite{hellman14no}
and \cite{simon17without} for examples. In particular, existence
is not guaranteed on (non-redundant) information structures.\footnote{\cite{vanzandt10supermodular} establishes existence of BNE on the
universal type space for supermodular games.} In the Appendix, we report an example from \cite{hellman14no} where
BNE fails to exist and report a BIBCE for that example.

\subsection{Varying Solution Concepts and Information Structures}

We first discuss the connection between BNE and BIBCE outcomes. The
next proposition states that the set of BIBCE outcomes depends only
on the non-redundant information structure: 
\begin{prop}
\label{prop: measurable}For any base game\textup{ $\mathcal{G}$}
and general information structure \textup{$\mathcal{S}$,} \textup{$\mathcal{O}\left(\text{\ensuremath{\mathcal{G}}},\mathcal{S}\right)=\mathcal{O}\left(\text{\ensuremath{\mathcal{G}}},P_{\mathcal{S}}\right)$.}
\end{prop}
This is true because any need for redundancies / correlating devices
is built into the solution concept. This observation parallels the
observation of \cite{dekel07icr} that interim correlated rationalizability
depends only on hierarchies of beliefs; \cite{liu15} showed that
(a subjective version of) BIBCE is equilibrium analogue of ICR and
thus provides a proof. \footnote{See also \cite{bergemann2017belief} for a discussion of these issues.}
For completeness, we give a proof in our notation in the Appendix.
It is immediate from the definitions that BNE outcomes are BIBCE outcomes
for any information structure.
\begin{prop}
\label{prop: BNEgen}For any game\textup{ $\mathcal{G}$} and general
information structure \textup{$\mathcal{S}$, $\mathcal{O}^{BNE}\left(\text{\ensuremath{\mathcal{G}}},\mathcal{S}\right)\subseteq\mathcal{O}\left(\text{\ensuremath{\mathcal{G}}},\mathcal{S}\right)$}.
\end{prop}
Now define the strategic distance between a pair of general information
structures $\mathcal{S}$ and $\mathcal{S}'$ to be the strategic
distance between their non-redundant representations, so $d^{**}\left(\mathcal{S},\mathcal{S}'\right):=d^{*}\left(P_{\mathcal{S}},P_{\mathcal{S}'}\right)$.

\subsection{Correlation and Non-Redundant Information Structures }

On general information structures, players' ability to correlate their
actions (using redundant correlating devices) matters for the set
of BNE. However, non-redundant information structures, as long as
there are at least two states, are very rich objects and, intuitively,
there will be plenty of opportunity to approximate arbitrary correlating
devices within them. 

Our use of BIBCE as a solution concept allowed us to focus attention
on (non-redundant) information structures and ensured existence of
equilibrium. Thus we obtained cleaner results with this solution concept.
But suppose one is interested in Bayes Nash equilibrium. In this case,
redundant types / correlating devices potentially matter for equilibrium.
And we potentially have problems with equilibrium existence. In this
section, we will argue that there is a natural approach to dealing
with these difficulties (maintaining BNE as the preferred solution
concept) and that the same almost common knowledge topology is relevant
for continuity of equilibrium outcomes. 

Our main observation observation is that, since the universal type
space is rich (as long as there at least two states), any correlating
device can be embedded in a (non-redundant) information structure
by perturbing types' first-order beliefs. The following proposition
establishes that any BIBCE on any finite information structure can
be approximated by an approximate BNE of some nearby simple information
structure. 
\begin{prop}
\label{encode} Let $|\Theta|\geq2$. For any finite, general information
structure\textup{ $\mathcal{S}$}, any \textup{$\sigma\in\mathcal{B}\left(\mathcal{G},\mathcal{S}\right)$}
and any $\varepsilon>0$, there exists (i) a simple information structure
\textup{$\mathcal{S}'$} such that \linebreak $d^{**}(\mathcal{S},\mathcal{S}')\leq\varepsilon$
; (ii) a decision rule $\sigma'$ such that (a) $\sigma'$ is a $\varepsilon$-BNE
of \textup{$\left(\mathcal{G},\mathcal{S}'\right)$ and (b) the outcome
induced by $\mathcal{S}'\circ\sigma'$ is }$\varepsilon$-close to
the outcome induced by $\mathcal{S}\circ\sigma$.
\end{prop}
\begin{proof}
Let $Q\circ\sigma\in\Delta\left(A\times T\times\Theta\right)$ be
the extended outcome corresponding to $\sigma\in\mathcal{B}(\mathcal{G},\mathcal{S})$.
Consider the (non-canonical) information structure where each player's
signal space was $S_{i}=A_{i}\times T_{i}$ and the prior was $Q\circ\sigma$.
Note that since $\sigma$ is an arbitrary BIBCE, this information
structure will in general have redundancies. But under this information
structure, there is a pure strategy BNE $\sigma'$ where each player
sets his action equal to his ``recommendation'' (i.e., the action
component of his signal): $\sigma'(a_{i}|a_{i},\tau_{i})=\boldsymbol{1}_{a_{i}}$
and so 
\[
\sum_{\theta,\tau}\Delta u_{i}(a,a_{i}',\theta)\prod_{i}\sigma'_{i}(a_{i}|a_{i},\tau_{i})\sigma(a|\theta,\tau)Q(\theta,\tau_{-i}|\tau_{i})=\sum_{\theta,\tau}\Delta u_{i}(a,a_{i}',\theta)\sigma(a|\theta,\tau)Q(\theta,\tau_{-i}|\tau_{i})\geq0.
\]

We now apply the same pertubation as in the proof of Proposition \ref{prop: denseness }:
Let $\Omega_{i,Q\circ\sigma}(\tau_{i},a_{i}):=\left\{ \left(\hat{\tau}_{i},\hat{a}_{i}\right):P\circ\sigma(\hat{\tau}_{i},\hat{a}_{i})>0,\hat{\tau}_{i}^{1}=\tau_{i}^{1}\right\} $.
Let $\rho_{i,\varepsilon}\colon\Theta\times\mathcal{T}_{i}\to\Delta(\left\{ 0,1\right\} )$
have the property that $\rho_{i,\varepsilon}(\theta,\hat{\tau}_{i},\hat{a}_{i})=\rho_{i,\varepsilon}(\theta',\hat{\tau}_{i},\hat{a}_{i})$
for all $\theta,\theta'\in\Theta$ and all $\hat{\tau}_{i},\hat{a}_{i}$
in the support of $Q\circ\sigma$ satisfying $|\Omega_{i,P\circ\sigma}(\hat{\tau}_{i},\hat{a}_{i})|=1$.
For any $\hat{\tau}_{i},\hat{a}_{i}$ in the support of $Q\circ\sigma$
satisfying $|\Omega_{i,Q\circ\sigma}(\hat{\tau}_{i},\hat{a}_{i})|>1$
let $||\rho_{i,\varepsilon}(\theta,\hat{\tau}_{i},\hat{a}_{i})-\rho_{i,\varepsilon}(\theta',\hat{\tau}_{i},\hat{a}_{i})||_{2}<\varepsilon$
so that for all distinct $\left(\tilde{\tau}_{i},\tilde{a}_{i}\right),\left(\bar{\tau}_{i},\bar{a}_{i}\right)\in\Omega_{i,P_{\delta}}(\tau_{i},a_{i})$
and $s_{i},s_{i}'\in\left\{ 0,1\right\} $,
\[
\frac{\rho_{i,\varepsilon}(s_{i}|\theta,\tilde{\tau_{i}},\tilde{a}_{i})Q\circ\sigma(\theta|\tilde{\tau}_{i},\tilde{a}_{i})}{\sum_{\hat{\theta}}\rho_{i,\varepsilon}(s_{i}|\hat{\theta},\tilde{\tau_{i}},\tilde{a}_{i})Q\circ\sigma(\hat{\theta}|\tilde{\tau}_{i},\tilde{a}_{i})}\neq\frac{\rho_{i,\varepsilon}(s_{i}'|\theta,\bar{\tau}_{i},\bar{a}_{i})Q\circ\sigma(\theta|\bar{\tau}_{i},\bar{a}_{i})}{\sum_{\hat{\theta}}\rho_{i,\varepsilon}(s_{i}'|\hat{\theta},\bar{\tau}_{i},\bar{a}_{i})Q\circ\sigma(\hat{\theta}|\bar{\tau}_{i},\bar{a}_{i})}.
\]
The prior $\hat{Q}(\theta,\tau,a,s)=Q\circ\sigma(\theta,\tau,a)\prod_{i}\rho_{i,\varepsilon}(s_{i}|\theta,\tau,a)$
induces canonical prior $Q'\in\mathcal{P}$ so that $d^{ACK}(P,P')<\varepsilon$
and $\sigma'$ induces an outcome equivalent $\varepsilon$-BNE on
$P'$. 
\end{proof}
Proposition \ref{encode} shows that that any correlation device required
for a BIBCE can be embedded in the universal state space at the expense
of $\varepsilon$-slack in the obedience constraint. Notice that under
complete information (i.e., if there was a single state), it would
not be possible to do so. Now the denseness of simple information
structures implies the immediate corollary that this is true for all
information structures (not just canonical ones). 
\begin{cor}
\label{cor:BNE1}Let $|\Theta|\geq2$. For any information structure
$P\in\mathcal{P}$, any BIBCE \textup{$\sigma\in\mathcal{B}\left(\mathcal{G},P\right)$}
and any $\varepsilon>0$, there exists (i) a simple information structure
$P'\in\mathcal{P}^{*}$ such that $d^{*}\left(P,P'\right)\leq\varepsilon$
; and (ii) a decision rule $\sigma$' such that (a) $\sigma$' is
a $\varepsilon$-BNE of \textup{$\left(\mathcal{G},P'\right)$ and
(b) the outcome induced by $P'\circ\sigma'$ is }$\varepsilon$-close
to the outcome induced by $P\circ\sigma$. 
\end{cor}
\begin{proof}
The denseness of simple information structures implies that there
exists a simple information structure $P''$ with $d^{*}\left(P,P"\right)\leq\varepsilon$.
Now the corollary follows from applying Proposition \ref{encode}
to $P''$. 
\end{proof}
Corollary \ref{cor:BNE1} implies that if we extend the notion of
approximate BNE to allow not only only slack in the obedience constraints
but also to allow nearby (in the ACK topology) information structures,
we can first establish the existence of approximate BNE and then establish
continuity of (approximate) BNE outcomes with respect to the ACK topology. 

We first define an extended notion of $\varepsilon$-BNE outcomes: 

\[
\mathcal{O}_{\varepsilon}^{BNE^{*}}\left(\text{\ensuremath{\mathcal{G}}},P\right)=\left\{ \nu\in\Delta\left(A\times\Theta\right):\exists\ P'\text{ with }d^{*}\left(P,P'\right)\ensuremath{\leq\varepsilon},\text{\ensuremath{\sigma\in\mathcal{B}_{\varepsilon}^{BNE}(\text{\ensuremath{\mathcal{G}}},P)}}\text{ s.t. }||v-v_{\sigma}||<\varepsilon\right\} 
\]

Now we have existence: 
\begin{cor}
For every game $\mathcal{G}$, prior $P\in\mathcal{P}$ and $\varepsilon>0$,
$\mathcal{O}_{\varepsilon}^{BNE}(\mathcal{G},P)\neq\emptyset$. 
\end{cor}
\begin{proof}
This result follows from Corollary \ref{cor:BNE1} and the fact that
$\mathcal{O}^{BNE}(\mathcal{G},P)\neq\emptyset$ for all finite $P\in\mathcal{P}$. 
\end{proof}
We now introduce a ``richness'' property of a base game. 
\begin{defn}
(richness) A base game $\mathcal{G}$ is \emph{rich} if for every
action profile $a\in A$, there exists $\theta_{a}\in\varTheta$ such
that for every player $i$, every $a_{i}\in A_{i}$ and every $a_{i}'\neq a_{i}$,
\[
u_{i}(a_{i},a_{-i},\theta_{a})-u_{i}(a_{i}',a_{-i},\theta_{a})>0.
\]
\end{defn}
With richness, we have continuity of approximate BNE outcomes: 
\begin{prop}
\label{prop:BNE_Main}Let $|\Theta|\geq2$. Then for every rich base
game $\mathcal{G}$ and any information structure $P\in\mathcal{P}$,
\[
\lim_{\varepsilon\downarrow0}\bigcup_{\hat{P}\in\mathcal{P}^{*}:d^{ACK}(P,\hat{P})\leq\varepsilon}\mathcal{O}^{BNE}(\mathcal{G},\hat{P})=\mathcal{O}(\mathcal{G},P).
\]
\end{prop}
\begin{proof}
Let $\hat{P}\in\mathcal{P}^{*}$ satisfy $d^{ACK}(P,\hat{P})\leq\varepsilon$.
Then by Corollary \ref{cor:BNE1}, for every $\nu\in\mathcal{O}(\mathcal{G},P)$
there exists a pure strategy $\varepsilon$-BNE, $\sigma$, of $(\mathcal{G},\hat{P})$,
so that $||\nu_{\sigma}-\nu||_{2}\leq\varepsilon$. For every $\tau$
denote the associated action recommendation by $\alpha(\theta,\tau)=\left(\alpha_{i}(\tau_{i})\right)_{i}$,
where $\sigma(\alpha(\theta,\tau)|\theta,\tau)=1$. For any choice
$\delta>0$, define the stochastic map $\rho_{\delta}\colon\Theta\times\mathcal{T}\to\Delta(A)$
\[
\rho_{\delta}(a|\theta,\tau)=\begin{cases}
1-\delta & \text{ if }a=\alpha(\theta,\tau)\\
\delta & \text{ if }\theta=\theta_{a}\\
0 & \text{otherwise.}
\end{cases}
\]
For every $(\theta,\tau,a)\in\Theta\times\mathcal{T}\times A$, let
$P_{\delta}(\theta,\tau,a):=\hat{P}(\theta,\tau)\rho_{\delta}(a|\theta,\tau)$
and note that $P_{\delta}$ has a canonical representation $\hat{P}_{\delta}\in\mathcal{P}$.
Let $J_{\mathcal{G}}:=\min_{i,a,a'}u_{i}(a_{i},a_{-i},\theta_{a})-u_{i}(a_{i}',a_{-i},\theta_{a})$.
Then for every $i$ and type $\tau_{i}$ in the support of $\hat{P}$,
and any deviation $a_{i}'$,
\[
\begin{aligned}\sum_{\theta,\tau,a}\Delta u_{i}(\alpha_{i}(\tau_{i}),\alpha_{-i}(\tau_{-i}),a_{i}',\theta)\hat{P}_{\delta}(\theta,\tau_{-i}|\tau_{i}) & >-\varepsilon(1-\delta)+\delta\sum_{\theta,\tau,a}\Delta u_{i}(\alpha_{i}(\tau_{i}),\alpha_{-i}(\tau_{-i}),a_{i}',\theta_{\alpha(\theta,\tau)})\hat{P}_{\delta}(\theta,\tau_{-i}|\tau_{i})\\
 & \geq-\varepsilon(1-\delta)+\delta J_{\mathcal{G}}
\end{aligned}
\]
Letting $\delta=\frac{\varepsilon}{J_{\mathcal{G}}+\varepsilon}$
implies that $\sigma\in\mathcal{B}_{BNE}(\mathcal{G},\hat{P}_{\delta}).$Moreover,
$d^{ACK}(\hat{P},\hat{P}_{\delta})\leq\delta$ and so
\[
\lim_{\varepsilon\downarrow0}\bigcup_{\hat{P}\in\mathcal{P}^{*}\cap\mathcal{P}:d^{ACK}(P,\hat{P})\leq\varepsilon}\mathcal{O}^{BNE}(\mathcal{G},\hat{P})\supseteq\mathcal{O}(\mathcal{G},P).
\]
 The property of upper hemi-continuity established in Proposition
\ref{prop:Properconv} readily extends to the priors in the subset
$\mathcal{P}$ and so we also have that 
\[
\lim_{\varepsilon\downarrow0}\bigcup_{\hat{P}\in\mathcal{P}^{*}\cap\mathcal{P}:d^{ACK}(P,\hat{P})\leq\varepsilon}\mathcal{O}^{BNE}(\mathcal{G},\hat{P})\subseteq\lim_{\varepsilon\downarrow0}\bigcup_{\hat{P}\in\mathcal{P}^{*}\cap\mathcal{P}:d^{ACK}(P,\hat{P})\leq\varepsilon}\mathcal{O}(\mathcal{G},\hat{P})\subseteq\mathcal{O}(\mathcal{G},P),
\]
and so the result follows.
\end{proof}

\subsection{Discussion}

We conclude this section by discussing related literature about correlating
devices and BNE. 

A number of papers have highlighted the importance of redundant types,
or correlating devices, for BNE, see for example \cite{liu09redundant}
and \cite{sadzik19revealed}.  Our approach in this section is to
observe that such correlation devices can be embedded in the universal
state space in an almost payoff-irrelevant way, so it is natural to
work with (non-redundant) information structures even if one is interested
in BNE. \cite{elpe11} propose an alternative to the standard universal
type space that embeds some correlation devices. 

In the context of complete information games, \cite{brandenburger08intrinsic}
(see also \cite{du12correlated}) asked if correlation devices (supporting
correlated equilibrium) could reflect higher-order strategic uncertainty
(in this case, they said there is intrinsic correlation) or if extrinsic
correlation is required. Their answer is that most correlated equilibria
can be explained by intrinsic correlation alone. An analogous question
(in an incomplete information context) to ask is which BIBCE could
reflect higher-order uncertainty about strategic uncertainty and payoffs.
The spirit of our results is that most BIBCE can be justified this
way. 

\cite{gossner00comparison} provided a partial order on correlating
devices (for complete information games) capturing which correlating
devices would support a larger set of correlated equilibria all games.
A natural exercise would be to define a topology on correlating devices
(generating continuity of the set of correlated equilibria) although
as far as we know that has not been done. We could imagine decomposing
a general information structure into a canonical information structure
and a correlating device and defining a topology on canonical information
structure / correlating device pairs that was sufficient for continuity
of BNE outcomes. We have not pursued this approach. 

\section{Information Design\label{sec:Information-Design}}

When studying information design problems, there will typically be
many equilibria. In formulating information design problems, one must
decide which equilibrium will be played. Two standard choices are
to assume (i) the best equilibrium for the designer is played; or
(ii) the worst equilibrium for the designer is played. In this section,
we propose a formulation of information design problems that includes
both those cases, but also allows for any continuous selection of
equilibrium. 

We will consider the following class of information design problems.
A designer has a continuous (in the Hausdorff topology) objective
function on \emph{sets} of outcomes 
\[
V\colon2^{\Delta(A\times\Theta)}\setminus\emptyset\to\mathbb{R}.
\]
Recall that $\mathcal{P}^{*}\subseteq\mathcal{P}$ denotes the collection
of simple (i.e., finite and first-order) information structures. Now
we have: 
\begin{prop}
\label{prop: information design}For any rich $\mathcal{G}$ and any
open set $\mathcal{P}'\subseteq\mathcal{P}$,
\[
\sup_{\mathcal{S}:P_{\mathcal{S}}\in\mathcal{P}'}V(\mathcal{O}^{BNE}(\mathcal{G},\mathcal{S}))=\sup_{P\in\mathcal{P}'\cap\mathcal{P}^{*}}V(\mathcal{O}^{BNE}(\mathcal{G},P))=\sup_{P\in\mathcal{P}'\cap\mathcal{P}^{*}}V(\mathcal{O}(\mathcal{G},P)).
\]
\end{prop}
Thus to choose the optimal information structure within an open set,
it is enough to focus on either BNE or BIBCE with simple information
structures. 

We will first describe how the information design problems described
above fit within this class, and then describe how the result follows
from denseness reported in Section \ref{subsec:Denseness-of-Simple}
and the BNE results reported in Section \ref{sec:Bayes-Nash-Equilibrium}. 

\subsection{Applications}

Our assumption is that the designer cares about a set of outcomes.
We are thinking that the designer cares about the outcomes consistent
with a solution concept. This sub-section spells out the leading examples
of designer objectives where the designer is interested in the best,
or the worst, or some continuous selection from the equilibrium outcomes. 

Suppose that the designer evaluates outcomes with utility function
$u:A\times\Theta\to\mathbb{R}$. 

The usual approach in information design is to assume that the designer
can choose which equilibrium is played. In this case, if $O\subseteq\Delta(A\times\Theta)$
is the set of equilibrium outcomes, the designer's utility will be:
\[
V_{MAX}\left(O\right)=\sup_{\nu\in O}\sum_{a,\theta}\nu\left(a,\theta\right)u\left(a,\theta\right)
\]
This objective is continuous and Proposition \ref{prop: information design}
applies. However, the revelation principle applies to this problem,
so we already know that we can restrict attention to finite information
structures, without appeal to Proposition \ref{prop: information design}
and the machinary behind it. 

An alternative assumption in information design is that there is ``adversarial
equilibrium selection,'' i.e., the designer expects the worst possible
equilibrium (for her) to be played. In this case the designer's utility
over sets of outcomes will be 
\[
V_{MIN}\left(O\right)=\inf_{\nu\in O}\sum_{a,\theta}\nu\left(a,\theta\right)u\left(a,\theta\right).
\]
A few papers have studied this problem is recent years (\cite{mathevet20information},
\cite{inostroza23adversarial}, \cite{moot24} and \cite{li23global}).
It is well known that the maximum is typically not attained in this
design problem. However, while it is clear in the particular settings
of these papers that the supremum can be approached using simple information
structures, there is no existing general statement of this property.
Thus Proposition \ref{prop: information design} is a useful tool
for this literature. 

More generally, we could let the designer's objective correspond to
an arbitrary continuous selection from BIBCE, so there exists $f\colon2^{\Delta(A\times\Theta)}\setminus\emptyset\to\Delta(A\times\Theta)$
with the following property: for every $O'\subseteq O\subseteq\Delta(A\times\Theta)$, 

\begin{equation}
f(O)\in O'\implies f(O')=f(O),\label{eq:Cont_V}
\end{equation}
such that 
\[
V_{f}(O)=\sum_{a,\theta}f(a,\theta|O)u(a,\theta).
\]

\subsection{Proof of Proposition \ref{prop: information design}}

Our denseness result implies:
\begin{lem}
\label{prop:BNEInfoDes}For every rich game $\mathcal{G}$ and every
information structure $P$, there exists $\varepsilon>0$ and a simple
information structure $P^{\varepsilon}\in\mathcal{P}^{*}$ so that
$d^{ACK}(P,P^{\varepsilon})<\varepsilon$ and 
\[
|V(\mathcal{O}(P,\mathcal{G}))-V(\mathcal{O}^{BNE}(P^{\varepsilon},\mathcal{G}))|<\varepsilon.
\]
\end{lem}
\begin{proof}
Suppose $\sigma\in\mathcal{B}(\mathcal{G},P)$ satisfies $V_{G}(\mathcal{O}(\mathcal{G},P))=\sum_{a,\theta}\nu_{\sigma}\left(a,\theta\right)u\left(a,\theta\right)$.
Then there are sequences $(P^{k},\varepsilon^{k})_{k}$with $P^{k}\in\mathcal{P}^{*}$
for all $k$ and $\varepsilon^{k}\downarrow0$ so that $d^{ACK}(P,P^{k})<\varepsilon_{k}$
for all $k$. Moreover, by upper hemi-continuity established in Proposition
\ref{prop:Properconv}, there is a subset $\mathcal{O}^{\infty}(\mathcal{G},P)\subseteq\mathcal{O}(P,\mathcal{G})$
so that the sequence satisfies $\lim_{k\uparrow\infty}d_{\mathcal{H},\mathcal{G}}(\mathcal{O}^{BNE}(\mathcal{G},P^{k}),\mathcal{O}^{\infty}(\mathcal{G},P))=0$.
Moreover, by the arguments in Proposition \ref{prop:BNE_Main} we
can pick the sequence so that $\nu_{\sigma}\in\mathcal{O}^{\infty}(\mathcal{G},P)$
and so by property (\ref{eq:Cont_V}) and continuity of $V$ we have
that $\lim_{k\uparrow\infty}V(\mathcal{O}^{BNE}(P^{k},\mathcal{G}))=V(\mathcal{O}(P,\mathcal{G}))$
and so the result follows. 
\end{proof}
Now the proof of Proposition \ref{prop: information design} is completed
as follows. By Lemma \ref{prop:BNEInfoDes}, we have that 
\[
\sup_{P\in\mathcal{P}'\cap\mathcal{P}^{*}}V_{G}(\mathcal{O}^{BNE}(\mathcal{G},P))=\sup_{P\in\mathcal{P}'}V_{G}(\mathcal{O}(\mathcal{G},P)).
\]
Since $\mathcal{P}^{*}\cap\mathcal{P}'\subseteq\mathcal{P}'$, we
have that 
\[
\sup_{P\in\mathcal{P}'\cap\mathcal{P}^{*}}V_{G}(\mathcal{O}^{BNE}(\mathcal{G},P))\leq\sup_{P\in\mathcal{P}'}V_{G}(\mathcal{O}^{BNE}(\mathcal{G},P)).
\]
Moreover, by property (\ref{eq:Cont_V}) and the fact that $\mathcal{O}^{BNE}(\mathcal{G},P)\subseteq\mathcal{O}(\mathcal{G},P)$,
we have that 
\[
\sup_{P\in\mathcal{P}'}V_{G}(\mathcal{O}^{BNE}(\mathcal{G},P))\leq\sup_{P\in\mathcal{P}'}V_{G}(\mathcal{O}(\mathcal{G},P))
\]
 and so 
\[
\begin{aligned}\sup_{P\in\mathcal{P}'\cap\mathcal{P}^{*}}V_{G}(\mathcal{O}^{BNE}(\mathcal{G},P)) & \leq\sup_{P\in\mathcal{P}'}V_{G}(\mathcal{O}^{BNE}(\mathcal{G},P))\\
 & \leq\sup_{P\in\mathcal{P}'}V_{G}(\mathcal{O}(\mathcal{G},P))\\
 & =\sup_{P\in\mathcal{P}'\cap\mathcal{P}^{*}}V_{G}(\mathcal{O}^{BNE}(\mathcal{G},P)).
\end{aligned}
\]

\section{\label{sec:Alternative-Formulations}Alternative Formulations and
Related Literature}

In this section, we will discuss a number of alternative topologies
characterizing convergence of strategic outcomes. One purpose in doing
so is that it will allow us to formally relate our work to the relevant
related literatures. 

\subsection{\label{subsec:Interim-based-Topology}Interim Topologies }

We have defined and characterize an (ex ante) strategic topology on
(common prior) information structures under an equilibrium solution
concept (BIBCE). By contrast, \cite{dekel2006topologies} defined
an interim strategic topology on hierarchies of beliefs under the
solution concept of interim correlated rationalizability (ICR). Two
belief hierarchies were said to be close in the interim strategic
topology if, in any game, an action that was ICR at one hierarchy
was approximately ICR at the other hierarchy. \cite{chen2017characterizing}
provide a characterization of the interim strategic topology in terms
of belief hierarchies. Crucially, the interim strategic topology imposes
restrictions on the tails of hierarchies of beliefs, unlike the product
topology. We provide a formal statement of the characterization of
\cite{chen2017characterizing} in the Appendix.

Our definition of the almost common knowledge topology used the product
topology in defining the event that interim beliefs were close. But
we noted that the use of the product topology was not essential. A
first purpose of this section is to record what properties the interim
topology must satisfy in order to induce our almost common knowledge
topology: it is enough that it is induced by a ``nice'' metric that
(1) refines the product topology; and (2) has a countable dense subset. 

The interim strategic topology satisfies these properties. But we
also establish that if we had used the interim strategic topology
as our initial interim topology, we could have dispensed with the
requirement of approximate common knowledge in our definition of the
ex ante topology. Intuitively, this is because interim strategic topology
imposes enough restrictions on the tails of hierarchies of beliefs
to generate the required approximate common knowledge. 
\begin{defn}
(Nice Interim Metric) A metric $d$ on $\Omega$ is nice if (1) its
induced topology refines the product topology; (2) there is a countable
subset of $\Omega^{0}$ which is dense in $\Omega$. 
\end{defn}
A topology on $\Omega$ is nice if it is induced by a nice metric.
Our ACK topology would be the same if we replace the product interim
topology with any nice interim topology in the definition. We used
the product topology. We could have used total variation as a metric
on interim beliefs, as \cite{kajii1998payoff} do. But we could also
have substituted the interim strategic topology as defined by \cite{dekel2006topologies}
and characterized by \cite{chen2017characterizing}. Importantly,
that topology (unlike the product topology) also imposes restrictions
on infinite tails of hierarchies of beliefs which has implications
for approximate common knowledge. In particular, the interim strategic
topology the following property.\textbf{ }
\begin{defn}
(Common Belief Invariance) A metric $d_{CB}$ on $\Omega$ satisfies
common belief invariance if for all events $E,E'\subseteq\Omega$
and every $\varepsilon>0$, 
\[
C^{1-\varepsilon}\left(\mathcal{N}_{d_{CB},\varepsilon}(E)\cap\mathcal{N}_{d_{CB},\varepsilon}(E')\right)=\mathcal{N}_{d_{CB},\varepsilon}(E)\cap\mathcal{N}_{d_{CB},\varepsilon}(E'),
\]
where $\mathcal{N}_{d_{CB},\varepsilon}(E)$ is the union of $\varepsilon$-neighborhoods
around the points in $E$ using metric $d_{CB}$. 
\end{defn}
For any metric $d:\Omega\times\Omega\to[0,\infty)$, $\varepsilon>0$
and $P\in\mathcal{P}$, let $$\text{supp}_{d,\varepsilon}(P):=\bigcup_{\omega\in\Omega:P(\mathcal{N}_{d,\varepsilon}(\omega))}\mathcal{N}_{d,\varepsilon}(\omega).$$
Now we can consider the simplest and weakest natural way of translating
an interim distance into an ex ante distance
\begin{defn}
The weak ex-ante distance induced by an interim distance $d$ is defined
as 
\[
d'(P,P')=\inf\left\{ \varepsilon\geq0:\begin{array}{c}
P\left(\hat{T}_{d,\varepsilon}(P,P')\right)>1-\varepsilon\\
\ensuremath{P'}\left(\hat{T}_{d,\varepsilon}(P,P')\right)\ensuremath{>1-}\varepsilon
\end{array}\right\} ,
\]
where $\hat{T}_{d,\varepsilon}(P,P'):=\text{supp}_{d,\varepsilon}(P)\cap\text{supp}_{d,\varepsilon}(P')$. 
\end{defn}
\begin{prop}
\label{prop:-nice-Int.Top}The weak ex ante distance induced by any
nice interim metric satisfying common belief invariance induces the
ACK topology. 
\end{prop}
\begin{proof}
Appendix.
\end{proof}
Here, we have the special property that approximate common knowledge
is for free.
\begin{prop}
\label{prop:The-ex-ante}The interim strategic topology is nice and
common belief invariant.
\end{prop}
\begin{proof}
Follows from Propositions \ref{prop:(Niceness)-The-interim} and \ref{prop:(Common-Belief-Invariance)}
in the Appendix. 
\end{proof}
For example, \cite{kajii1998payoff} define a topology on ex ante
information structures (discussed below) but use total variation as
a metric on interim beliefs.

\subsection{\label{subsec:Value-Based-Topology}Value-Based Topology}

We could alternatively define our topology in terms of convergence
of the ex ante expected equilibrium payoffs rather than equilibrium
outcomes. This was the approach of \cite{monderer1996proximity} and
\cite{kajii1998payoff} and also the recent work of \cite{gensbittel22value}
for zero sum games. This distinction is not important for our results. 

Let $V$$(G,P)$ be the set of ex ante utilities of players (in $\mathbb{R}^{n}$)
that can arise from some BIBCE of $(G,P)$, and let $V_{\varepsilon}$$(G,P)$
be the set of ex ante utilities of players (in $\mathbb{R}^{n}$)
that are within $\varepsilon$ of some $\varepsilon$-BIBCE of $(G,P)$.
We can say that $P,P'$ are $\varepsilon$-value close in game $\mathcal{G}$
if $V(\mathcal{G},P)\subseteq V_{\varepsilon}(G,P')$ and $V(\mathcal{G},P')\subseteq V_{\varepsilon}(G,P)$.
\begin{defn}
Let $d^{V}(P,P'|\mathcal{G}$) be the infimum of the set of $\varepsilon$
such that $P\text{ and }P'$ are $\varepsilon$-value close in game
$\mathcal{G}$.
\end{defn}
\begin{lem}
Now $d^{*}(P^{k},P|\mathcal{G})\rightarrow P$ if and only if \textup{$d^{V}(P^{k},P|\mathcal{G})\rightarrow0\text{}$
for all $\mathcal{G}$.}
\end{lem}
\begin{proof}
One direction is immediate, because convergence of outcomes implies
converges of values. In the other direction, it is enough to change
payoffs so differences in outcome translate into large differences
in payoffs. Consider two outcomes $\nu,\nu'$ of $P$ and $P'$ respectively
so that $V_{i}(\nu)=V_{i}(\nu')$ but $\nu\neq\nu'$. Consider the
augmented payoffs $u_{i}(a,\theta)+h_{i}(a_{-i},\theta)$ and associated
values $V_{i}^{h}(\nu)=V_{i}(\nu)+\sum_{a,\theta}\nu(a,\theta)h_{i}(a_{-i},\theta)$
and $V_{i}^{h}(\nu')=V_{i}(\nu')+\sum_{a,\theta}\nu'(a,\theta)h_{i}(a_{-i},\theta)$.
Hence 
\[
V_{i}^{h}(\nu)-V_{i}^{h}(\nu')=\sum_{a,\theta}h_{i}(a_{-i},\theta)(v(a,\theta)-v'(a,\theta))
\]
Consider the choice $h_{i}(a_{-i},\theta)=\boldsymbol{1}_{\left\{ (a_{-i},\theta):\exists\ a_{i}\in A_{i}\text{ s.t. }v(a,\theta)>v'(a,\theta)\right\} }-\boldsymbol{1}_{\left\{ (a_{-i},\theta):\exists\ a_{i}\in A_{i}\text{ s.t. }v(a,\theta)<v'(a,\theta)\right\} }$
and so $V_{i}^{h}(\nu)-V_{i}^{h}(\nu')>0$. 
\end{proof}
Thus there is little difference working with outcome-based strategic
topologies and value-based strategic topologies. 

However, in the case of zero-sum games, the value is uniquely defined
although many outcomes might give rise to the same value. So it is
convenient and natural to work with value-based strategic topologies
in that case. \cite{peski08comparison} and \cite{gossner20value}
characterize changes in information structure that increase one player's
payoff in all zero sum games. \cite{gensbittel22value} study essentially
the value-based strategic topology defined above but restricted to
zero-sum games. 

\subsection{Join Measurability\label{subsec:Countable-Type-Spaces}}

Our approach in this paper has been to fix a set of payoff-relevant
states $\Theta$, and look at common prior information structures
that describe beliefs and higher-order beliefs about those states.
Then we characterize the coarsest topology under which equilibrium
outcomes converge for any game where payoffs are measurable with respect
to $\Theta$. In particular, we do not allow games to depend in an
arbitrary way on players' types (or signals). 

An alternative approach would be to allow any game where payoffs were
measurable with respect to the join of players' types. Equivalently,
we could restrict attention to information structures where each payoff-relevant
state could arise under only one profile of types; we call this a
``join measurability'' restriction on information structures. It
was implicitly maintained in the works of \cite{monderer1996proximity}
and \cite{kajii1998payoff}. This restriction greatly simplifies the
arguments. In particular, in the proofs of sufficiency analogous to
Proposition \ref{prop: sufficiency} the join measurability approach
allows a straightforward mapping of a strategy profile on one information
structure to another. We were not able to do that, and required a
continuous extension exploiting the structure of the universal type
space. In the proofs of necessity analogous to Proposition \ref{prop: necessity-1},
the join measurability approach requires only an email game like component
where an infection argument operates and not also an iterated scoring
rule to reveal finite levels of beliefs, as in this paper. 

This is the most important difference between the work of \cite{monderer1996proximity}
and \cite{kajii1998payoff}, and this work. There are number of other
differences that are less important. First, the earlier papers focused
on BNE as a solution concept, while we focus on BIBCE. Second, they
focused on countable information structures (to ensure existence of
BNE), whereas we do not impose that restriction. Third, their topologies
were value-based whereas our topology is outcome-based (a difference
that we argued was not important in the previous section) . Fourth,
we restrict ourselves to a finite set of payoff-relevant states, but
these papers must allow for countable payoff relevant states. 

The set of information structures considered in \cite{monderer1996proximity}
and \cite{kajii1998payoff}, while both satisfying join measurability,
were modelled differently. \cite{monderer1996proximity} fixed a state
space and prior probability. An information structure was then a profile
of (countable) partitions of the state space. And payoffs could depend
in arbitrary ways on the state space. On the other hand, \cite{kajii1998payoff}
fixed a countable set of ``types'' (or labels) for each player and
allowed arbitrary probability distributions on the types space. The
exact connection between the similar topologies defined on different
classes of information structures was not known until recently, when
\cite{kambhampati23payoffcontinuity} showed an equivalence between
the results.

\subsection{\label{subsec:Improper-Priors-and}Improper Priors and Completeness}

In this paper, we have focused on common prior information structures.
Our results imply that equilibrium outcomes converge along Cauchy
sequences in the ACK topology. However, in general, Cauchy sequences
may not have well-defined limits within the space of information structures.
In this section, we show that if we enrich the class of information
structures to include ``improper'' common prior information structures,
and extend the ACK topology to this class of information structures,
then all Cauchy sequences do converge to a well defined limit. This
result is of independent interest, in the light of the importance
of improper common prior limits in the literature on higher-order
beliefs in games (discussed below).

An improper prior on the universal type space is simply a measure
with perhaps infinite mass such that there is a conditional probability
consistent with the interim beliefs on the universal type space. Formally,
we have:
\begin{defn}
(Improper Prior) A measure on $\Omega$, $Q\colon\mathscr{B}\to[0,\infty]$,
is an improper prior if for every player $i$ there is a measurable
map $Q_{i}\colon\mathcal{T}_{i}\to\Delta(\Omega)$ so that
\[
\tau_{i}^{*}=\text{marg}_{\Theta\times\mathcal{T}_{-i}}\left(Q_{i}(\tau_{i})\right),\ Q\text{-a.e.}
\]
and for every measurable $E\in\mathscr{B}$,
\[
\int_{\text{proj}_{\mathcal{T}_{i}}\left(E\right)}Q_{i}(E|\tau_{i})Q(\text{d}\tau_{i})=Q(E).
\]
\end{defn}
Let $\bar{\mathcal{P}}$ denote the set of improper priors and note
that $\mathcal{P}\subseteq\bar{\mathcal{P}}.$ The approximate common
knowledge topology extends in a natural way to $\mathcal{\bar{\mathcal{P}}}$
:
\begin{defn}
\label{def:ACK-1}(Approximate Common Knowledge Distance) For every
$P,P'\in\bar{\mathcal{P}}$, let 

\[
d^{ACK}\left(P,P'\right):=\inf\left\{ \varepsilon\geq0:\begin{array}{c}
P\left(\Omega\setminus C^{1-\varepsilon}\left(\hat{T}_{\varepsilon}(P,P')\right)\right)<\varepsilon P\left(C^{1-\varepsilon}\left(\hat{T}_{\varepsilon}(P,P')\right)\right)\\
P'\left(\Omega\setminus C^{1-\varepsilon}\left(\hat{T}_{\varepsilon}(P,P')\right)\right)<\varepsilon P'\left(C^{1-\varepsilon}\left(\hat{T}_{\varepsilon}(P,P')\right)\right)
\end{array}\right\} .
\]
\end{defn}
Notice that this Definition coincides with the earlier Definition
\ref{def:ACK} when applied to proper priors. 

A decision rule $\sigma$ is a BIBCE of an improper prior $Q\in\mathcal{Q}$,
if $\sigma$ is belief invariant and obedience holds almost everywhere.
For any $\varepsilon>0$, let the collection of $\varepsilon$-BIBCE
on an improper prior $Q$ be denoted by $\overline{\mathcal{B}}^{\varepsilon}(\mathcal{G},Q)$. 
\begin{prop}
The extended approximate common knowledge topology on $\overline{\mathcal{P}}$
is complete. 
\end{prop}
\begin{proof}
Let $\mathscr{S}\subseteq\mathscr{B}$ be a semi-ring so that the
sigma algebra it generates equals $\mathscr{B}$. Let $(P^{k})_{k}$
be a Cauchy sequence in $\mathcal{P}$ and $(\varepsilon_{k})_{k}$
a sequence so that for every $k$,
\[
\varepsilon_{k}\leq\sup_{h>k}d^{ACK}(P^{k},P^{h}).
\]
Define a limit pre-measure on $\mathscr{S}$, which on any event $E\in\mathscr{S}$
takes the form
\[
\xi_{0}(E):=\lim_{k\uparrow\infty}P^{k}(\mathcal{N}_{\varepsilon_{k}}(E_{n})),
\]
where we set $\mathcal{N}_{\varepsilon}(\varnothing)=\varnothing$.
We now show that $\xi_{0}$ is a pre-measure on $\mathscr{S}$. Indeed,
$\xi_{0}(\varnothing)=0$. Since each $P^{k}$ is a measure, they
are all finitely additive and so for any finite, disjoint $(E_{n})_{n\leq N}$,
the limit and sum can be exchanged so that 
\[
\begin{split}\xi_{0}\left(\bigcup_{n\leq N}E_{n}\right) & =\lim_{k\uparrow\infty}\sum_{n\leq N}P^{k}(\mathcal{N}_{\varepsilon_{k}}(E_{n}))\\
 & =\sum_{n\leq N}\xi_{0}(E_{n})
\end{split}
\]
Finally, countable monotonicity follows from Reverse Fatou's lemma
(each $P^{k}$ is bounded and non-negative): 
\[
\begin{split}\xi_{0}\left(\bigcup_{n\in\mathbb{N}}E_{n}\right) & =\lim_{k\uparrow\infty}\sum_{n\in\mathbb{N}}P^{k}(\mathcal{N}_{\varepsilon^{k}}(E_{n}))\\
 & \leq\sum_{n\in\mathbb{N}}\xi_{0}(E_{n}).
\end{split}
\]
By Caratheodory-Hahn's Extension Theorem, $\xi_{0}$ extends to a
measure $\xi$ on the sigma algebra generated by $\mathscr{S}$. It
remains to show that $\xi_{0}$ is a Canonical Improper Prior. We
need to prove that $\xi_{0}$ satisfies the consistency condition.
For every event $E\in\mathscr{B}$, we have that 
\[
\begin{split}|P^{k}(E)-\xi(E)| & =\bigg|\int_{\mathcal{T}_{i}}\tau_{i}(\mathcal{N}_{\varepsilon}(E))P^{k}(\text{d}\tau_{i})-\lim_{k\uparrow\infty}\int_{\mathcal{T}_{i}}\tau_{i}(\mathcal{N}_{\varepsilon_{k}}(E))P^{k}(\text{d}\tau_{i})\bigg|\\
 & \leq\varepsilon.
\end{split}
\]
Hence $\xi_{0}\in\overline{\mathcal{P}}$. Deduce that $\overline{\mathcal{P}}$
is a complete metric space. 
\end{proof}

\section{Appendix}

\subsection{Proof of Proposition \ref{prop:(Existence-of-Equilibria)}}

In the set-up of \cite{stinchcombe2011correlated} a prior is a countably
additive probability $P\in\Delta(\prod_{i}\Omega_{i}),$ where each
$\Omega_{i}$ is endowed with a sigma algebra $\mathcal{F}_{i}$ and
$P$ is defined on a sigma algebra containing the product sigma algebra.
Payoffs are given by $U_{i}\colon\prod_{i}\Omega_{i}\to\mathbb{R}^{A}$.
It is assumed that $\int_{\Omega}||U_{i}(\omega)||_{\infty}dP(\omega)<\infty$.
A behavior strategy is a $\mathcal{F}_{i}$-measurable map $b_{i}\colon\Omega_{i}\to\Delta(A_{i})$.
Let $\mathbb{B}_{i}$ be the set of behavior strategies and let $\mathbb{B}:=\prod_{i}\mathbb{B}_{i}$.
A measure $\nu\in\Delta(\prod_{i}\mathbb{B}_{i})$ is a correlated
equilibrium\footnote{In \cite{stinchcombe2011correlated} this corresponds to a variation
Strategy Correlated Equilibrium where deviation strategies depend
action recommendations at realized types only.} if for every $i$ and measurable deviation $\gamma_{i}\colon\Omega_{i}\times\Delta\left(A_{i}\right)\to\Delta\left(A_{i}\right)$,
\[
\int_{\mathbb{B}}\int_{\Omega}\left(\sum_{a\in A}U_{i}(\omega)(a)\prod_{i}b_{i}(a|\omega)-\sum_{a\in A}U_{i}(\omega)(a)\cdot\gamma_{i}(\omega,b_{i}(a|\omega))\cdot\prod_{j\neq i}b_{j}(a|\omega)\right)dP(\omega)d\nu(b)\geq0.
\]

We show that for every correlated equilibrium $(P,\nu)$ there is
an outcome equivalent equilibrium $(P,\sigma)$. We now translate
this set-up into ours for any choice of general information structure
$\mathcal{S}$: $\Omega_{i}$ must correspond to $\Theta\times T_{i}$,
where $\mathcal{F}_{i}$ is generated by the projection of $\Theta\times T_{i}$
onto $T_{i}$ . Payoffs translate into our set-up by setting $U_{i}(\theta,\tau)(a)=u_{i}(\theta,a)$.
The condition that $\int_{\Omega}||U_{i}(\omega)||_{\infty}dP(\omega)<\infty$
is then satisfied since $\Theta$ is finite and payoffs depend $\Theta$
and not $T$. A behavior strategy corresponds to a marginal decision
rule, $\sigma(a_{i}|t_{i}).$It remains to show that a distribution
$\nu$ over profiles of marginal decision rules as in \cite{stinchcombe2011correlated}
is equivalent to a belief invariant decision rule. First, note that
every measure $\nu\in\Delta(\prod_{i}\mathbb{B}_{i})$ induces a belief
invariant decision rule $\sigma_{\nu}$ where for every measurable
event $E\subseteq A$
\[
\sigma_{\nu}(a|\omega)=\int_{\mathbb{B}}\prod_{i}b_{i}(\tau_{i})(a_{i})\ d\nu(b).
\]
Indeed, the resulting marginal probability $\sigma_{\nu}(a_{i}|\omega)=\int_{\mathbb{B}_{i}}b_{i}(\tau_{i})(a_{i})\ \text{d}\nu(b_{i})$
verifies belief invariance.

\subsection{Proof of Proposition \ref{prop: measurable}}

Let $(\mathcal{G},P)$ be base game and an information structure,
i.e. so that $P\in\mathcal{P}$. Consider an information structure
$\hat{\mathcal{S}}=((S_{i},\hat{P}_{i})_{i\in I},\hat{P})$ satisfying
$P=P_{\hat{\mathcal{S}}}$ and let $\sigma'$ be a BIBCE of $(\mathcal{G},\hat{\mathcal{S}})$.
We construct an outcome equivalent BIBCE $\sigma$ of $(\mathcal{G},\hat{\mathcal{S}})$.
For every $\theta\in\Theta$ and $P$-almost every $\tau\in\mathcal{T}$,
let 
\[
\sigma(a|\theta,\tau)=\int_{\overline{\tau}^{-1}(\tau)}\sigma'(a|\theta,s')\text{ d}\hat{P}(s'|\theta,\overline{\tau}^{-1}(\tau)),
\]
where $\tau\mapsto\hat{P}(\cdot|\theta,\overline{\tau}^{-1}(\tau))\in\Delta(\overline{\tau}^{-1}(\tau))$
is a conditional probability on $S$. Then obedience constraints for
any $P$-almost every type $\tau_{i}\in\mathcal{T}_{i}$ are satisfied

\begin{multline*}
\sum_{\theta,a_{-i}}(u_{i}(a,\theta)-u_{i}(a_{i}',a_{-i},\theta))\int_{S_{-i}}\sigma(a|\theta,\tau)\text{ d}P(\tau_{-i},\theta|\tau_{i})=\\
\sum_{\theta,a_{-i}}(u_{i}(a,\theta)-u_{i}(a_{i}',a_{-i},\theta))\int_{S_{-i}}\int_{\overline{\tau}^{-1}(\tau)}\sigma'(a|\theta,s')\text{ d}\hat{P}(s'|\theta,\overline{\tau}^{-1}(\tau))\text{ d}P(\tau_{-i},\theta|\tau_{i}).
\end{multline*}
Noting that 
\[
\int_{\mathcal{T}_{-i}}\int_{\overline{\tau}^{-1}(\tau)}\sigma'(a|\theta,s')\text{ d}\hat{P}(s'|\theta,\overline{\tau}^{-1}(\tau))\text{ d}P(\tau_{-i},\theta|\tau_{i})=\int_{S_{-i}}\sigma'(a|\theta,s')\text{ d}\hat{P}(s_{-i}',\theta|s_{i}),
\]
we conclude that $\sigma$ is a BIBCE of $(\mathcal{G},P)$. Conversely,
note that every BIBCE $\sigma$ of $(\mathcal{G},P)$ induces a $id\times\overline{\tau}$-measurable
decision rule $\sigma'$, where for every action profile $a\in A$,
state $\theta\in\Theta$ and $\hat{P}$-almost every $s'\in S$,
\[
\sigma'(a|\theta,s')=\sigma(a|\theta,\overline{\tau}(s')).
\]
Performing a change of variables 
\begin{multline*}
\int_{S_{-i}}\sigma'(a|\theta,s')\text{ d}\hat{P}(s_{-i}',\theta|s_{i})=\int_{\mathcal{T}_{-i}}\int_{\overline{\tau}^{-1}(\tau)}\sigma'(a|\theta,s')\text{ d}\hat{P}(s'|\theta,\overline{\tau}^{-1}(\tau))\text{ d}P(\tau_{-i},\theta|\tau_{i})\\
=\int_{\mathcal{T}_{-i}}\sigma(a|\theta,\tau)\text{ d}P(\tau_{-i},\theta|\tau_{i}).
\end{multline*}
and so obedience follows. 

\subsection{Proof of Proposition \ref{prop: BNEgen}}

Let $(\sigma_{i}\colon S_{i}\to\Delta(A_{i}))_{i\in I}$ be a BNE
under $(\mathcal{G},\mathcal{S})$. Consider the following representation
of $\mathcal{S}$: Define the graph of $\overline{\tau},$ $\phi_{\overline{\tau}}\colon s\mapsto(\overline{\tau}(s),s)$
and consider the push forward probability $P^{*}:=P\circ(id\times\phi_{\overline{\tau}})^{-1}$.
Let $\sigma^{R}\colon\Theta\times\mathcal{T}\to\Delta(S)$ denote
a ($id\times\overline{\tau}$)-conditional probability of $P^{*}$
so that for any measurable $E\subseteq S$ and state $\theta\in\Theta$
in the support of $P$, 
\[
P(E|\theta)=\int_{\mathcal{T}}\sigma^{R}(E|\theta,\tau)\text{d}P^{*}(\tau|\theta).
\]
We now construct a BIBCE $\sigma^{*}$ on $\text{marg}_{\Omega}P^{*}$
as follows: 
\[
\sigma^{*}(a|\theta,\tau)=\int_{S}\prod_{i\in I}\sigma_{i}(a|s_{i})\text{ d}\sigma^{R}(s|\theta,\tau).
\]
Since $\sigma$ satisfies obedience, so does $\sigma^{*}$:

\begin{multline*}
\sum_{\theta,a_{-i}}(u_{i}(a,\theta)-u_{i}(a_{i}',a_{-i},\theta))\int_{\mathcal{T}_{-i}}\sigma^{*}(a|\theta,\tau)\text{ d}\text{marg}_{\Omega}P^{*}(\theta,\tau_{-i}|\tau_{i})\\
=\sum_{\theta,a_{-i}}(u_{i}(a,\theta)-u_{i}(a_{i}',a_{-i},\theta))\int_{S}\prod_{i\in I}\sigma_{i}(a|s_{i})\text{ d}P(\theta,s_{-i}|s_{i}).
\end{multline*}

\subsection{Proof of Proposition \ref{prop:The-ACK-topology}}

The ACK topology is metrizable by the metric
\[
d^{f,g}(P,P')=\inf\left\{ \varepsilon\geq0:\begin{array}{c}
P\left(C^{1-g(\varepsilon)}\left(\hat{T}_{g(\varepsilon)}(P,P')\right)\right)>f(\varepsilon)\\
P'\left(C^{1-g(\varepsilon)}\left(\hat{T}_{g(\varepsilon)}(P,P')\right)\right)>f(\varepsilon)
\end{array}\right\} .
\]
where $g(\varepsilon):=\varepsilon^{2}$ and $f(\varepsilon):=1-\frac{1}{2}\varepsilon(\varepsilon+1)$. 
\begin{proof}
First note that $d^{f,g}(P,P)=0$ and $d^{f,g}(P,P')=d^{f,g}(P',P)$
are both immediate. Suppose now that $d^{f,g}(P,P')=0$. Then we have
that $\text{supp}(P)=\text{supp}(P')$, which by the definition of
$\mathcal{P}$ means that $P=P'$. It remains to show that $d^{f,g}$
satisfies the triangle inequality. Let $P^{1},P^{2},P^{3}\in\mathcal{P}$
satisfy $d^{f,g}(P^{1},P^{2})<\varepsilon_{1,2}$ and $d^{f,g}(P^{3},P^{2})<\varepsilon_{3,2}$,
for $\varepsilon_{1,2},\varepsilon_{3,2}\in[0,1]$. Note that for
every $\omega\in\hat{T}_{g(\varepsilon_{1,2})}(P^{1},P^{2})\cap\hat{T}_{g(\varepsilon_{2,3})}(P^{2},P^{3})\subseteq\text{supp}_{g(\varepsilon_{1,2})}(P^{1})\cap\text{supp}_{g(\varepsilon_{3,2})}(P^{3})$
there is $\omega'\in\text{supp}(P^{2})$ so that 
\[
d_{\Pi}(\omega,\omega')\leq\underline{\varepsilon}:=\min\left\{ g(\varepsilon_{1,2}),g(\varepsilon_{3,2})\right\} .
\]
Moreover, $P^{2}(F)>1-\left(1-f(\varepsilon_{1,2})+1-f(\varepsilon_{3,2})\right)$,
where
\[
F:=C^{1-g(\varepsilon_{1,2})}\left(\hat{T}_{g(\varepsilon_{1,2})}(P^{1},P^{2})\right)\cap C^{1-g(\varepsilon_{3,2})}\left(\hat{T}_{g(\varepsilon_{2,3})}(P^{2},P^{3})\right)\subseteq C^{1-g(\varepsilon_{1,2}+\varepsilon_{3,2})}\left(\hat{T}_{g(\varepsilon_{1,2}+\varepsilon_{3,2})}(P^{1},P^{3})\right),
\]
and the last containment is due to the fact that $g(\varepsilon_{1,2})+g(\varepsilon_{3,2})=\varepsilon_{1,2}^{2}+\varepsilon_{3,2}^{2}\leq(\varepsilon_{1,2}+\varepsilon_{3,2})^{2}=g(\varepsilon_{1,2}+\varepsilon_{3,2})$.
Consider a $P^{1},P^{3}$ and $P^{2}$ measurable, surjective map
$\phi\colon F\to F\cap\text{supp}(P^{2})$ so that for all $\omega\in F$,
$d_{\Pi}(\phi(\omega),\omega)<\underline{\varepsilon}$. Letting $F_{i}:=\text{proj}_{\mathcal{T}_{i}}(F)$,
\[
\begin{aligned}|P^{k}(F)-P^{2}(F)| & =\bigg|\int_{F_{i}}\tau_{i}(F)P^{k}(\text{d}\tau_{i})-\int_{F_{i}}\tau_{i}(F)P^{2}(\text{d}\tau_{i})\bigg|\\
 & \leq\underline{\varepsilon}
\end{aligned}
.
\]
Hence $P^{k}(F)\geq f(\varepsilon_{1,2})+f(\varepsilon_{3,2})-1-\underline{\varepsilon}$.
Then for any $\varepsilon_{1,2},\varepsilon_{3,2}\in[0,1]$,
\[
\begin{aligned}P^{k}\left(C^{1-g(\varepsilon_{1,2}+\varepsilon_{3,2})}\left(\hat{T}_{g(\varepsilon_{1,2}+\varepsilon_{3,2})}(P^{1},P^{3})\right)\right) & \geq f(\varepsilon_{1,2})+f(\varepsilon_{3,2})-1-\underline{\varepsilon}\end{aligned}
.
\]
Given $g(\varepsilon)=\varepsilon^{2}$ and $f(\varepsilon)=1-\frac{1}{2}\varepsilon(\varepsilon+1)$.
Then we have 
\[
\begin{aligned}f(\varepsilon_{1,2})+f(\varepsilon_{3,2})-\underline{\varepsilon}-1 & \geq f(\varepsilon_{1,2}+\varepsilon_{3,2})\\
1-\frac{1}{2}\varepsilon_{1,2}(\varepsilon_{1,2}+1)+1-\frac{1}{2}\varepsilon_{3,2}(\varepsilon_{3,2}+1)-\underline{\varepsilon}-1 & \geq1-\frac{1}{2}\left(\varepsilon_{1,2}+\varepsilon_{3,2}\right)(\varepsilon_{1,2}+\varepsilon_{3,2}+1)\\
-\frac{1}{2}\varepsilon_{1,2}(\varepsilon_{1,2}+1)-\frac{1}{2}\varepsilon_{3,2}(\varepsilon_{3,2}+1)-\underline{\varepsilon} & \geq-\frac{1}{2}\left(\varepsilon_{1,2}+\varepsilon_{3,2}\right)(\varepsilon_{1,2}+\varepsilon_{3,2}+1)\\
(\varepsilon_{1,2}^{2}+\varepsilon_{3,2}^{2}+\varepsilon_{1,2}+\varepsilon_{3,2})+2\underline{\varepsilon} & \leq\left(\left(\varepsilon_{1,2}+\varepsilon_{3,2}\right)^{2}+\varepsilon_{1,2}+\varepsilon_{3,2}\right)\\
(\varepsilon_{1,2}^{2}+\varepsilon_{3,2}^{2}+\varepsilon_{1,2}+\varepsilon_{3,2})+2\underline{\varepsilon} & \leq\left(\varepsilon_{1,2}^{2}+\varepsilon_{3,2}^{2}+\varepsilon_{1,2}+\varepsilon_{3,2}\right)+2\varepsilon_{1,2}\varepsilon_{3,2}\\
\min_{k\in\left\{ 1,3\right\} }\varepsilon_{k,2}^{2} & \leq\varepsilon_{1,2}\varepsilon_{3,2}\\
\min_{k\in\left\{ 1,3\right\} }\varepsilon_{k,2} & \leq\max_{k\in\left\{ 1,3\right\} }\varepsilon_{k,2}.
\end{aligned}
\]
Hence $d^{f,g}(P^{1},P^{3})<\varepsilon_{1,2}+\varepsilon_{3,2}$
and so $d^{f,g}$ satisfies the triangle inequality. 
\end{proof}

\subsection{Proof of Lemma \ref{claim:GridPart} }
\begin{proof}
For every precision $z\in\mathbb{N}$ define a finite grid approximation
recursively: Consider the grid on first order beliefs, given by 
\[
A_{i}^{1,z}:=\left\{ a\in\mathbb{R}^{\Theta}:a_{\theta}\in\left\{ \frac{n}{z}:0\leq n\leq z\right\} ,\sum_{\theta}a_{\theta}=1\right\} 
\]
 for any player $i$. Given a finite set $A_{i}^{m-1,z}\subseteq\mathcal{T}_{i}^{m-1}$
for every $i$, define 
\begin{equation}
A_{i}^{m,z}:=\left\{ a\in\mathbb{R}^{\prod_{n<m}A_{-i}^{n,z}\times\Theta}:\begin{array}{c}
a_{a^{1},\dots,a^{m-1},\theta}\in\left\{ \frac{n}{z}:0\leq n\leq z\right\} ,\\
\sum_{a^{1},\dots,a^{m-1},\theta}a_{a^{1},\dots,a^{m-1},\theta}=1
\end{array}\right\} .\label{eq:Grid}
\end{equation}
Then for any $m\in\mathbb{N}$, and any $t^{m}\in\mathcal{T}^{m},$there
is $a^{m}\in\prod_{i\in I}A_{i}^{m,z}$ so that $d_{Weak}^{m}(t^{m},a^{m})<1/z$.
Then there is $m,z$ large enough so that for every, $\theta\in\Theta$,
every $g\in\mathcal{T}$satisfying $\left(g^{1},\dots g^{m}\right)=a^{m}$
and $\tau\in\mathcal{T}$satisfying $(\tau^{1},\dots,\tau^{m})=t^{m},$
\[
\begin{split}d_{\Pi}(\left(\theta,\tau\right),\left(\theta,g\right))= & \sum_{n=1}^{\infty}\eta^{n}d_{Weak}^{n}(\tau^{1},\dots,\tau^{n},g^{1},\dots,g^{n})\\
= & \frac{\eta}{z}\frac{1-\eta^{m+1}}{1-\eta}+\eta^{m}\frac{1}{1-\eta}\\
< & \varepsilon.
\end{split}
\]
Let $G_{\varepsilon}:=\prod_{i\in I}A_{i}^{m,z}\cap(\Omega\setminus\mathcal{N}_{\varepsilon}(E))$.
Suppose that $d_{\Pi}(\omega,\omega_{g})<\varepsilon$ for some $\omega\in\Omega$
and $\omega_{g}\in G_{\varepsilon}.$Then clearly, it must be that
$\omega\notin E$, since the opposite would imply that $\omega_{g}\in\mathcal{N}_{\varepsilon}(E)\cap G_{\varepsilon}$.
Suppose now that $\omega\in\Omega\setminus\mathcal{N}_{\varepsilon}(E)$
and so by the above there is $\omega_{g}\in G_{\varepsilon}$ so that
$d_{\Pi}(\omega',\omega_{g})<\varepsilon$.
\end{proof}

\subsection{Proof of Lemma \ref{claim_GridGame} (Iterative Scoring Rule)}
\begin{proof}
Fix $m,z$. and let $A^{m,z}$ be defined as in \ref{eq:Grid}. Following
\cite{dekel2006topologies}, for every profile $\left(a^{1},\dots,a^{m},\theta\right)\in A^{m,z}\times\Theta$
we define
\begin{multline*}
u_{i}^{m}\left(a^{1},\dots,a^{m},\theta\right):=2a_{i}^{1}(\theta)-\sum_{\hat{\theta}}\left(a_{i}^{1}(\hat{\theta})\right)^{2}\\
+\sum_{n=2}^{m}\left(2a_{i}^{n}(a^{1},\dots,a^{n-1},\theta)-\sum_{\hat{a}^{1},\dots,\hat{a}^{n-1},\hat{\theta}}\left(a_{i}^{n}(\hat{a}^{1},\dots,\hat{a}^{n-1},\hat{\theta})\right)^{2}\right).
\end{multline*}
 \cite{dekel2006topologies} show that the game where payoffs are
given by $u_{i}^{m},$ the uniquely interim (correlated) rationalizable
action profile of $\tau$ is the profile $a^{1},\dots,a^{m}$ which
is $1/z$-close to $\tau$ in terms of distance $d_{Weak}^{m}$ on
$\mathcal{T}^{m}$. The result then follows from Claim \ref{claim:GridPart}. 

\newpage{}
\end{proof}

\subsection{Open Set Definition of Strategic Topology}
\begin{defn}
\label{def:(Strategic-Distance)}(Strategic Distance) A function $\overline{d}_{\mathcal{P}}\colon\mathcal{P}\times\mathcal{P}\to[0,\infty)$
is a strategic distance if,
\end{defn}
\begin{enumerate}
\item for every game $\mathcal{G}$ and $\varepsilon>0$, there exists $\delta_{\varepsilon,\mathcal{G}}>0$
so that for all $P,P'\in\mathcal{P}$ satisfying $\overline{d}_{\mathcal{P}}(P,P')<\delta_{\varepsilon,u}$,
we have that $d^{u}(P,P')<\varepsilon$,
\item for every $\varepsilon>0$ and $P,P'\in\mathcal{P}$ satisfying $\sup_{u}d^{u}(P,P')<\varepsilon$,
we also have $\overline{d}_{\mathcal{P}}(P,P')<\varepsilon$.
\end{enumerate}

\subsection{Proof of Proposition \ref{prop:-nice-Int.Top}\protect 
}\begin{proof}
Let $d$ induce a nice, common belief invariant topology on $\Omega$
which is a stronger topology than the one induced by $d_{\Pi}$ and
we have that $\hat{T}_{d,\varepsilon}(P,P')\subseteq\hat{T}_{\varepsilon}(P,P')$.
By Proposition \ref{prop_CPB_free} we deduce that $\hat{T}_{d,\varepsilon}(P,P')=C^{1-\varepsilon}\left(\hat{T}_{d,\varepsilon}(P,P')\right)\subseteq C^{1-\varepsilon}\left(\hat{T}_{\varepsilon}(P,P')\right)$
and so 
\[
\begin{aligned}\min\left\{ P(\hat{T}_{d,\varepsilon}(P,P')),P(\hat{T}_{d,\varepsilon}(P,P'))\right\} >1-\varepsilon & \implies d^{ACK}(P,P')<\varepsilon\\
 & \implies\forall\ \mathcal{G},\ \exists\ N>0\text{ s.t. }d^{*}(P,P'|\mathcal{G})<N\varepsilon.
\end{aligned}
\]
We now show that for every $\varepsilon>0$ there is $\delta>0$ so
that $P(C^{1-\delta}\left(\hat{T}_{\delta}(P,P')\right))>1-\delta\implies P(\hat{T}_{d,\varepsilon}(P,P'))>1-\varepsilon$.
Suppose that $\text{supp}(P)\neq\text{supp}(P')$ (otherwise the condition
is trivially satisfied). There is a countable subset $\Omega_{d}^{0}\subseteq\Omega^{0}$
which is dense under $d$. Suppose that for every $\varepsilon>0$
there is a finite grid $G_{d,\varepsilon}$ whose $\varepsilon$-neighborhood
covers a subset $H_{\varepsilon}\subseteq\hat{T}_{d,\varepsilon}(P,P')$
with $P(H_{\varepsilon})>1-\varepsilon$. Since $\Omega$ is Hausdorff
under the product topology, there is a $\delta$ so that $\min_{g,g'\in G_{d,\varepsilon}:g\neq g'}d_{\Pi}(g,g')>\delta$
and so the result follows. It remains to prove the existence of such
a grid. We show that for any nice metric $d$, $\Omega$ is Polish.
To establish that $\Omega$ is Polish it is enough to show that it
is complete under $d$. Let $(\omega^{k})_{k}$ be Cauchy. Since $d$
refines the product topology, the sequence $(\omega^{k})_{k}$ is
Cauchy under the product topology. Since $\Omega$ is Polish under
the product topology (see \cite{mertens2015repeated}), we deduce
that $\Omega$ is Polish under $d$. Then $P$ is regular and so every
set, in particular the set $\hat{T}_{d,\varepsilon}(P,P')$, admits
a compact approximation $H_{\varepsilon}$ from below. By compactness,
$H_{\varepsilon}$ admits such a grid. 
\end{proof}

\subsection{Interim Strategic Topology}

We follow \cite{chen2017characterizing} in describing the interim
strategic topology. 
\begin{defn}
(Frame) A frame is a profile of maps $(\pi_{i})_{i\in I}$ where for
every $i\in I$, $\pi_{i}\colon\mathcal{T}_{i}\to F_{i}$, $F_{i}$
is a finite set and for every prior $P\in\mathcal{P}$ and all types
$\tau_{i},\tau_{i}'\in\mathcal{T}_{i}$, 
\[
P(\cdot,\cdot|\tau_{i})\circ(id\times\pi_{-i})^{-1}=P(\cdot,\cdot|\tau_{i}')\circ(id\times\pi_{-i})^{-1}\implies\pi_{i}(\tau_{i})=\pi_{i}(\tau_{i}').
\]
\end{defn}
For any set $E\subseteq\Omega$ and any player $i$ , let $E_{i}\subseteq\mathcal{T}_{i}$
denote the projection of $E$ on $\mathcal{T}_{i}.$ Let $\mathscr{F}$
denote the collection of events that are measurable with respect to
a frame. 
\[
\mathscr{F}=\left\{ E\in\mathscr{B}:\exists\ \pi\in\Pi,\tau\in E\text{ s.t. }\forall\ i\in I,\ \pi_{i}^{-1}(\pi_{i}(\tau))=E_{i}\right\} .
\]
 The formal definition is recursive:
\begin{defn}
(Uniform Weak Distance on Frames) For all pairs $\omega=(\theta,\tau),\hat{\omega}=(\hat{\theta},\hat{\tau})\in\Omega$,
define the Uniform Prokohorov Distance on Frames by $d^{\mathscr{F}}(\omega,\hat{\omega}):=\sup_{m}d^{m}(\omega,\hat{\omega})$,
where the sequence of functions $(d^{m})_{m}$ on $\Omega\times\Omega$
is defined recursively from a metric representing the weak{*} topology
on profiles of first order beliefs $d^{1}(\omega,\hat{\omega})=d_{\Delta(\Theta)^{I}}(\tau^{1},\hat{\tau}^{1})$,
and for every $m>1$, 
\[
d^{m}(\omega,\hat{\omega}):=d_{\Theta}(\theta,\hat{\theta})+\inf\{\delta>0:\forall\ i,\ \tau_{i}(E)\leq\hat{\tau}_{i}(\mathcal{N}_{d^{m-1},\delta}(E))+\delta,\ \forall\ E\in\mathscr{F}\}.
\]
\end{defn}
\begin{defn}
(Interim Strategic Topology) Define the uniform weak topology on frames
as the topology on $\Omega$ generated by the sets $\left\{ \omega'\in\Omega:d^{\mathscr{F}}(\omega,\omega')<\varepsilon\right\} $
for $\omega\in\Omega$ and $\varepsilon>0$.
\end{defn}
\begin{prop}
\label{prop:(Niceness)-The-interim}(Niceness) The interim strategic
topology is nice.
\end{prop}
\begin{proof}
\cite{dekel2006topologies} show that a countable set of finite states
is dense in $\Omega$ in the interim strategic topology. The interim
strategic topology is a stronger topology than the product topology
and is metrizable (see \cite{dekel2006topologies}). 
\end{proof}
\begin{prop}
\label{prop:(Common-Belief-Invariance)}(Common Belief Invariance)
\label{prop_CPB_free} For every canonical $P,P'\in\mathcal{P}$ and
$\varepsilon>0$, 
\[
C_{1-\varepsilon}\left(\hat{T}_{d^{\mathscr{F}},\varepsilon}(P,P')\right)=\hat{T}_{d^{\mathscr{F}},\varepsilon}(P,P').
\]
\end{prop}
\begin{proof}
Consider $\omega=(\theta,\tau)\in\Omega_{P}$ and $\hat{\omega}=(\hat{\theta},\hat{\tau})\in\Omega_{P'}$
with $d^{\mathscr{F}}(\omega,\hat{\omega})<\varepsilon$. Then it
must be that $\omega,\hat{\omega}\in\hat{T}_{d^{\mathscr{F}},\varepsilon}(P,P')$.
Since $\Omega_{P}\in\mathscr{F}$, the fact that $d^{\mathscr{F}}(\omega,\hat{\omega})<\varepsilon$
implies that 
\[
\tau_{i}^{*}(\Omega_{P})\leq\hat{\tau}_{i}^{*}(\mathcal{N}_{d^{\mathscr{F}},\varepsilon}(\Omega_{P}))+\varepsilon.
\]
Since $\tau_{i}^{*}(\Omega_{P})=1$ and $\hat{\tau}_{i}^{*}(\Omega_{P'})=1$
we have 
\[
\hat{\tau}_{i}^{*}(\hat{T}_{d^{\mathscr{F}},\varepsilon}(P,P'))=\hat{\tau}_{i}^{*}(\mathcal{N}_{d^{\mathscr{F}},\varepsilon}(\Omega_{P}))\geq1-\varepsilon.
\]

A symmetric argument implies that $\tau_{i}^{*}(\hat{T}_{d^{\mathscr{F}},\varepsilon}(P,P'))\geq1-\varepsilon$
and so $\hat{T}_{d^{\mathscr{F}},\varepsilon}(P,P')=B^{1-\varepsilon}(\hat{T}_{d^{\mathscr{F}},\varepsilon}(P,P'))$,
which establishes the result. 
\end{proof}

\subsection{\label{subsec:Continuity-of-Exact}Continuity of Exact BIBCE in Rich
Games}

Our main result established that the ACK topology is the coarsest
topology generating continuity in approximate BIBCE outcomes. For
many applications of interest, we would like to make statements about
continuity of exact BIBCE outcomes. We can extend our results to exact
BIBCE if we allow for sufficiently rich games. 

We introduce a ``strong richness'' property of a base game. 
\begin{defn}
(strong richness) A base game $\mathcal{G}$ satisfies \emph{strong
richness} if there exists $\theta^{*}\in\varTheta$ such that for
every player $i$ and every $a_{i}\in A_{i},$ there exists $a_{-i}\in A_{-i}$
such that $(a_{i},a_{-i})$ is a strict Nash Equilibrium of $(u_{j}(\cdot,\cdot,\theta^{*}))_{j\in I}$. 
\end{defn}
Richness requires that the set of possible payoffs is sufficiently
rich. It is in the spirit of but stronger than richness properties
in the literature, e.g., \cite{weyi07b}. Recall that it is a maintained
assumption that every state $\theta\in\varTheta$ is assigned positive
probability. 
\begin{lem}
\label{prop:epsilonto0}For every base game $\mathcal{G}$ satisfying
strong richness, $\varepsilon>0$; and simple information structure
$P\in\mathcal{P}$, there exists a simple information structure $P^{+}\in\mathcal{P}$
so that $d^{ACK}(P,P^{+})<\varepsilon$ and 
\[
d_{\mathcal{H}}(\mathcal{O^{\varepsilon}}(\mathcal{G},P),\mathcal{O}(\mathcal{G},P^{+}))<2M|A\times\Theta|\varepsilon,
\]
where \textup{$d_{\mathcal{H}}(X,Y)$ is the Hausdorff distance between
$X,Y\subseteq\Delta(A\times\Theta)$.}
\end{lem}
\begin{proof}
By richness part (2), for any player $i$ and any action $a_{i}\in A_{i}$
there is $\alpha_{-i}(a_{i})\in A_{-i}$ so that $(a_{i},\alpha_{-i}(a_{i}))$
is a strict NE in the game with payoffs $(\overline{u}_{j}(\cdot,\cdot):=u_{j}(\cdot,\cdot,\theta^{*}))_{j\in I}$.
We first perturb $P$ so that every type assigns probability at least
$\delta>0$ to $\theta^{*}$ by adding a new profile of types $t^{\tau_{i}}=(t_{j}^{\tau_{i}})_{j\in I}$
for every player $i$ and every type $\tau_{i}\in T_{i}:=\left\{ \tau_{i}:\tau\in\text{supp}(P)\right\} $.
Define the extended prior $P^{+}$ with support given by
\[
T^{+}:=\text{supp}(P)\cup\left(\bigcup_{i\in I}\left(T^{i}\cup\bar{T}^{i}\right)\right),
\]
where $T^{i}:=\left\{ (\theta^{*},\tau_{i},t_{-i}^{\tau_{i}}):\tau\in\text{supp}(P)\right\} $
and $\bar{T}^{i}:=\left\{ (\theta^{*},t^{\tau_{i}}):\tau_{i}\in T_{i}\right\} $.
For every $(\theta,\tau)\in T^{+}$ and any given choices $\delta,\eta\in[0,1)$,
define 
\[
P^{+}(\theta,\hat{\tau}):=\begin{cases}
(1-\eta)(1-\delta)P(\theta,\hat{\tau}) & \text{ if }(\theta,\hat{\tau})\in\text{supp}(P),\\
(1-\eta)\delta/|I|\sum_{i\in I}\boldsymbol{1}_{\left(\theta,\hat{\tau}\right)\in T^{i}}P(\hat{\tau}_{i})\mu(\theta) & \text{ if }(\theta,\hat{\tau})\in\cup_{i\in I}T^{i},\\
\eta/|I|\sum_{i\in I}\boldsymbol{1}_{\left(\theta,\hat{\tau}\right)\in\bar{T}^{i}}\sum_{\tau_{i}}\boldsymbol{1}_{\hat{\tau}=t^{\tau_{i}}}P(\tau_{i})\mu(\theta) & \text{ if }(\theta,\hat{\tau})\in\cup_{i\in I}\bar{T}^{i}.
\end{cases}
\]
For every $\sigma\in\mathcal{B^{\varepsilon}}(\mathcal{G},P)$ we
construct an associated BIBCE $\sigma^{+}\in\mathcal{B}(\mathcal{G},P^{+})$
as follows: 
\[
\sigma^{+}(a|\theta,\hat{\tau})=\begin{cases}
\sum_{i\in I}\boldsymbol{1}_{\left(\theta,\hat{\tau}\right)\in T^{i}\cup\overline{T}^{i}}\boldsymbol{1}_{a=\left(a_{i},\alpha_{-i}(a_{i})\right)}\sigma(a_{i}|\hat{\tau}_{i}) & \text{ if }(\theta,\hat{\tau})\in\cup_{i\in I}\left(T^{i}\cup\bar{T}^{i}\right),\\
\sigma(a|\theta,\tau) & \text{ otherwise.}
\end{cases}
\]
We first verify that $\sigma^{+}$is belief invariant: 
\[
\sum_{a_{-i}}\sigma^{+}(a|\theta,\hat{\tau}_{-i},\hat{\tau}_{i})=\begin{cases}
\sigma(a_{i}|\hat{\tau}_{i}) & \text{if }(\theta,\hat{\tau})\in\text{supp}(P),\\
\sigma(a_{i}|\hat{\tau}_{i}) & \text{if }(\theta,\hat{\tau})\in T^{i}\cup\bar{T}^{i},\\
\sum_{j}\boldsymbol{1}_{\left(\theta,\hat{\tau}\right)\in T^{j}\cup\bar{T}^{j}}\sigma(a_{i}|\hat{\tau}_{j}) & \text{if }(\theta,\hat{\tau})\in\cup_{j\neq i}T^{j}\cup\bar{T}^{j},
\end{cases}
\]
and so $\left(\theta,\hat{\tau}\right)\mapsto\sigma(a_{i}|\theta,\hat{\tau})$
is measurable with respect to the projection $(\theta,\hat{\tau})\mapsto\hat{\tau}_{i}$
for all $a_{i}\in A_{i}$. This decision rule increases the expected
payoff of $\sigma$ for every type in the support of $P$: Let $J_{\mathcal{G}}:=\min_{a_{i},\hat{a}_{i}\in A_{i},i\in I}\left(\bar{u}_{i}(a_{i},\alpha_{-i}(a_{i}))-\bar{u}_{i}(\hat{a}_{i},\alpha_{-i}(a_{i}))\right)>0$
and define $\delta:=\frac{\varepsilon}{J_{\mathcal{G}}+\varepsilon}<\varepsilon$.
Consider any $\tau_{i}\in T_{i}$ and deviation $a_{i}'$

\[
\begin{aligned}\sum_{\theta,\tau}\sum_{a}\Delta u_{i}(a,a_{i}',\theta)P^{+}\circ\sigma^{+}(a,\theta,\tau_{-i}|\tau_{i}) & >-\varepsilon(1-\delta)+\delta\sum_{a_{i}}\Delta\bar{u}_{i}(a_{i},\alpha_{-i}(a_{i}),a_{i}')\ \sigma(a_{i}|\tau_{i})\\
 & \geq-\varepsilon(1-\delta)+J_{\mathcal{G}}\delta=0.
\end{aligned}
\]
It remains to check if obedience also holds for each type $t_{i}^{\tau_{j}}$
of player $i$ and any associated $\tau_{j}\in T_{j}$ of player $j$:
\[
\begin{aligned}\sum_{\theta,\tau}\sum_{a}\Delta u_{i}(a,a_{i}',\theta)\ P^{+}\circ\sigma^{+}(a,\theta,\tau_{-i}|t_{i}^{\tau_{j}}) & =\sum_{a_{j}}\Delta\bar{u}_{i}(a_{j},\alpha_{-j}(a_{j}),a_{i}')\ \sigma(a_{j}|\tau_{j})\\
 & \geq J_{\mathcal{G}}>0.
\end{aligned}
\]
Hence $\sigma^{+}\in\mathcal{B}(\mathcal{G},P^{+})$ and outcomes
$\nu_{\sigma}\in\mathcal{O^{\varepsilon}}(\mathcal{G},P)$ and $\nu_{\sigma^{+}}\in\mathcal{O}(\mathcal{G},P^{+})$
satisfy $||\nu_{\sigma}-\nu'||_{2}\leq2M\frac{\varepsilon}{J_{\mathcal{G}}+\varepsilon}<2M\varepsilon$
. Let $g_{i}(\tau):=\overline{\tau}_{i}\left(P^{+}(\cdot,\cdot|\tau_{i})\right)$
denote the new canonical type/belief hierarchy of player $i$, for
each $\tau\in\text{supp}(P)$. Then for every event $E\subseteq\text{supp}(P)$
and any player $i$, $|g_{i}(\tau)(E)-\tau_{i}(E)|\leq\delta.$ Since
the total variation norm is stronger than the product topology, we
conclude that $d_{\Pi}(g(\tau),\tau)\leq\delta<\varepsilon$. Hence
$P^{+}\left(\mathcal{N}_{\varepsilon}(\text{supp}(P))\right)>1-\varepsilon$
and so $d^{ACK}(P,P^{+})<\varepsilon.$ Hence for every $\nu'\in\mathcal{O}(\mathcal{G},P^{+})$
there is $\nu\in\mathcal{O^{\varepsilon}}(\mathcal{G},P)$ so that
$||\nu-\nu'||_{2}\leq2M|A\times\Theta|\varepsilon$. 
\end{proof}
\begin{prop}
\label{prop:Properconv}For every base game $\mathcal{G}$ and any
information structure $P\in\mathcal{P}$ there is a sequence $(\hat{P}^{k})_{k}$
of simple information structures ACK-converging to $P$ and a subset
$\mathcal{O}^{\infty}(\mathcal{G},P)\subseteq\mathcal{O}(\mathcal{G},P)$
so that $\lim_{k\uparrow\infty}d_{\mathcal{H}}(\mathcal{O}(\mathcal{G},\hat{P}^{k}),\mathcal{O}^{\infty}(\mathcal{G},P))=0.$
If $\mathcal{G}$ satisfies strong richness, then 
\[
\lim_{k\uparrow\infty}d_{\mathcal{H}}(\mathcal{O}(\mathcal{G},\hat{P}^{k}),\mathcal{O}(\mathcal{G},P))=0.
\]
\end{prop}
\begin{proof}
By Proposition \ref{prop: denseness } there is a sequence of simple
priors $P^{k}$ converging strategically to $P$ and a sequence of
positive $(\varepsilon^{k})_{k}$ converging monotonically to $0$
so that for every $k$, $d^{*}(P^{k},P)<\varepsilon^{k}.$ We first
prove upper hemi-continuity of BIBCE with respect to strategic convergence.
Let $(\sigma^{k})_{k}$ be a sequence of BIBCE so that for every $k\in\mathbb{N}$,
$\sigma^{k}\in\mathcal{B}^{0}(\mathcal{G},P^{k})$ and the induced
sequence of outcomes $(\nu_{\sigma^{k}})_{k}$ converges to some outcome
$\nu_{\infty}$, i.e. $\lim_{k\uparrow\infty}||\nu_{\sigma^{k}}-\nu_{\infty}||_{2}=0$.
We have shown that for every $k$ there is a $\varepsilon^{k}$-BIBCE
$\sigma^{\varepsilon^{k}}\in\mathcal{B}^{\varepsilon^{k}}(\mathcal{G},P)$
so that $||\nu_{\sigma^{k}}-\nu_{\sigma^{\varepsilon^{k}}}||_{2}\leq4M|A\times\Theta|\varepsilon^{k}$.
Our construction in the proof of Proposition \ref{prop: sufficiency}
has the feature that the sequence $\left(\sigma^{\varepsilon^{k}}\right)_{k}$
converges almost surely to an obedient decision rule (as it was constructed
from a conditional expectation) and so there is a BIBCE $\sigma^{\infty}\in\mathcal{B}(\mathcal{G},P)$
so that $\nu_{\sigma^{\infty}}=\nu_{\infty}$, which establishes upper
hemi-continuity: There is a subset $\mathcal{O}^{\infty}(\mathcal{G},P)\subseteq\mathcal{O}(\mathcal{G},P)$
so that 
\[
\lim_{k\uparrow\infty}d_{\mathcal{H},\mathcal{G}}(\mathcal{O}(\mathcal{G},P^{k}),\mathcal{O}^{\infty}(\mathcal{G},P))=0.
\]
Lower hemi-continuity is a consequence of Lemma \ref{prop:epsilonto0}
and richness part (2) of $\mathcal{G}$. Indeed, for every BIBCE $\sigma\in\mathcal{B}(\mathcal{G},P)$
and every $k\in\mathbb{N}$, $d^{*}(P,P^{k})<\varepsilon^{k}$ implies
that there is $\sigma^{\varepsilon^{k},k}\in\mathcal{B}^{\varepsilon^{k}}(\mathcal{G},P^{k})$
so that $||\nu_{\sigma^{\varepsilon^{k},k}}-\nu_{\sigma}||_{2}<\varepsilon^{k}$.
Hence, there is a superset $\overline{\mathcal{O}}^{\infty}(\mathcal{G},P)\supseteq\mathcal{O}(\mathcal{G},P)$
so that $\lim_{k\uparrow\infty}d_{\mathcal{H},\mathcal{G}}(\mathcal{O}^{\varepsilon^{k}}(\mathcal{G},P^{k}),\overline{\mathcal{O}}^{\infty}(\mathcal{G},P))=0.$
From the construction in the proof of Lemma \ref{prop:epsilonto0}
there is simple $\hat{P}^{k}\in\mathcal{P}$ so that $d^{*}(P^{k},\hat{P}^{k})<\varepsilon^{k}$
so that $d_{\mathcal{H},\mathcal{G}}(\mathcal{O}^{\varepsilon^{k}}(\mathcal{G},P^{k}),\mathcal{O}(\mathcal{G},\hat{P}^{k}))<2M(u)\varepsilon^{k}$.
By the triangle inequality, the sequence $(\hat{P}^{k})_{k}$ also
converges strategically to $P$. By upper hemi-continuity established
earlier and the triangle inequality there is a subset $\underline{\mathcal{O}}^{\infty}(\mathcal{G},P)\subseteq\mathcal{O}(\mathcal{G},P)$
so that $\lim_{k\uparrow\infty}d_{\mathcal{H},\mathcal{G}}(\mathcal{O}^{\varepsilon^{k}}(\mathcal{G},P^{k}),\underline{\mathcal{O}}^{\infty}(\mathcal{G},P))=0.$
But then we have that $\overline{\mathcal{O}}^{\infty}(\mathcal{G},P)=\underline{\mathcal{O}}^{\infty}(\mathcal{G},P)=\mathcal{O}(\mathcal{G},P)$
and so the result follows.
\end{proof}
Focusing attention on BIBCE as a solution concept, we have: 
\begin{lem}
\label{cor:information design}For any game $\mathcal{G}$ and any
open $\mathcal{P}^{*}\subseteq\mathcal{P}$,
\[
\sup_{P\in\mathcal{P}^{*}\cap\mathcal{P}^{SIMPLE}}V(\mathcal{O}(\mathcal{G},P))\leq\sup_{P\in\mathcal{P}^{*}}V(\mathcal{O}(\mathcal{G},P)).
\]
Moreover, if $\mathcal{G}$ satisfies strong richness, 
\[
\sup_{P\in\mathcal{P}^{*}\cap\mathcal{P}^{SIMPLE}}V(\mathcal{O}(\mathcal{G},P))=\sup_{P\in\mathcal{P}^{*}}V(\mathcal{O}(\mathcal{G},P)).
\]
\end{lem}
The first statement follows from our denseness result (\ref{prop: denseness })
while the second statement follows from our exact continuity results
(Propositions \ref{prop:epsilonto0}and \ref{prop:Properconv})

\subsection{Example Existence of BIBCE \label{Asubsec:Example-Existence-of}}

We consider the game by \cite{hellman14no} where BNE fails to exist
and yet construct a straightforward BIBCE of this game. There are
two players $A$ and $B$ and two actions $L,R$ (see \cite{hellman14no}
for the payoff matrix)

\[
\sigma(y,\theta)=\begin{cases}
\frac{1}{2}(L,L),\frac{1}{2}(R,R) & \text{, if }\theta\in\left\{ (A,1),(B,1)\right\} \\
\frac{1}{2}(R,L),\frac{1}{2}(L,R) & \text{, if }\theta\in\left\{ (A,-1),(B,-1)\right\} 
\end{cases}
\]

This decision rule is independent of the type space and satisfies
belief invariance: For any $a_{i}\in\left\{ L,R\right\} $ and any
$\theta\in\left\{ (A,1),(B,1),(A,-1),(B,-1)\right\} $
\[
\sum_{a_{-i}\in\left\{ L,R\right\} }\sigma(a_{i},a_{-i}|y,\theta)=\sigma(a_{i},L|y,\theta)+\sigma(a_{i},R|y,\theta)=\frac{1}{2}.
\]

We do not need to report the full information structure considered
in \cite{hellman14no}, to verify obedience as $\sigma$ is independent
of types. Each player $i\in\left\{ A,B\right\} $ has one of two first-order
beliefs $P_{i}(\cdot|y)$ indexed by $y\in\left\{ -1,1\right\} $
\[
P_{A}(\theta|y)=\begin{cases}
\frac{1}{2} & \text{ if }\theta=(A,y)\\
\frac{1}{4} & \text{ if }\theta=(B,1)\\
\frac{1}{4} & \text{ if }\theta=(B,-1)
\end{cases},\ P_{B}(\theta|y)=\begin{cases}
\frac{1}{2} & \text{ if }\theta=(B,y)\\
\frac{1}{4} & \text{ if }\theta=(A,1)\\
\frac{1}{4} & \text{ if }\theta=(A,-1)
\end{cases}.
\]

So if recommended $a_{i}\in\left\{ L,R\right\} $ player $i$'s payoff
increment from playing a different action $a_{i}'\neq a_{i}$ is given
by when $i=A:$ 

For the case where $y=1$ and $a_{A}=L$
\begin{eqnarray*}
\frac{1}{2}\left(u_{A}(L,L,\left(A,1\right))-u_{A}(R,L,\left(A,1\right))\right)+\frac{1}{4}\left(u_{A}(L,L,\left(B,1\right))-u_{A}(R,L,\left(B,1\right))\right)\\
+\frac{1}{4}\left(u_{A}(L,R,(B,-1))-u_{A}(R,R,(B,-1))\right)=0.35>0.
\end{eqnarray*}

For the case where $y=1$ and $a_{A}=R$

\begin{eqnarray*}
\frac{1}{2}\left(u_{A}(R,R,\left(A,1\right))-u_{A}(L,R,\left(A,1\right))\right)+\frac{1}{4}\left(u_{A}(R,R,\left(B,1\right))-u_{A}(L,R,\left(B,1\right))\right)\\
+\frac{1}{4}\left(u_{A}(R,L,(B,-1))-u_{A}(L,L,(B,-1))\right)=0.15>0.
\end{eqnarray*}

For the case $y=-1$ and $a_{A}=L$,

\begin{eqnarray*}
\frac{1}{2}\left(u_{A}(L,R,\left(A,-1\right))-u_{A}(R,R,\left(A,-1\right))\right)+\frac{1}{4}\left(u_{A}(L,R,\left(B,-1\right))-u_{A}(R,R,\left(B,-1\right))\right)\\
+\frac{1}{4}\left(u_{A}(L,L,(B,1))-u_{A}(R,L,(B,1))\right)=0.35>0.
\end{eqnarray*}

For the case $y=-1$ and $a_{A}=R$,
\begin{eqnarray*}
\frac{1}{2}\left(u_{A}(R,L,\left(A,-1\right))-u_{A}(L,L,\left(A,-1\right))\right)+\frac{1}{4}\left(u_{A}(R,L,\left(B,-1\right))-u_{A}(L,L,\left(B,-1\right))\right)\\
+\frac{1}{4}\left(u_{A}(R,R,(B,1))-u_{A}(L,R,(B,1))\right)=0.15>0.
\end{eqnarray*}

Similarly for player $i=B$:

For the case where $y=1$ and $a_{B}=L$
\begin{eqnarray*}
\frac{1}{2}\left(u_{B}(L,L,\left(B,1\right))-u_{B}(L,R,\left(B,1\right))\right)+\frac{1}{4}\left(u_{B}(L,L,\left(A,1\right))-u_{B}(L,R,\left(A,1\right))\right)\\
+\frac{1}{4}\left(u_{B}(R,L,(A,-1))-u_{B}(R,R,(A,-1))\right)=0.35>0.
\end{eqnarray*}

For the case where $y=1$ and $a_{B}=R$

\begin{eqnarray*}
\frac{1}{2}\left(u_{B}(R,R,\left(B,1\right))-u_{B}(R,L,\left(B,1\right))\right)+\frac{1}{4}\left(u_{B}(R,R,\left(A,1\right))-u_{B}(R,L,\left(A,1\right))\right)\\
+\frac{1}{4}\left(u_{B}(L,R,(A,-1))-u_{B}(L,L,(A,-1))\right)=0.15>0.
\end{eqnarray*}

For the case $y=-1$ and $a_{B}=L$,

\begin{eqnarray*}
\frac{1}{2}\left(u_{B}(R,L,\left(B,-1\right))-u_{B}(R,R,\left(B,-1\right))\right)+\frac{1}{4}\left(u_{B}(R,L,\left(A,-1\right))-u_{B}(R,R,\left(A,-1\right))\right)\\
+\frac{1}{4}\left(u_{B}(L,L,(A,1))-u_{B}(L,R,(A,1))\right)=0.35>0.
\end{eqnarray*}

For the case $y=-1$ and $a_{B}=R$,
\begin{eqnarray*}
\frac{1}{2}\left(u_{B}(L,R,\left(B,-1\right))-u_{B}(L,L,\left(B,-1\right))\right)+\frac{1}{4}\left(u_{B}(L,R,\left(A,-1\right))-u_{B}(L,L,\left(A,-1\right))\right)\\
+\frac{1}{4}\left(u_{B}(R,R,(B,1))-u_{A}(L,R,(B,1))\right)=0.15>0,
\end{eqnarray*}

and thus all the obedience conditions for a BIBCE are satisfied.\newpage\bibliographystyle{ecta}
\bibliography{MyBibFile,general}

\begin{thebibliography}{41}
\newcommand{\enquote}[1]{``#1''}
\expandafter\ifx\csname natexlab\endcsname\relax\def\natexlab#1{#1}\fi

\bibitem[\protect\citeauthoryear{Balder}{Balder}{1988}]{balder88generalized}
\textsc{Balder, E.} (1988): \enquote{Generalized Equilibrium Results for Games
  with Incomplete Information,} \emph{Mathematics of Operations Research},
  1988, 265--276.

\bibitem[\protect\citeauthoryear{Bergemann and Morris}{Bergemann and
  Morris}{2016}]{bemo16}
\textsc{Bergemann, D. and S.~Morris} (2016): \enquote{Bayes Correlated
  Equilibrium and the Comparison of Information Structures in Games,}
  \emph{Theoretical Economics}, 11, 487--522.

\bibitem[\protect\citeauthoryear{Bergemann and Morris}{Bergemann and
  Morris}{2017}]{bergemann2017belief}
---\hspace{-.1pt}---\hspace{-.1pt}--- (2017): \enquote{Belief-Free
  Rationalizability and Informational Robustness,} \emph{Games and Economic
  Behavior}, 104, 744--759.

\bibitem[\protect\citeauthoryear{Brandenburger and Friedenberg}{Brandenburger
  and Friedenberg}{2008}]{brandenburger08intrinsic}
\textsc{Brandenburger, A. and A.~Friedenberg} (2008): \enquote{Intrinsic
  Correlation in Games,} \emph{Journal of Economic Theory}, 141, 28--67.

\bibitem[\protect\citeauthoryear{Carlsson and Van~Damme}{Carlsson and
  Van~Damme}{1993}]{cava93}
\textsc{Carlsson, H. and E.~Van~Damme} (1993): \enquote{Global games and
  equilibrium selection,} \emph{Econometrica: Journal of the Econometric
  Society}, 989--1018.

\bibitem[\protect\citeauthoryear{Chen, Di~Tillio, Faingold, and Xiong}{Chen
  et~al.}{2017}]{chen2017characterizing}
\textsc{Chen, Y.-C., A.~Di~Tillio, E.~Faingold, and S.~Xiong} (2017):
  \enquote{Characterizing the Strategic Impact of Misspecified Beliefs,}
  \emph{Review of Economic Studies}, 84, 1424--1471.

\bibitem[\protect\citeauthoryear{Dekel, Fudenberg, and Morris}{Dekel
  et~al.}{2006}]{dekel2006topologies}
\textsc{Dekel, E., D.~Fudenberg, and S.~Morris} (2006): \enquote{Topologies on
  Types,} \emph{Theoretical Economics}, 1, 275--309.

\bibitem[\protect\citeauthoryear{Dekel, Fudenberg, and Morris}{Dekel
  et~al.}{2007}]{dekel07icr}
---\hspace{-.1pt}---\hspace{-.1pt}--- (2007): \enquote{Interim Correlated
  Rationalizability,} \emph{Theoretical Economics}.

\bibitem[\protect\citeauthoryear{Du}{Du}{2012}]{du12correlated}
\textsc{Du, S.} (2012): \enquote{Correlated Equilibrium and Higher-Order
  Beliefs about Play,} \emph{Games and Economic Behavior}, 76, 74--87.

\bibitem[\protect\citeauthoryear{Ely and Peski}{Ely and Peski}{2011}]{elpe11}
\textsc{Ely, J. and M.~Peski} (2011): \enquote{Critical Types,} \emph{Review of
  Economic Studies}, 78, 907--937.

\bibitem[\protect\citeauthoryear{Engl}{Engl}{1995}]{engl95lower}
\textsc{Engl, G.} (1995): \enquote{Lower Hemicontinuity of the Nash Equilibrium
  Correspondence,} \emph{Games and Economic Behavior}, 9, 151--160.

\bibitem[\protect\citeauthoryear{Forges}{Forges}{1993}]{forges93five}
\textsc{Forges, F.} (1993): \enquote{Five Legitimate Definitions of Correlated
  Equilibrium in Games with Incomplete Information,} \emph{Theory and
  Decision}, 61, 329--344.

\bibitem[\protect\citeauthoryear{Forges}{Forges}{2006}]{forges06revisited}
---\hspace{-.1pt}---\hspace{-.1pt}--- (2006): \enquote{Correlated Equilibrium
  in Games with Incomplete Information Re-Visited,} \emph{Theory and Decision},
  61, 329--324.

\bibitem[\protect\citeauthoryear{Gensbittel, Peski, and Renault}{Gensbittel
  et~al.}{2022}]{gensbittel22value}
\textsc{Gensbittel, F., M.~Peski, and J.~Renault} (2022): \enquote{Value-Based
  Distance between Information Structures,} \emph{Theoretical Economics}.

\bibitem[\protect\citeauthoryear{Gossner}{Gossner}{2000}]{gossner00comparison}
\textsc{Gossner, O.} (2000): \enquote{Comparison of Information Structures,}
  \emph{Games and Economic Behavior}, 30, 44--63.

\bibitem[\protect\citeauthoryear{Gossner and Mertens}{Gossner and
  Mertens}{2020}]{gossner20value}
\textsc{Gossner, O. and J.-F. Mertens} (2020): \enquote{The Value of
  Information in Zero-Sum Games,} Tech. rep., LSE.

\bibitem[\protect\citeauthoryear{Harsanyi}{Harsanyi}{1967-68}]{hars67}
\textsc{Harsanyi, J.} (1967-68): \enquote{Games with Incomplete Information
  Played by 'Bayesian' Players,} \emph{Management Science}, 14, 159--189,
  320--334, 485--502.

\bibitem[\protect\citeauthoryear{Hellman}{Hellman}{2014}]{hellman14no}
\textsc{Hellman, Z.} (2014): \enquote{A Game with no Bayesian Approximate
  Equilibria,} \emph{Journal of Economic Theory}, 153, 138--151.

\bibitem[\protect\citeauthoryear{Inostroza and Pavan}{Inostroza and
  Pavan}{2023}]{inostroza23adversarial}
\textsc{Inostroza, N. and A.~Pavan} (2023): \enquote{Adversarial Coordination
  and Public Information Design,} Tech. rep., University of Toronto and
  Northwestern University.

\bibitem[\protect\citeauthoryear{Kajii and Morris}{Kajii and
  Morris}{1998}]{kajii1998payoff}
\textsc{Kajii, A. and S.~Morris} (1998): \enquote{Payoff Continuity in
  Incomplete Information Games,} \emph{Journal of Economic Theory}, 82,
  267--276.

\bibitem[\protect\citeauthoryear{Kambhampati}{Kambhampati}{2023}]{kambhampati23payoffcontinuity}
\textsc{Kambhampati, A.} (2023): \enquote{Payoff Continuity in Games of
  Incomplete Information Across Models of Knowledge,} Tech. rep., United States
  Naval Academy.

\bibitem[\protect\citeauthoryear{Lehrer, Rosenberg, and Shmaya}{Lehrer
  et~al.}{2010}]{lehrer10signaling}
\textsc{Lehrer, E., D.~Rosenberg, and E.~Shmaya} (2010): \enquote{Signaling and
  Mediation in Games with Common Interests,} \emph{Games and Economic
  Behavior}, 68, 670--682.

\bibitem[\protect\citeauthoryear{Li, Song, and Zhao}{Li
  et~al.}{2023}]{li23global}
\textsc{Li, F., Y.~Song, and M.~Zhao} (2023): \enquote{Global Manipulation by
  Local Obfuscation,} \emph{Journal of Economic Theory}.

\bibitem[\protect\citeauthoryear{Lipman}{Lipman}{2003}]{lipman2003finite}
\textsc{Lipman, B.~L.} (2003): \enquote{Finite order implications of common
  priors,} \emph{Econometrica}, 71, 1255--1267.

\bibitem[\protect\citeauthoryear{Liu}{Liu}{2009}]{liu09redundant}
\textsc{Liu, Q.} (2009): \enquote{On Redundant Types and the Bayesian
  Formulation of Incomplete Information,} \emph{Journal of Economic Theory},
  144, 2115--2145.

\bibitem[\protect\citeauthoryear{Liu}{Liu}{2015}]{liu15}
---\hspace{-.1pt}---\hspace{-.1pt}--- (2015): \enquote{Correlation and Common
  Priors in Games with Incomplete Information,} \emph{Journal of Economic
  Theory}, 157, 49--75.

\bibitem[\protect\citeauthoryear{Mathevet, Perego, and Taneva}{Mathevet
  et~al.}{2020}]{mathevet20information}
\textsc{Mathevet, L., J.~Perego, and I.~Taneva} (2020): \enquote{On Information
  Design in Games,} \emph{Journal of Political Economy}, 128, 1370--1404.

\bibitem[\protect\citeauthoryear{Mertens, Sorin, and Zamir}{Mertens
  et~al.}{2015}]{mertens2015repeated}
\textsc{Mertens, J.-F., S.~Sorin, and S.~Zamir} (2015): \emph{Repeated games},
  vol.~55, Cambridge University Press.

\bibitem[\protect\citeauthoryear{Mertens and Zamir}{Mertens and
  Zamir}{1985}]{mertens1985formulation}
\textsc{Mertens, J.-F. and S.~Zamir} (1985): \enquote{Formulation of Bayesian
  analysis for games with incomplete information,} \emph{International journal
  of game theory}, 14, 1--29.

\bibitem[\protect\citeauthoryear{Milgrom and Weber}{Milgrom and
  Weber}{1985}]{milgrom2015distributional}
\textsc{Milgrom, P. and R.~Weber} (1985): \enquote{Distributional Strategies
  for Games with Incomplete Information,} \emph{Mathematics of Operations
  Research}, 10, 619--632.

\bibitem[\protect\citeauthoryear{Monderer and Samet}{Monderer and
  Samet}{1989}]{monderer1989approximating}
\textsc{Monderer, D. and D.~Samet} (1989): \enquote{Approximating Common
  Knowledge with Common Beliefs,} \emph{Games and Economic Behavior}, 1,
  170--190.

\bibitem[\protect\citeauthoryear{Monderer and Samet}{Monderer and
  Samet}{1996}]{monderer1996proximity}
---\hspace{-.1pt}---\hspace{-.1pt}--- (1996): \enquote{Proximity of Information
  in Games with Incomplete Information,} \emph{Mathematics of Operations
  Research}, 21, 707--725.

\bibitem[\protect\citeauthoryear{Morris, Oyama, and Takahashi}{Morris
  et~al.}{2024}]{moot24}
\textsc{Morris, S., D.~Oyama, and S.~Takahashi} (2024): \enquote{Implementation
  via Information Design in Binary-Action Supermodular Games,}
  \emph{Econometrica}.

\bibitem[\protect\citeauthoryear{Peski}{Peski}{2008}]{peski08comparison}
\textsc{Peski, M.} (2008): \enquote{Comparison of Information Structures in
  Zero-Sum Games,} \emph{Games and Economic Behavior}, 62, 732--735.

\bibitem[\protect\citeauthoryear{Rubinstein}{Rubinstein}{1989}]{rubinstein89email}
\textsc{Rubinstein, A.} (1989): \enquote{The Electronic Mail Game: Strategic
  Behavior under 'Almost Common Knowledge',} \emph{American Economic Review}.

\bibitem[\protect\citeauthoryear{Sadzik}{Sadzik}{2019}]{sadzik19revealed}
\textsc{Sadzik, T.} (2019): \enquote{Beliefs Revealed in Bayesian-Nash
  Equilibrium,} Tech. rep., UCLA.

\bibitem[\protect\citeauthoryear{Simon}{Simon}{2003}]{simon03ergodic}
\textsc{Simon, R.~S.} (2003): \enquote{Games of Incomplete Information, Ergodic
  Theory, and the Measurability of Equilibria,} \emph{Israel Journal of
  Mathematics}, 138, 73--92.

\bibitem[\protect\citeauthoryear{Simon and Tomkowicz}{Simon and
  Tomkowicz}{2017}]{simon17without}
\textsc{Simon, R.~S. and G.~Tomkowicz} (2017): \enquote{A Bayesian Game without
  e-Equilibria,} Tech. rep., LSE.

\bibitem[\protect\citeauthoryear{Stinchcombe}{Stinchcombe}{2011}]{stinchcombe2011correlated}
\textsc{Stinchcombe, M.~B.} (2011): \enquote{Correlated Equilibrium Existence
  for Infinite Games with Type-Dependent Strategies,} \emph{Journal of Economic
  Theory}, 146, 638--655.

\bibitem[\protect\citeauthoryear{van Zandt}{van
  Zandt}{2010}]{vanzandt10supermodular}
\textsc{van Zandt, T.} (2010): \enquote{Interim Bayesian Nash Equilibrium on
  Universal Type Spaces for Supermodular Games,} \emph{Journal of Economic
  Theory}, 145, 249--263.

\bibitem[\protect\citeauthoryear{Weinstein and Yildiz}{Weinstein and
  Yildiz}{2007}]{weyi07b}
\textsc{Weinstein, J. and M.~Yildiz} (2007): \enquote{A Structure Theorem for
  Rationalizability with Applications to Robust Predictions of Refinements,}
  \emph{Econometrica}, 75, 365--400.

\end{thebibliography}

\end{document}